\documentclass[opre]{informs4Preprint}

\newif\ifinformsclass
\makeatletter
\@ifclassloaded{informs4}{\informsclasstrue}{\informsclassfalse}
\makeatother

\RequirePackage{newtxmath}

\EquationsNumberedThrough
\TheoremsNumberedThrough
\ECRepeatTheorems

\usepackage{natbib}
\bibpunct[, ]{(}{)}{,}{a}{}{,}\def\bibfont{\small}

\ifinformsclass
  \usepackage{mathtools}
\else
\fi
\ifinformsclass\else
\usepackage{dsfont}
\fi
\ifinformsclass\else
\usepackage{graphicx}
\fi
\usepackage{changepage}
\usepackage{enumitem}
\usepackage{scalerel}
\usepackage{accents}
\ifinformsclass\else
\usepackage{layout} \usepackage{printlen}
\fi
\usepackage{booktabs}
\usepackage{float}
\usepackage{xspace}
\ifinformsclass\else
\usepackage{mathrsfs}
\fi
\usepackage{amssymb}

\usepackage{threeparttable}

\ifinformsclass\else
\usepackage{optidef}
\fi

\ifinformsclass\else
\usepackage{nicefrac,xfrac}
\fi

\ifinformsclass\else
\usepackage[toc,page]{appendix}
\fi

\usepackage{makecell}
\usepackage{tablefootnote}
\usepackage{lipsum}
\ifinformsclass\else
\usepackage{caption}
\usepackage{subcaption}
\fi

\usepackage{xcolor}

\definecolor{cornellred}{rgb}{0.7, 0.11, 0.11}
\definecolor{maroon}{rgb}{0.52, 0, 0}
\definecolor{dgreen}{rgb}{0.0, 0.5, 0.0}
\definecolor{ballblue}{rgb}{0.13, 0.67, 0.8}
\definecolor{royalblue(web)}{rgb}{0.25, 0.41, 0.88}
\definecolor{bleudefrance}{rgb}{0.19, 0.55, 0.91}
\definecolor{royalazure}{rgb}{0.0, 0.22, 0.66}
\usepackage{hyperref}
\hypersetup{
	colorlinks = true,
	linkcolor=royalazure,
	urlcolor=royalazure,
	citecolor=royalazure,
	linkbordercolor = {white}
}

\ifinformsclass\else
\usepackage{mfirstuc}
\fi
\usepackage{multirow}
\usepackage{array}
\ifinformsclass\else
\usepackage{complexity}
\fi
\usepackage[suppress]{color-edits}  \addauthor[Yiding]{yf}{purple}    \addauthor[Mengfan]{mf}{red}
\addauthor[Zongqi]{zq}{blue}
\addauthor[Bo]{bo}{teal}
\usepackage{csquotes}
\usepackage{comment}

\ifinformsclass

\else
\usepackage{bbm}
\fi

\usepackage[breakable, skins]{tcolorbox}
\definecolor{myLightGray}{gray}{0.9} 

\DeclareRobustCommand{\myboxtwo}[2][gray!15]{
\begin{tcolorbox}[ 
        colback=white,      colframe=gray,  
        boxrule=0.2pt,      arc=2pt,outer arc=2pt,
        left=12pt,
        right=12pt,
        top=5pt,
        bottom=5pt,
        width=1.07\linewidth,
        enlarge left by=-0.55cm,
        before upper=\renewcommand{\baselinestretch}{1.3}\selectfont,
        after upper=\normalfont
        ]
 #2
 \end{tcolorbox}
}

\DeclareRobustCommand{\mybox}[2][myLightGray]{\begin{tcolorbox}[
        left=0.5pt,
        right=0.5pt,
        top=0.5pt,
        bottom=0.5pt,
        colback=#1,
        colframe=#1,
        width=\dimexpr\textwidth\relax, 
        boxsep=3pt,
        arc=2pt,outer arc=2pt,
        before upper={\setlength{\baselineskip}{1.3\baselineskip}}
        ]
        #2    
\end{tcolorbox}
}

\usepackage[noabbrev]{cleveref}
\crefname{claim}{claim}{claims}
\crefname{equation}{eq.}{eqs.}
\Crefname{equation}{eq.}{eqs.}
\crefname{enumi}{property}{properties}
\crefname{algocf}{algorithm}{algorithms}
\Crefname{algocf}{Algorithm}{Algorithms}
\crefname{adversary}{adversary}{adversaries}
\Crefname{adversary}{Adversary}{Adversaries}

\usepackage{tikz}
\ifinformsclass\else
\usepackage{pgfplots}
\pgfplotsset{width=10cm,compat=1.9}
\fi
\definecolor{darkpastelgreen}{rgb}{0.01, 0.75, 0.24}
\definecolor{darkorchid}{rgb}{0.6, 0.2, 0.8}
\definecolor{blue(pigment)}{rgb}{0.2, 0.2, 0.6}
\ifinformsclass\else
\usepackage{algpseudocode}
\fi
\usepackage[ruled,vlined,linesnumbered,noend]{algorithm2e}
\SetKw{KwBreak}{break}
\SetKwInOut{KwAux}{Notation}
\SetKw{KwHalt}{halt}

\SetCommentSty{mycommfont}
\Crefname{algorithm}{Algorithm}{Algorithms}
\makeatletter
\newcommand{\algolinelabel}[1]{\begingroup
    \def\@currentcounter{AlgoLine}\cref@constructprefix{AlgoLine}{\cref@result}\protected@edef\cref@currentlabel{[AlgoLine][\arabic{AlgoLine}][\cref@result]\theAlgoLine}\label{#1}\endgroup}
\makeatother
\crefname{AlgoLine}{line}{lines}
\Crefname{AlgoLine}{Line}{Lines}
\allowdisplaybreaks

\newcounter{adversarycf}
\newcounter{algocfbackup}

\makeatletter
\newenvironment{adversary}[1][]{\refstepcounter{adversarycf}\setcounter{algocfbackup}{\value{algocf}}

\let\adversary@oldalgorithmcfname\algorithmcfname
  \let\adversary@oldfnum@algocf\fnum@algocf
  \let\adversary@oldthealgocf\thealgocf
  \@ifundefined{theHalgocf}{\let\adversary@oldtheHalgocf\relax
  }{\let\adversary@oldtheHalgocf\theHalgocf
  }\let\adversary@oldalgocf@caption@algo\algocf@caption@algo
  \let\adversary@oldlabel\label

\renewcommand{\thealgocf}{\arabic{adversarycf}}

\renewcommand{\algorithmcfname}{Adversary}\renewcommand{\fnum@algocf}{\AlCapSty{\AlCapFnt\algorithmcfname\nobreakspace\arabic{adversarycf}}}

\long\def\algocf@caption@algo##1[##2]##3{\ifthenelse{\equal{\algocf@algocfref}{\relax}}{}{\algocf@captionref}\@ifundefined{hyper@refstepcounter}{\relax}{\renewcommand{\theHalgocf}{adv.\arabic{adversarycf}}\hyper@refstepcounter{algocf}}\algocf@latexcaption{##1}[{##2}]{{##3}}}

\def\adversary@label##1{\adversary@oldlabel[adversary]{##1}}\def\adversary@labelopt[##1]##2{\adversary@oldlabel[##1]{##2}}\def\label{\@ifnextchar[{\adversary@labelopt}{\adversary@label}}

  \begin{algorithm}[#1]}{\end{algorithm}

\let\label\adversary@oldlabel
  \let\algocf@caption@algo\adversary@oldalgocf@caption@algo
  \renewcommand{\thealgocf}{\adversary@oldthealgocf}\ifx\adversary@oldtheHalgocf\relax\else
    \renewcommand{\theHalgocf}{\adversary@oldtheHalgocf}\fi
  \let\fnum@algocf\adversary@oldfnum@algocf
  \let\algorithmcfname\adversary@oldalgorithmcfname
  \setcounter{algocf}{\value{algocfbackup}}}
\makeatother

\usetikzlibrary{arrows, decorations.markings}
\tikzstyle{vecArrow} = [thick, decoration={markings,mark=at position
   1 with {\arrow[semithick]{open triangle 60}}},
   double distance=1.4pt, shorten >= 5.5pt,
   preaction = {decorate},
   postaction = {draw,line width=1.4pt, white,shorten >= 4.5pt}]
\tikzstyle{innerWhite} = [semithick, white,line width=1.4pt, shorten >= 4.5pt]

\allowdisplaybreaks

\ifinformsclass\else
\usepackage{xfrac}
\fi
\usepackage{thm-restate}

\newcommand{\xhdr}[1]{\vspace{2mm} \noindent{\bf #1}}

\newcommand{\primed}{^{\dagger}}

\newcommand{\dimension}{d}
\newcommand{\ts}{t}
\newcommand{\tsSet}{[\timeHorizon]}
\newcommand{\tsBreak}{\timeHorizon_{1}}
\newcommand{\timeHorizon}{T}

\newcommand{\context}{x}
\newcommand{\contextAt}[1]{\context_{#1}} 

\newcommand{\ball}{\mathbb{B}}

\newcommand{\hvecs}{\Phi}
\newcommand{\hvec}{\theta}

\newcommand{\width}[2]{\mathtt{Width}\left(#1;#2\right)}
\newcommand{\vol}[1]{\mathtt{Vol}\left(#1\right)}

\newcommand{\feedback}{\sigma}
\newcommand{\feedbackAt}[1]{\feedback_{#1}}
\newcommand{\feedbackSellerAt}[1]{\feedback^{\mathrm{s}}_{#1}}
\newcommand{\feedbackBuyerAt}[1]{\feedback^{\mathrm{b}}_{#1}}
\newcommand{\feedbackSellerAtOf}[2]{\feedback^{\mathrm{s}}_{#1,#2}}
\newcommand{\feedbackBuyerAtOf}[2]{\feedback^{\mathrm{b}}_{#1,#2}}

\newcommand{\sellerVec}{\theta} \newcommand{\sellerVecAt}[1]{\sellerVec_{#1}} 
\newcommand{\sellerVecs}{\Phi}
\newcommand{\sellerVecsAt}[1]{\sellerVecs_{#1}}

\newcommand{\buyerVec}{\beta} \newcommand{\buyerVecAt}[1]{\buyerVec_{#1}} 
\newcommand{\buyerVecs}{\Psi}

\newcommand{\sellerValue}{c} \newcommand{\sellerPrice}{p} \newcommand{\sellerPriceAt}[1]{\sellerPrice_{#1}} \newcommand{\sellerPriceAtOf}[2]{\sellerPrice_{#1,#2}} \newcommand{\sellerPriceAtOfseg}[2]{\sellerPrice_{#1}\super{#2}} \newcommand{\sellerNum}{m} \newcommand{\sellerIndex}{i} \newcommand{\sellerValueAt}[1]{\sellerValue_{#1}}
\newcommand{\sellerSet}{\mathcal{S}}
\newcommand{\sellerSubset}[1]{\sellerSet\super{#1}}
\newcommand{\sellerSubsetAt}[2]{\sellerSubset{#1}_{#2}}

\newcommand{\sellerAccAt}[1]{\sellerSet^{\mathrm{acc}}_{#1}}
\newcommand{\sellerValueLb}{\ubar{\sellerValue}}
\newcommand{\sellerValueLbAt}[1]{\sellerValueLb_{#1}}
\newcommand{\sellerValueUb}{\bar{\sellerValue}}
\newcommand{\sellerValueUbAt}[1]{\sellerValueUb_{#1}}

\newcommand{\buyerValue}{v} \newcommand{\buyerPrice}{q} \newcommand{\buyerPriceAt}[1]{\buyerPrice_{#1}} \newcommand{\buyerPriceAtOf}[2]{\buyerPrice_{#1,#2}} \newcommand{\buyerPriceAtOfseg}[2]{\buyerPrice_{#1}\super{#2}} \newcommand{\buyerNum}{n} \newcommand{\buyerIndex}{j}  \newcommand{\buyerValueAt}[1]{\buyerValue_{#1}}
\newcommand{\buyerSet}{\mathcal{B}}
\newcommand{\buyerSubset}[1]{\buyerSet\super{#1}}
\newcommand{\buyerSubsetAt}[2]{\buyerSubset{#1}_{#2}}

\newcommand{\buyerAccAt}[1]{\buyerSet^{\mathrm{acc}}_{#1}} 
\newcommand{\buyerValueLb}{\ubar{\buyerValue}}
\newcommand{\buyerValueLbAt}[1]{\buyerValueLb_{#1}}
\newcommand{\buyerValueUb}{\bar{\buyerValue}}

\newcommand{\LB}[1]{\ubar{#1}}
\newcommand{\UB}[1]{\bar{#1}}
\newcommand{\buyerIntGapAt}[1]{\Delta_{#1}^{b}}
\newcommand{\sellerIntGapAt}[1]{\Delta_{#1}^{s}}

\newcommand{\buyerGapLevelAt}[1]{{h_{#1}}}

\newcommand{\agentSet}{\mathcal{N}}

\newcommand{\traderNum}{n}

\newcommand{\postPrice}{p} \newcommand{\postPriceAt}[1]{\postPrice_{#1}}  \newcommand{\sellerTwoPrices}[1]{\sellerPrice\super{#1}}
\newcommand{\buyerTwoPrices}[1]{\buyerPrice\super{#1}}

\newcommand{\OPT}{\textsc{Opt}}

\newcommand{\gft}{\textsc{Gft}}
\newcommand{\gftAt}[1]{\gft_{#1}}
\newcommand{\gftOpt}{\OPT}

\newcommand{\profit}{\textsc{Profit}}
\newcommand{\profitAt}[1]{\profit_{#1}}
\newcommand{\profitOpt}{\profit^{\star}}
\newcommand{\profitOptAt}[1]{\profitOpt_{#1}}

\newcommand{\ALG}{\texttt{ALG}\xspace}

\newcommand{\reg}{\textsc{Regret}}
\newcommand{\regAt}[1]{\reg_{#1}}

\newcommand{\regProfit}{\reg^{\profit}}
\newcommand{\regProfitAt}[1]{\regProfit_{#1}}

\newcommand{\match}{\mathcal{M}}
\newcommand{\matchAt}[1]{\match_{#1}}

\newcommand{\num}[1]{\#\left(#1\right)}

\newcommand{\priceIncreasingThresholdAt}[1]{\delta_{#1}}
\newcommand{\phaseNum}{\ell}
\newcommand{\totalPhaseNum}{L}
\newcommand{\phaseNumVar}{{\ell\primed}}
\newcommand{\hvecsPhase}[1]{\mathrm{P}_{#1}}
\newcommand{\hvecsAt}[1]{\hvecs_{#1}}
\newcommand{\phaseCount}{\phi}

\newcommand{\buyerValueOrdAt}[2]{\buyerValue_{(#1),#2}}
\newcommand{\buyerValueOrdLbOfAt}[2]{\buyerValueLb_{(#1),#2}}
\newcommand{\buyerValueOrdUbOfAt}[2]{\buyerValueUb_{(#1),#2}}
\newcommand{\sellerValueOrdAt}[2]{\sellerValue_{(#1),#2}}
\newcommand{\sellerValueOrdUbOfAt}[2]{\sellerValueUb_{(#1),#2}}
\newcommand{\sellerValueOrdLbOfAt}[2]{\sellerValueLb_{(#1),#2}}

\newcommand{\buyerValueLbOfAt}[2]{\buyerValueLb_{#1,#2}}
\newcommand{\buyerValueUbOfAt}[2]{\buyerValueUb_{#1,#2}}

\newcommand{\sellerValueUbOfAt}[2]{\sellerValueUb_{#1,#2}}
\newcommand{\sellerValueLbOfAt}[2]{\sellerValueLb_{#1,#2}}

\newcommand{\optProfitUpperBoundAt}[1]{\smash{\overline{\profit}}_{#1}^\star}

\newcommand{\ficTradingBuyersAt}[1]{\bar{\buyerSet}_{#1}}
\newcommand{\ficTradingSellersAt}[1]{\ubar{\sellerSet}_{#1}}

\newcommand{\sellerIntGapOfAt}[2]{\sellerIntGapAt{#1 , #2}}
\newcommand{\buyerIntGapOfAt}[2]{\buyerIntGapAt{#1 , #2}}

\newcommand{\sellerSearchIndex}[1]{\sellerIndex'_{#1}}
\newcommand{\buyerSearchIndex}[1]{\buyerIndex'_{#1}}
\newcommand{\tradeNum}{k}

\newcommand{\padIndex}{\ell}

\newcommand{\padPara}[1]{z_{#1}}

\newcommand{\sellerVecsOfAt}[2]{\sellerVecs_{#1,#2}}
\newcommand{\buyerVecsOfAt}[2]{\buyerVecs_{#1,#2}}

\newcommand{\tempPriceSellerOfAt}[2]{y_{#1,#2}^{\mathrm{s}}}
\newcommand{\tempPriceBuyerOfAt}[2]{y_{#1,#2}^{\mathrm{b}}}

\newcommand{\SinglePriceMech}{{\sffamily uniform pricing}\xspace}
\newcommand{\TwoPriceMech}{{\sffamily side-specific pricing}\xspace}
\newcommand{\SegPriceMech}{{\sffamily segmented pricing}\xspace}
\newcommand{\SinglePriceMechs}{{\sffamily uniform pricing}\xspace}
\newcommand{\TwoPriceMechs}{{\sffamily side-specific pricing}\xspace}

\newcommand{\OneToManyGFT}{\textsc{Trade\allowbreak-then\allowbreak-Refine}\xspace}

\newcommand{\MSSPM}{\textsc{{Optimistic\allowbreak-Profit\allowbreak-Search}}\xspace}

\newcommand{\ConservativeSegmentedPricing}{\textsc{Trade\allowbreak-and\allowbreak-Explore}\xspace}

\newcommand{\tradeRank}{\ell}

\newcommand{\ubar}[1]{\underaccent{\bar}{#1}}
\newcommand{\totalNum}{N}

\providecommand{\totalNum}{N}
\newcommand{\maxTradeNum}{K}

\providecommand{\OPT}{\textsc{Opt}}

\newcommand{\asymptoticO}{O}
\newcommand{\segPrecision}{\epsilon}

\newcommand{\sellerExplorePrice}{\widetilde{\sellerPrice}}
\newcommand{\buyerExplorePrice}{\widetilde{\buyerPrice}}
\newcommand{\sellerExplorePriceOf}[1]{\sellerExplorePrice_{#1}}
\newcommand{\buyerExplorePriceOf}[1]{\buyerExplorePrice_{#1}}
\newcommand{\sellerExplorePriceOfAt}[2]{\sellerExplorePrice_{#1,#2}}
\newcommand{\buyerExplorePriceOfAt}[2]{\buyerExplorePrice_{#1,#2}}
\newcommand{\sellerExplorePriceAt}[1]{\sellerExplorePrice_{#1}}
\newcommand{\buyerExplorePriceAt}[1]{\buyerExplorePrice_{#1}}

\newcommand{\selectedTradeNum}{\LB{\tradeNum}}
\newcommand{\selectedTradeNumAt}[1]{\selectedTradeNum_{#1}}
\newcommand{\commonPrice}{\pi}
\newcommand{\commonPriceAt}[1]{\commonPrice_{#1}}
\newcommand{\commonPriceLb}{\LB{\commonPrice}}
\newcommand{\commonPriceUb}{\UB{\commonPrice}}
\newcommand{\commonPriceLbAt}[1]{\commonPriceLb_{#1}}
\newcommand{\commonPriceUbAt}[1]{\commonPriceUb_{#1}}
\newcommand{\offeredPrice}{\rho}

\newcommand{\genericTraderIndex}{r}
\newcommand{\traderIntGap}{\Delta}
\newcommand{\traderIntGapOfAt}[2]{\traderIntGap_{#1,#2}}

\renewcommand{\sellerIntGapAt}[1]{\traderIntGap_{#1}^{\mathrm{s}}}
\renewcommand{\buyerIntGapAt}[1]{\traderIntGap_{#1}^{\mathrm{b}}}

\newcommand{\explorationOffset}{\eta}
\newcommand{\explorationOffsetOfAt}[2]{\explorationOffset_{#1,#2}}
\newcommand{\sellerExplorationOffsetOfAt}[2]{\explorationOffset^{\mathrm{s}}_{#1,#2}}
\newcommand{\buyerExplorationOffsetOfAt}[2]{\explorationOffset^{\mathrm{b}}_{#1,#2}}
\newcommand{\countedAcceptanceNum}{A}
\newcommand{\countedAcceptanceNumAt}[1]{\countedAcceptanceNum_{#1}}

\newcommand{\sellerRepresentativeAt}[1]{\sellerIndex_{#1}}
\newcommand{\buyerRepresentativeAt}[1]{\buyerIndex_{#1}}
\newcommand{\traderRadius}{R}
\newcommand{\traderRadiusOfAt}[2]{\traderRadius_{#1,#2}}
\newcommand{\sellerRadius}{\traderRadius^{\mathrm{s}}}
\newcommand{\buyerRadius}{\traderRadius^{\mathrm{b}}}
\newcommand{\sellerRadiusAt}[1]{\sellerRadius_{#1}}
\newcommand{\buyerRadiusAt}[1]{\buyerRadius_{#1}}
\newcommand{\chargedTraderSet}{\mathcal{I}}
\newcommand{\chargedTraderSetAt}[1]{\chargedTraderSet_{#1}}

\newcommand{\sellerValueOrdUb}[1]{\sellerValueUbAt{(#1)}}
\newcommand{\buyerValueOrdLb}[1]{\buyerValueLbAt{(#1)}}
\newcommand{\selectedGft}{G_0}
\newcommand{\remainingOpt}{\OPT_{\mathrm{rem}}}
\newcommand{\acceptanceCount}{h}

\newcommand{\QuantileTradeAndExplore}{\textsc{{Contextual} Trade-and-Explore}\xspace}

\newcommand{\ContextualTradeThenRefine}{\textsc{Contextual Trade\allowbreak-then\allowbreak-Refine}\xspace}
\newcommand{\contextNormAt}[1]{\lambda_{#1}} \newcommand{\contextDirAt}[1]{u_{#1}} \newcommand{\selectedBuyerAt}[1]{\buyerIndex_{#1}} \newcommand{\genericVec}{y} \newcommand{\querySet}{\mathcal{K}} \newcommand{\queryWidth}{w} \newcommand{\medianPad}[1]{\bar{z}_{#1}} \newcommand{\quantileScale}{h} \newcommand{\quantileWidthScale}[1]{W_{#1}} \newcommand{\quantileFraction}[1]{\rho_{#1}} \newcommand{\padRadius}{z} \newcommand{\padPotential}[2]{P_{#1}(#2)}    \newcommand{\capThickness}{s}  
\newcommand{\set}[1]{\left\{#1\right\}}

\newcommand{\lrangle}[1]{\langle #1 \rangle}
\newcommand{\paren}[1]{\left(#1\right)}

\newcommand{\abs}[1]{\left\lvert#1\right\rvert}

\newcommand{\super}[1]{^{(#1)}}

\providecommand{\argmax}{\operatorname*{arg\,max}}
\providecommand{\argmin}{\operatorname*{arg\,min}}

\newcommand{\reals}{\mathbb{R}}

\newcommand{\deq}{\triangleq}

\newcommand{\innerproduct}[1]
{\langle#1\rangle}

\newcommand{\prob}[2][]{\text{Pr}\ifthenelse{\not\equal{}{#1}}{_{#1}}{}\!\left[{\def\givenn{\middle|}#2}\right]}
\newcommand{\expect}[2][]{\mathbb{E}\ifthenelse{\not\equal{}{#1}}{_{#1}}{}\!\left[{\def\givenn{\middle|}#2}\right]}
\newcommand{\norm}[1]{\left\|#1\right\|}

\newcommand{\tparen}{\big}
\newcommand{\tprob}[2][]{\text{Pr}\ifthenelse{\not\equal{}{#1}}{_{#1}}{}\tparen[{\def\given{\tparen|}#2}\tparen]}
\newcommand{\texpect}[2][]{\mathbb{E}\ifthenelse{\not\equal{}{#1}}{_{#1}}{}\tparen[{\def\given{\tparen|}#2}\tparen]}

\newcommand{\sprob}[2][]{\text{Pr}\ifthenelse{\not\equal{}{#1}}{_{#1}}{}[#2]}
\newcommand{\sexpect}[2][]{\mathbb{E}\ifthenelse{\not\equal{}{#1}}{_{#1}}{}[#2]}

\newcommand{\ceil}[1]{\left\lceil #1 \right\rceil}
\newcommand{\floor}[1]{\lfloor #1 \rfloor}

\newcommand{\bfloor}[1]{\left\lfloor #1 \right\rfloor}

\makeatletter
\DeclareRobustCommand{\qed}{\ifmmode \else \leavevmode\unskip\penalty9999 \hbox{}\nobreak\hfill
  \fi
  \quad\hbox{\qedsymbol}}
\newcommand{\qedsymbol}{\openbox}
\renewenvironment{proof}[1][\proofname]{\par
  \normalfont
  \topsep6\p@\@plus6\p@ \trivlist
  \item[\hskip\labelsep\itshape
    #1.]\ignorespaces
}{\qed\endtrivlist
}
\newcommand{\proofname}{Proof}
\makeatother

\OneAndAHalfSpacedXI

\begin{document}

\RUNAUTHOR{Feng, Ma, Peng, and Wan}
\RUNTITLE{Dynamic Pricing in Two-Sided Markets: Learning Barriers and Price Differentiation}

\TITLE{Dynamic Pricing in Two-Sided Markets: Learning Barriers and Price Differentiation}

\ARTICLEAUTHORS{\AUTHOR{Yiding Feng}
\AFF{Hong Kong University of Science and Technology, China, \EMAIL{ydfeng@ust.hk}}

\AUTHOR{Mengfan Ma}
\AFF{Central China Normal University, China, \EMAIL{mengfanma1@gmail.com}}

\AUTHOR{Bo Peng}
\AFF{Tel Aviv University, Israel, \EMAIL{ahqspbo@gmail.com}}

\AUTHOR{Zongqi Wan}
\AFF{Great Bay University, China, \EMAIL{zongqiwan98@gmail.com}}
}

\ABSTRACT{
We study dynamic pricing for gains-from-trade maximization in two-sided markets. Over $\timeHorizon$ rounds, a platform posts prices to traders with fixed, unknown costs and values and learns from their accept/reject decisions. Under weak budget balance and random matching among accepting traders, we investigate how market composition and pricing flexibility affect the regret achievable by deterministic learning policies.

Uniform pricing achieves constant regret in bilateral trade and $O(\log\log\timeHorizon)$ regret in one-to-many markets, with both rates remaining order-optimal even when side-specific pricing is allowed. However, side-specific pricing incurs $\Omega(\timeHorizon)$ worst-case regret with only three sellers and three buyers. This barrier persists even when a common price achieves the optimum under full information. The difficulty arises from a mismatch phenomenon: less efficient accepting traders can displace more efficient traders, making price exploration costly.

We overcome this barrier through segmented pricing, which partitions each side into at most two groups and assigns a price to each group. Our algorithm achieves $O(\buyerNum\log\log\timeHorizon+\buyerNum^2)$ regret with $\buyerNum$ traders on each side, which is order-optimal in the horizon. Thus, limited price differentiation within each side restores vanishing average regret even when the platform seeks to maximize gains-from-trade. We also extend our results to contextual markets and show that side-specific pricing suffices for doubly logarithmic regret under contextual profit maximization.\footnote{{An extended abstract of this work has appeared in the Twenty-Seventh ACM Conference on Economics and Computation (EC 2026).}}

 }


\maketitle

\section{Introduction}
\label{sec:intro}

Pricing is a central operational decision in two-sided markets, where a platform intermediates trade between sellers and buyers. When the platform sets the terms of exchange, its prices determine how many sellers are willing to sell and how many buyers are willing to buy. Paying sellers more attracts more supply, and charging buyers less attracts more demand. The platform cannot push both levers without limit, however: the payment it offers sellers cannot exceed the price it charges buyers. Such decisions arise in settings such as managed retail marketplaces and online service platforms, where the platform sets both the payments to suppliers and the prices charged to customers.\footnote{\Cref{app:applications} describes concrete applications and discusses evidence on the economic scale of the broader markets in which these platforms operate.} Which trades ultimately occur depends on two factors: which traders accept their offers, and which of the accepting traders are matched with each other.

A key difficulty is that the platform typically does not know how much traders are willing to accept or pay. Each seller privately knows their production cost, and each buyer privately knows their value for the product. The platform must therefore learn good prices through repeated interactions. Posted prices offer a simple way to conduct these interactions: each trader receives a take-it-or-leave-it offer and either accepts or rejects it. Each response reveals partial information about the trader's private cost or value, which we call the trader's \emph{type}. For example, if a seller accepts a price $p$, the platform learns only that the seller's cost is at most $p$. This information can guide subsequent pricing decisions. At the same time, the prices posted to gather information also determine who participates in the current round. A poorly chosen price may block beneficial trades, or may allow less efficient traders (sellers with higher costs or buyers with lower values) to trade in place of more efficient ones. We refer to the second effect as \emph{mismatch}.

This interaction between learning and allocation makes the choice of pricing scheme important. The simplest scheme is uniform pricing, which offers the same price to every trader. When traders' types are known, a well-chosen uniform price can achieve an efficient allocation in many markets. Offering one price to all sellers and a possibly different price to all buyers (side-specific pricing) gives the platform separate control over participation on each side. Differentiating prices among traders on the same side (as in segmented pricing, defined below) adds still more flexibility, but makes the pricing scheme more complex. It is therefore natural to ask how much pricing flexibility a platform needs when it must learn traders' types. The question arises even in markets where a single common price would suffice if those types were known.

Motivated by this question, we study dynamic pricing in a two-sided market in which $\sellerNum$ sellers and $\buyerNum$ buyers interact with a platform, which we refer to as the \emph{principal}, over $\timeHorizon$ rounds. In each round, every seller can supply at most one unit of a homogeneous product, and every buyer demands at most one unit. Each seller has a private cost and each buyer has a private value, both in $[0,1]$. In our baseline model, these types remain fixed over all rounds. In each round, the principal posts prices and observes only whether each trader accepts or rejects. A matching then pairs as many accepting sellers with accepting buyers as possible; it is selected uniformly at random among all such maximum-cardinality matchings. Thus, when the numbers of accepting sellers and buyers differ, some accepting traders remain unmatched. In each matched pair, the buyer pays and the seller receives their respective posted prices.

Our primary objective is to maximize \emph{gains-from-trade} (GFT), the total surplus created by trade: the total value of matched buyers minus the total cost of matched sellers. We evaluate a learning policy by its \emph{regret}, the cumulative shortfall of its GFT relative to the first-best GFT, which is the highest GFT attainable when all traders' types are known. Regret that grows sublinearly in $\timeHorizon$ means that the average per-round loss in GFT vanishes as the horizon grows; we refer to this as \emph{vanishing average regret}. Throughout the paper, we focus on deterministic learning policies and compute regret in expectation over the random matching.
We consider three classes of pricing schemes, ordered by increasing flexibility:
\begin{itemize}
 \item \emph{Uniform pricing} posts one common price to all sellers and buyers.
 \item \emph{Side-specific pricing} posts one price to all sellers and another, possibly different, price to all buyers.
 \item \emph{Segmented pricing} partitions the sellers into at most two groups and the buyers into at most two groups, and posts one price to each group. The groups may change from round to round.
\end{itemize}
Different degrees of price differentiation arise in practice. Wholesale electricity markets apply a common price to producers and consumers within a given pricing area and delivery interval, resembling uniform pricing. Managed retail marketplaces and ride-hailing platforms, in contrast, separate the price that customers pay from the payment that suppliers receive, as in side-specific pricing.\footnote{We discuss these applications and their connection to our pricing schemes in \Cref{app:applications}.} Segmented pricing adds a limited form of differentiation within each side, which allows us to examine what this extra flexibility can achieve during learning. In our model, segments affect only the prices that traders receive: any accepting seller can still be matched with any accepting buyer.

We examine these pricing schemes in three market settings of increasing complexity: bilateral trade (one seller and one buyer), one-to-many markets (one trader on one side and several on the other), and many-to-many markets (several traders on both sides). This progression separates two learning difficulties: finding prices that induce trade at all, and finding prices that induce the right traders to participate. Our central research question is:
\begin{displayquote}
 \myboxtwo{\emph{As a two-sided market moves from bilateral trade to many-to-many settings, how does learning efficient prices become harder, and how much pricing flexibility restores vanishing average regret?}}
\end{displayquote}

\subsection{Our Main Contributions}

Our results show that market composition (whether each side has one trader or many) and pricing flexibility jointly determine how costly it is to learn efficient prices. In bilateral trade, uniform pricing achieves constant regret, which does not grow with the horizon. In one-to-many markets, it achieves doubly logarithmic regret, of order $\log\log\timeHorizon$. In many-to-many markets, however, even side-specific pricing cannot avoid linear regret in the worst case, so the average loss in GFT does not vanish.

\mybox{
\begin{displayquote}
\textbf{\underline{Main message.}} Even when a single common price would be efficient if traders' types were known, a platform that learns with one price on each side can incur linear regret. The cause is mismatch: while prices are being explored, less efficient traders can displace more efficient ones. Allowing at most two price segments on each side restores vanishing average regret, with an order-optimal doubly logarithmic dependence on the horizon.
\end{displayquote}
}

We develop this message through three main results and two extensions, summarized in \Cref{tab:our-results}.

\begin{table}
    \centering
    \caption{
        A summary of our results on regret bounds in various two-sided market settings.}
    {\footnotesize\begin{tabular}{lccc}
\toprule
\multirow{2}{*}{\textbf{Settings}} & \multicolumn{2}{c}{\textbf{Gains-from-trade}} & \textbf{Profit} (\TwoPriceMech)
\\
\cmidrule(lr){2-3} \cmidrule(lr){4-4}
& \textbf{Regret bound} & \textbf{Mechanism} & \textbf{Regret bound}
\\
\midrule
\textbf{Bilateral Trade} &
\makecell{$\Theta(1)$\\{[}Thm.~\ref{thm:uniform-pricing}{]}} &
\SinglePriceMech &
\multirow[c]{8}{*}{\makecell{$\Theta(\log \log \timeHorizon)$\\{[}Thm.~\ref{thm:extension:contextual-profit}{]}}}
\\
\cmidrule(lr){1-3}
\textbf{One-to-Many Market} &
\makecell{$\Theta(\log \log \timeHorizon)$\\{[}Thms.~\ref{thm:uniform-pricing}, \ref{thm:mismatch}{]}} &
\SinglePriceMech &
\\
\cmidrule(lr){1-3}
\multirow[c]{4}{*}{\textbf{Many-to-Many Market}} &
\makecell{$\Theta(\timeHorizon)$\\{[}Thm.~\ref{thm:mismatch}{]}} &
\TwoPriceMech &
\\
\cmidrule(lr){2-3}
& \makecell{$\Theta(\log \log \timeHorizon)$\\{[}Thm.~\ref{thm:segmented-pricing}{]}} &
\SegPriceMech &
\\
\bottomrule
\end{tabular} }
    \begin{tablenotes}
  \item \emph{\footnotesize\underline{Note}:
  The symbol $\timeHorizon$ denotes the time horizon. All bounds suppress the dependence on the number of traders on each side.
  The notation $\Theta(\cdot)$ indicates a tight regret bound; the lower bound in each $\Theta(\cdot)$ entry holds against any algorithm that implements the corresponding pricing scheme class.
  Each upper bound extends to $\dimension$-dimensional contextual markets under the same pricing scheme, with regret $O(\dimension\, R(\timeHorizon) + \dimension \log \dimension)$, where $R(\timeHorizon)$ denotes the corresponding regret bound in the baseline model; this extension retains the order-optimal dependence on $\timeHorizon$ (Thms.~\ref{thm:extension:contextual-gft} and~\ref{thm:extension:contextual-profit}).}
\end{tablenotes}
    \label{tab:our-results}
\end{table}

\xhdr{Uniform pricing suffices in bilateral trade and one-to-many markets.}
We first design a single algorithm, {\OneToManyGFT} (\Cref{alg:GFTfixed}), that uses uniform pricing in both settings. Its regret is at most $1$ in bilateral trade and $O(\log\log\timeHorizon)$ in one-to-many markets (\Cref{thm:uniform-pricing}). The latter bound does not depend on the number of traders on the larger side.

The gap between these two guarantees reflects the difference between inducing trade and inducing efficient trade. In bilateral trade, any price that both traders accept captures all of the available GFT. Once such a price is found, the principal does not need to learn either trader's type more precisely. With one seller and several buyers, however, a price that the seller and several buyers accept may lead the seller to trade with a buyer who does not have the highest value. This is the simplest form of mismatch. To control the resulting loss, the algorithm must further refine the price to screen out lower-value buyers.

\xhdr{Mismatch creates stronger learning barriers as the market expands.}
We show that the increase from constant to doubly logarithmic regret is unavoidable, even if the principal can use side-specific pricing. Specifically, every learning algorithm using side-specific pricing incurs $\Omega(\log\log\timeHorizon)$ regret on some instance with one seller and two buyers (\Cref{thm:mismatch}\ref{item:mismatch:one-to-many}). Together with our upper bound, this shows that {\OneToManyGFT} has the optimal dependence on $\timeHorizon$ among all algorithms using uniform or side-specific pricing.

The barrier becomes much stronger when both sides contain multiple traders. We construct a family of instances with only three sellers and three buyers on which every algorithm using side-specific pricing incurs $\Omega(\timeHorizon)$ regret (\Cref{thm:mismatch}\ref{item:mismatch:many-to-many}). Thus, separate prices for sellers and buyers are not enough to guarantee vanishing average regret.

Both lower bounds arise from mismatch: avoiding the displacement of more efficient traders by less efficient ones requires learning suitable prices, and that learning can itself be costly. Crucially, both families of hard instances admit a uniform price that achieves the first-best GFT when traders' types are known. The lower bounds therefore isolate a difficulty that comes from learning: they hold even though the available pricing scheme is expressive enough to implement an efficient allocation. They also hold with a small, fixed number of traders, so the higher regret cannot be attributed to a number of unknown types that grows with the horizon.

\xhdr{Segmented pricing overcomes the linear-regret barrier.}
We next introduce {\ConservativeSegmentedPricing}, which uses at most two price segments on each side. In a market with $\buyerNum$ sellers and $\buyerNum$ buyers, it achieves regret $O(\buyerNum\log\log\timeHorizon+\buyerNum^2)$. We also show that every algorithm using segmented pricing incurs $\Omega(\buyerNum\log\log\timeHorizon)$ regret on some instance. Hence the doubly logarithmic term is order-optimal in both the number of traders and the time horizon (\Cref{thm:segmented-pricing}).

The algorithm uses the additional flexibility to separate two kinds of traders: those whose acceptance can already be guaranteed, and those whose types still require exploration. It offers a common price to equally sized groups of sellers and buyers of the first kind, and exploratory prices to all remaining traders. The common price secures a guaranteed number of trades, while the exploratory prices are designed to limit the regret caused when the remaining traders also accept. The regret bound holds for every maximum-cardinality matching among accepting traders, and hence also under the randomized matching rule.

One implication is that limited price differentiation can support market efficiency while the principal learns. Even a principal whose sole objective is to maximize GFT can benefit from differentiating prices within each side. The benefit comes from coordinating information acquisition with participation: the principal can explore uncertain traders without putting at risk, in every experiment, the trades it can already secure.

\xhdr{Contextual markets.}
We extend the model so that costs and values can vary with an observed $\dimension$-dimensional context, such as product characteristics or market conditions. Each trader has a fixed, unknown feature vector, and the trader's cost or value in each round is the inner product of this vector with the current context. The principal observes each context before posting prices, and contexts may be chosen adversarially.

In contextual bilateral trade, uniform pricing achieves $O(\dimension\log\dimension)$ regret, independent of the horizon. In contextual one-to-many markets with $\buyerNum$ buyers, uniform pricing achieves $O(\buyerNum\dimension\log\log\timeHorizon+\buyerNum\dimension\log\dimension)$ regret. In contextual many-to-many markets with $\buyerNum$ traders on each side, an extension of our segmented-pricing algorithm achieves $O(\buyerNum\dimension\log\log\timeHorizon+\buyerNum^2\dimension\log\dimension)$ regret.
Thus, the dependence on the horizon in our baseline results, constant in bilateral trade and doubly logarithmic otherwise, carries over to contextual markets (\Cref{thm:extension:contextual-gft}). Up to the dependence on the number of traders, each bound takes the form $O(\dimension R(\timeHorizon)+\dimension\log\dimension)$, where $R(\timeHorizon)$ is the corresponding regret bound in the baseline model. The main additional challenge is to transfer information across contexts: what the principal learns about a trader under one context may reveal little about the trader's type under another. We address this challenge by maintaining a convex uncertainty set for each trader's feature vector and adapting volume-based tools from contextual search \citep{LLS-21}.

\xhdr{Profit maximization.}
We also study a principal who maximizes profit, defined as the payments collected from buyers minus the payments made to sellers. We analyze this objective directly in the contextual setting, which includes the baseline model as a special case. Because uniform pricing always yields zero profit, we use side-specific pricing and measure regret against the highest profit that side-specific pricing can achieve when traders' types are known.
For this objective, we design {\MSSPM}, which attains regret $O(\buyerNum\dimension\log\log\timeHorizon
    +\buyerNum^2\dimension\log\dimension)$ in markets with $\buyerNum$ traders on each side, again with order-optimal dependence on the horizon (\Cref{thm:extension:contextual-profit}).

Why is side-specific pricing enough for profit maximization but not for GFT maximization? Under side-specific pricing, profit equals the number of trades multiplied by the price spread, that is, the difference between the buyer-side and seller-side prices. Given which traders accept, profit does not depend on which of them the matching rule selects. Mismatch, which drives the learning barrier for GFT, therefore does not directly reduce profit. The principal must still learn the right trade volume and price spread, but side-specific pricing provides enough flexibility to do so.

\xhdr{Numerical evidence and implications.}
Our numerical experiments compare {\ConservativeSegmentedPricing} with {\MSSPM} and with pooled versions of Scouting Bandits and Scouting Blindits \citep{CCCFL-24a} in randomly drawn markets with 100 or 500 traders per side. {\ConservativeSegmentedPricing} captures 99.36--99.81\% of first-best GFT, more than every other algorithm in each of the 2000 instances, and leaves almost all of it to traders. Scouting Bandits and Scouting Blindits capture 42.13--74.55\%. {\MSSPM} keeps on average 47.82\% of first-best GFT as platform profit and forgoes about a quarter. Overall, our algorithms outperform Scouting Bandits and Scouting Blindits (\Cref{sec:numerical-experiments}).

\subsection{Overview of Our Algorithms and Techniques}

Our algorithms and lower bounds build on the interaction between the information revealed by posted prices and the losses caused by the resulting trades. We explain the main ideas below.

\xhdr{Trade first, then refine.}
In bilateral trade, identifying both private types precisely is unnecessary for GFT maximization. It suffices to locate any price between the seller's cost and the buyer's value. This observation suggests maintaining a common uncertainty interval and conducting a binary search until trade occurs.

The regret structure makes this search inexpensive. When the difference between the buyer's value and the seller's cost is large, the interval of acceptable prices is wide and a trading price is easy to locate. When the difference is small, the loss from failing to trade is also small. The shrinking search interval therefore controls the regret incurred before the first trade, yielding a bound independent of the horizon.

The same initial search applies to one-to-many markets, but the first trade may be inefficient. Our algorithm then switches to a conservative price-refinement procedure inspired by \citet{KL-03}. Consider a single seller and multiple buyers. Increasing the price can exclude lower-value buyers, but increasing it too far causes all buyers to reject. These outcomes have asymmetric consequences: when buyers continue to accept, regret is bounded by the remaining uncertainty about their values; when all buyers reject, regret can be constant. Conservative search balances these consequences by making small price increases during acceptance rounds and obtaining a much sharper reduction in uncertainty after rejection. Summing the losses across doubly logarithmically many uncertainty scales gives the one-to-many guarantee.

\xhdr{Lower bounds that preserve the cost of mismatch.}
The one-seller, two-buyer lower bound exploits the same asymmetry, but requires an additional construction. If both buyers accept, expected regret is half the difference between their values. Thus, uncertainty about the location of the buyer values alone does not imply substantial regret: the values could be arbitrarily close, making mismatch nearly costless.

Our adversary maintains a nested interval containing possible values for both buyers and ultimately chooses the interval endpoints as their fixed values. An aggressive buyer-side price causes both buyers to reject and incurs constant regret. A conservative price permits both buyers to accept but reveals information slowly. A carefully chosen stopping rule preserves enough separation between the final values to ensure that repeated acceptance rounds accumulate substantial mismatch regret. Every response remains consistent with the final fixed types. This yields the doubly logarithmic lower bound even when the seller receives a separate price.

In the many-to-many construction, uncertainty about a narrow interval of marginal types can displace traders whose contributions to GFT remain bounded away from zero. We use one zero-cost seller, two value-one buyers, and three additional traders whose types lie in a common uncertainty interval. The adversary prevents the posted prices from inducing all six traders to accept while preserving a nonempty interval of efficient common prices. Each round then either misses beneficial trades or leaves the zero-cost seller unmatched with constant probability. Unlike in the one-to-many setting, shrinking the uncertainty interval does not make these losses small, which leads to linear regret.

\xhdr{Simultaneously trade and explore.}
Our segmented pricing algorithm, {\ConservativeSegmentedPricing}, combines trade and exploration in each round. It maintains an uncertainty interval for each trader and uses seller upper bounds and buyer lower bounds to select the largest equally sized groups whose members are guaranteed to accept a common price. These selected groups receive the common price, while the remaining sellers receive a weakly lower exploratory price and the remaining buyers receive a weakly higher exploratory price.

The common price guarantees a number of trades, but does not reserve matching partners or ensure that every selected trader is matched. Traders outside the selected groups may accept their exploratory prices and displace selected traders. The key idea is to make these potentially harmful acceptances informative, so that the resulting mismatch also helps the algorithm learn.

For a trader with uncertainty width $\traderIntGap$, the individual exploratory price lies approximately $\traderIntGap^{3/2}/2$ above the lower cost bound for a seller or below the upper value bound for a buyer. Acceptance at such a price sharply contracts the uncertainty interval. The algorithm combines these individual prices into one exploratory price for outside sellers and one for outside buyers, coordinating them with the common price to preserve both acceptance by the selected groups and weak budget balance.

Let $\totalNum\triangleq\sellerNum+\buyerNum$ and $\maxTradeNum\triangleq\min\{\sellerNum,\buyerNum\}$ denote the total number of traders and the maximum possible number of trades in a round, respectively. The regret analysis links the cost of exploration to progress in learning. A potentially harmful outside acceptance sharply reduces uncertainty about the accepting trader's type, limiting the cumulative mismatch loss to $\asymptoticO(\totalNum\log\log\timeHorizon)$. The remaining loss comes from uncertainty about suitable trading prices. We control this loss through a potential that decreases as the uncertainty intervals shrink, yielding an additional $\asymptoticO(\totalNum\maxTradeNum)$ term. Thus, {\ConservativeSegmentedPricing} achieves low regret by ensuring that costly exploration also produces substantial information, rather than eliminating mismatch in every round.

The design of {\ConservativeSegmentedPricing} also provides a useful comparison with the familiar \emph{explore-then-exploit} approach \citep{LS-20}, which separates exploration and exploitation into successive phases. Here, the common-price offers play an exploitation role by securing participation from the selected groups, while the exploratory offers gather information about the remaining traders in the same round. Segmented pricing makes this coordination across trader groups possible. Its value is highlighted by \Cref{thm:mismatch}\ref{item:mismatch:many-to-many}: any deterministic explore-then-exploit policy restricted to side-specific pricing throughout the horizon incurs $\Omega(\timeHorizon)$ regret on some three-seller, three-buyer instance, regardless of how the exploration phase ends or the exploitation prices are chosen. In fact, the lower bound applies even when side-specific prices are continually adapted. Thus, changing the timing of exploration alone cannot overcome the learning barrier. \textsc{Trade-and-Explore} instead uses price differentiation within each side to support trade while exploring, with the exploratory prices ensuring that potentially harmful acceptances also provide substantial information.

 \subsection{Related Work}
\label{sec:related-work}
Our work connects to several lines of research.
We review the most relevant developments below.

\xhdr{Online learning in bilateral trade and two-sided market.}
In the online setting with stochastic or adversarial types, a series of works provides an essentially complete characterization of regret bounds \citep{BSZ-06,CCCFL-24a,CCCFL-24b,AFF-24,BCCF-24,CJLZ-25,DDFS-25,LCM-26}.
As in our model, these works primarily study fixed-price mechanisms that satisfy budget balance, under partial feedback in which the learner observes only traders' intentions (i.e., whether they accept or reject the posted price).
On the positive side, \citet{CCCFL-24a} design sublinear-regret algorithms for GFT maximization in the stochastic setting with smooth distributions.
\citet{AFF-22} provide an algorithm achieving a tight sublinear $2$-regret for GFT maximization for an oblivious adversary.
\citet{CCCFL-23} show that sublinear regret remains achievable beyond the i.i.d.~stochastic setting under a $\sigma$-smooth adversary model for GFT maximization.
On the negative side, no sublinear regret is possible even in the stochastic setting \citep{CCCFL-24a,AFF-24,CJLZ-25}.

In two-sided markets, \citet{BFN-24} consider sample-based design of DSIC mechanisms that maximize GFT, proving an impossibility under correlated values (already with one seller and two buyers) and giving an efficient learning algorithm for the one-seller--two-buyer case under independent values.
\citet{LCM-25} study a repeated one-seller--many-buyer market in a stochastic setting, in which each round runs a second-price auction among buyers and trades only if the winning bid exceeds a posted seller price, and obtain sublinear regret.
\citet{LH-26} learn profit-maximizing prices for multiple products under inventory constraints using generalized stochastic gradient descent. Other platform models study matching and assortment decisions \citep{HZ-22,JKK-21,GGISUW-22}. We study price learning with individual acceptance feedback and a prescribed matching rule.

\xhdr{Contextual search and pricing.}
\citet{KL-03} initiate the study of dynamic pricing in non-contextual settings. They show that double-logarithmic regret is attainable and establish a foundational $\Omega(\log\log \timeHorizon)$ regret lower bound for adversarial posted-price auctions. Extending this result to contextual scenarios yields a $\Omega(\dimension\log\log \timeHorizon)$ lower bound \citep{LS-18}. 
\citet{CLP-20,CLP-17} show that naive high-dimensional generalizations of the approach in \citet{KL-03} incur exponential regret. To address this, they develop the \textsc{EllipsoidPricing} algorithm, which maintains ellipsoidal uncertainty sets and attains $O(\dimension^2 \log \timeHorizon)$ regret. \citet{LLV-18} then introduce \textsc{ProjectedVolume}, improving the bound to $O(\dimension \log \timeHorizon)$.
\citet{LS-18} and \citet{LLS-21} eventually achieve the near-optimal regret bound $O_{\dimension}(\log\log \timeHorizon)$ by leveraging tools from integral geometry, in particular intrinsic volumes and the Steiner polynomial.
Other variants have also been studied, including contextual search with irrational agents \citep{KLPS-21,Zuo-24}, contextual pricing with noise \citep{LSS-22,TGMP-24}, stochastic contexts \citep{ARS-14,JH-19,LSS-22}, non-myopic agents \citep{GJM-19}, heterogeneous buyers \citep{MLS-18,LNOPZ-25}, non-stationary environment \citep{FR-26}, and general principal-agent game \citep{FMPW-26}. Other one-sided pricing models with stochastic valuations study learning unknown linear demand \citep{KZ-14}, demand learning with limited inventory \citep{WDY-14}, limited price changes \citep{CSW-17}, joint pricing and replenishment \citep{CCA-19}, and learning from a customer pool depleted by purchases \citep{CJL-24}. These studies connect experimentation to demand and operational constraints; our focus is simultaneous price differentiation across traders.
A closely related one is the contextual bilateral trade problem studied by \citet{GBCCP-25}, who propose an algorithm with $O(\dimension^2\log \timeHorizon)$ regret for GFT maximization. Our contextual bilateral result removes the dependence on the horizon. Concurrently, \citet{CFGLPR-26} obtain horizon-independent GFT regret and doubly logarithmic profit regret in contextual bilateral trade under individual accept/reject feedback and per-round budget balance, matching the horizon dependence of our bilateral results. Our main distinction from both works is the learning barrier in many-to-many markets, where multiple traders compete to participate, and its resolution through limited price differentiation within each side.

\xhdr{Welfare effects of price differentiation.}
Price differentiation can improve efficiency and benefit consumers. In a static monopoly model, \citet{BBM-15} show that segmentation can achieve efficient trade while preserving uniform-price monopoly profit. Research further examines how competition, product differentiation, consumer search, and fairness constraints shape these welfare effects \citep[e.g.,][]{CGMR-05,CCM-20,CEL-22,WLK-23}. Related work considers data policies, privacy constraints, and alternative measures of consumer benefit \citep{LL-23,BMS-24,FJMM-24,BCHX-25}. Our results identify a dynamic benefit: differentiated prices support learning to maximize GFT.

\section{Preliminaries}
\label{sec:prelim}

Our paper studies the problem of \emph{dynamic pricing for two-sided markets}. We describe the components of our base model below.

\xhdr{Environment.}
Consider an online two-sided market with $\sellerNum$ \emph{sellers} and $\buyerNum$ \emph{buyers} that interact through a principal (or learner) over $\timeHorizon$ rounds, indexed by $[\timeHorizon]\triangleq \{1, 2, \dots, \timeHorizon\}$. The set of sellers is denoted by $\sellerSet$, with $\abs{\sellerSet} = \sellerNum$, and the set of buyers is denoted by $\buyerSet$, with $\abs{\buyerSet} = \buyerNum$. Let $\agentSet \triangleq \sellerSet \cup \buyerSet$ be the set of all traders. In each round, each seller $\sellerIndex \in \sellerSet$ can produce at most one unit of a homogeneous product at a fixed \emph{private cost} $\sellerValueAt{\sellerIndex} \in [0, 1]$, while each buyer $\buyerIndex \in \buyerSet$ demands at most one unit and has a fixed \emph{private value} $\buyerValueAt{\buyerIndex} \in [0, 1]$ for receiving it. These costs and values constitute the traders' private types, are unknown to the principal, and remain constant throughout the $\timeHorizon$ rounds.\footnote{In \Cref{sec:extensions}, we study a contextual extension in which traders' private types may vary across rounds as the context changes.}

We also consider two special cases of this general market that are important in both the literature and practice. The first is \emph{bilateral trade}, the classic setting with a single seller and a single buyer. The second is the \emph{one-to-many market}, which consists of either a single seller and multiple buyers or a single buyer and multiple sellers.

\xhdr{Pricing schemes and trade implementation.}
In each round $\ts \in [\timeHorizon]$, the principal uses a \emph{pricing scheme} to post prices to all traders.
Each trader responds with a binary decision: \emph{accept} or \emph{reject}.
A seller accepts if and only if the offered price is at least their cost, while a buyer accepts if and only if the offered price is at most their value; otherwise, the trader rejects.\footnote{When a trader is indifferent between accepting and rejecting, we assume that they accept the offered price.}

Among all possible pricing schemes, we focus on the following three classes, ordered by increasing expressive power:
\begin{itemize}
    \item {\SinglePriceMech}: A single price $\postPriceAt{}$ is posted uniformly to all traders. Each trader decides whether to accept or reject this common price $\postPriceAt{}$.
    \item {\TwoPriceMech}: A seller-side price $\sellerPriceAt{}$ is offered to all sellers, and a buyer-side price $\buyerPriceAt{}$, which may differ from the seller-side price, is offered to all buyers. Each trader decides whether to accept or reject the applicable price: $\sellerPriceAt{}$ for sellers and $\buyerPriceAt{}$ for buyers.
    \item {\SegPriceMech}: The set of sellers $\sellerSet$ (resp.\ buyers $\buyerSet$) is partitioned into at most two disjoint segments $(\sellerSubsetAt{1}{},\sellerSubsetAt{2}{})$ (resp.\ $(\buyerSubsetAt{1}{},\buyerSubsetAt{2}{})$). Each segment is offered a price, denoted by $\sellerPriceAtOfseg{}{1}$, $\sellerPriceAtOfseg{}{2}$, $\buyerPriceAtOfseg{}{1}$, $\buyerPriceAtOfseg{}{2}$, respectively; these prices need not be distinct. Each trader decides whether to accept or reject the price assigned to their segment.
\end{itemize}
For both {\TwoPriceMech} and {\SegPriceMech}, we further impose a \emph{weak budget balance constraint}, requiring that the price offered to every seller be no greater than the price offered to every buyer.

After all traders have responded, a matching is selected uniformly at random from all maximum-cardinality bipartite matchings between the buyers and sellers who accepted their posted prices; we refer to this as the \emph{randomized matching rule}.\footnote{The regret upper bounds in Theorems~\ref{thm:uniform-pricing},~\ref{thm:segmented-pricing},~\ref{thm:extension:contextual-gft} and~\ref{thm:extension:contextual-profit} hold even if an adversary selects the matching from all maximum-cardinality matchings between accepting traders.} For each matched buyer-seller pair, one unit of the product is transferred from the seller to the buyer. The buyer pays their posted price, and the seller receives their posted price. All unmatched traders receive zero allocation and make or receive no payment.\footnote{Regardless of which maximum-cardinality matching is selected, each of the three pricing schemes, together with this allocation and payment rule, is dominant-strategy incentive compatible and ex post individually rational when traders are myopic.}

As described above, the principal does not know the traders' costs and values and interacts with them over $\timeHorizon$ rounds. We consider different settings according to the class of pricing schemes available to the principal: {\SinglePriceMech}, {\TwoPriceMech}, or {\SegPriceMech}. In each setting and each round $\ts \in [\timeHorizon]$, the principal chooses a pricing scheme from the specified class and posts the resulting prices to all traders. The principal then observes only each trader's accept/reject decision.

\xhdr{Gains-from-trade maximization.}
The principal's goal is to maximize the \emph{gains-from-trade} (\emph{GFT}), a fundamental economic objective in two-sided markets.\footnote{In \Cref{sec:extensions}, we consider a variant of the model with profit maximization as the objective.} GFT depends only on the allocation outcome (i.e., the matching between buyers and sellers). Specifically, given any matching $\match \subseteq \sellerSet \times \buyerSet$, the induced GFT denoted by $\gftAt{}[\match]$ is defined as the total value of all matched buyers minus the total cost of all matched sellers:
\begin{align*}
    \gftAt{}[\match] \triangleq \sum\nolimits_{(\sellerIndex,\buyerIndex) \in \match}
    \paren{\buyerValueAt{\buyerIndex} - \sellerValueAt{\sellerIndex}}.
\end{align*}
We define the \emph{optimum GFT benchmark} $\gftOpt$ as the highest GFT achievable over all feasible matchings between sellers and buyers.
This benchmark can be computed as
\begin{align*}
    \gftOpt = \sum\nolimits_{\ell\in [k]} \paren{\buyerValueAt{(\ell)} - \sellerValueAt{(\ell)}},
\end{align*}
where $\buyerValueAt{(\ell)}$ (resp.\ $\sellerValueAt{(\ell)}$) denotes the $\ell$-th largest buyer value (resp.\ $\ell$-th smallest seller cost). The \emph{efficient trade size} $k \in [0:\min\{\buyerNum,\sellerNum\}]$ is the largest index for which the corresponding buyer value is at least the corresponding seller cost, i.e., $\buyerValueAt{(k)} \ge \sellerValueAt{(k)}$; if no such index exists, then $k=0$ and $\gftOpt=0$.

\xhdr{Regret minimization.}
We evaluate an online learning algorithm by its \emph{regret} relative to the optimum GFT benchmark $\gftOpt$.
Formally, for any online learning algorithm $\ALG$, its immediate regret in each round $\ts \in [T]$ and its cumulative regret over all $T$ rounds are defined as follows:\footnote{Throughout the paper, we focus on deterministic online learning algorithms.}
\begin{align*}
    \regAt{\ts}[\ALG] \triangleq
    \gftOpt -
    \expect{\gft[\matchAt{\ts}^{\ALG}]}
    \;\;
    \mbox{and}
    \;\;
    \reg[\ALG] \triangleq \sum\nolimits_{\ts\in[T]} \regAt{\ts}[\ALG],
\end{align*}
where $\matchAt{\ts}^{\ALG}$ is a maximum-cardinality matching between the buyers and sellers who accepted their prices in round $\ts$ under algorithm $\ALG$, selected according to the randomized matching rule. The expectation is taken over the randomness of this rule. Similarly, a subscript $\ts$ and a superscript $\ALG$ on any price defined above, as in $\sellerPriceAt{\ts}^{\ALG}$ and $\buyerPriceAt{\ts}^{\ALG}$, indicate the price posted in round $\ts$ under algorithm $\ALG$.

When the algorithm $\ALG$ is clear from the context, we simplify the notation by omitting the explicit dependence on $\ALG$, for example, writing $\reg$ instead of $\reg[\ALG]$. We apply this convention to $\reg$, $\regAt{\ts}$, $\matchAt{\ts}$, $\buyerPriceAt{\ts}$, and $\sellerPriceAt{\ts}$, among other quantities.

\section{Uniform Pricing Suffices in Bilateral Trade and One-to-Many Markets}
\label{sec:uniform-pricing}

In this section, we study dynamic pricing in bilateral trade and one-to-many markets.
In both settings, when all traders' types are known, the optimum GFT benchmark $\OPT$ can be achieved under {\SinglePriceMech}.\footnote{For example, in a market with a single seller and multiple buyers, posting a common price equal to the highest buyer value suffices to achieve the optimum GFT benchmark $\OPT$.} This observation raises a natural question: does {\SinglePriceMech} still suffice to achieve low regret when the principal does not know traders' types and must learn from their accept/reject decisions? We answer this question affirmatively by designing an online learning algorithm using {\SinglePriceMech} that achieves doubly logarithmic regret in one-to-many markets and constant regret in bilateral trade.

In the following, we first present the algorithm and its theoretical regret guarantees. We then explain the intuition behind its key components and sketch the regret analysis. The algorithm and its analysis also lay the groundwork for our approach to general two-sided markets in \Cref{sec:segmented-pricing}.

By symmetry between the two sides, throughout this section we consider a single seller and either a single buyer in bilateral trade or multiple buyers in one-to-many markets. We omit the seller's index and denote the seller's cost by $\sellerValue$. Without loss of generality, we also assume that the highest buyer value $\buyerValueAt{(1)}$ strictly exceeds the seller's cost $\sellerValue$.

\xhdr{Algorithm overview.} Our learning algorithm, {\OneToManyGFT}, consists of two exploration phases followed by an exploitation phase. We describe the goal and procedure of each phase below. A formal description appears in \Cref{alg:GFTfixed}.

\begin{algorithm}[t]
\caption{\OneToManyGFT}
\label[algorithm]{alg:GFTfixed}
\SetAlgoLined
\SetNoFillComment
\DontPrintSemicolon
\KwIn{time horizon $\timeHorizon$}
\KwOut{{\SinglePriceMechs} in one-to-many market for each $\ts \in [\timeHorizon]$}

\tcc{$\feedbackSellerAt{\ts} = (\feedbackSellerAtOf{\ts}{\sellerIndex})_{\sellerIndex \in \sellerSet}$ and $\feedbackBuyerAt{\ts} = (\feedbackBuyerAtOf{\ts}{\buyerIndex})_{\buyerIndex \in \buyerSet}$ are the seller-side and buyer-side indicator vectors representing each trader's accept/reject decision; $\num{\feedbackSellerAt{\ts}}$ and $\num{\feedbackBuyerAt{\ts}}$ are the numbers of accepting sellers and buyers}

Initialize uncertainty interval {$\hvecs_1 = [\LB{\phi}_1,\UB{\phi}_1] \gets [0,1]$ }

\tcc{exploration phase I - identify threshold between $\sellerValueAt{1}$ and $\max_{\buyerIndex\in \buyerSet} \buyerValueAt{\buyerIndex}$}
\For{each round $\ts=1$ to $\timeHorizon$}{

Implement \SinglePriceMech with  $\postPriceAt{\ts}\gets(\LB{\phi}_{\ts}+\UB{\phi}_{\ts})/2$
    
    Observe feedback $\feedbackAt{\ts} = (\feedbackSellerAt{\ts}, \feedbackBuyerAt{\ts})$

    \lIf{$\num{\feedbackSellerAt{\ts}} = \num{\feedbackBuyerAt{\ts}}$}{set exploitation price $\sellerPrice\primed\gets \postPriceAt{\ts}$ and go to the exploitation phase}
    \lElseIf{$\num{\feedbackSellerAt{\ts}} = 0$ and $\num{\feedbackBuyerAt{\ts}} > 0$}{update $\hvecsAt{\ts+1} \gets [({\LB{\phi}_{\ts}+\UB{\phi}_{\ts}})/{2},\UB{\phi}_{\ts}]$}
    \lElseIf{$\num{\feedbackSellerAt{\ts}} = 1$ and $\num{\feedbackBuyerAt{\ts}}=0$}{update $\hvecsAt{\ts+1} \gets [\LB{\phi}_{\ts},({\LB{\phi}_{\ts}+\UB{\phi}_{\ts}})/{2}]$}
    \lElseIf{$\num{\feedbackSellerAt{\ts}} = 1$ and $\num{\feedbackBuyerAt{\ts}} > 1$}{update $\hvecsAt{\ts+1} \gets [({\LB{\phi}_{\ts}+\UB{\phi}_{\ts}})/{2},\UB{\phi}_{\ts}]$, $\tsBreak \gets \ts$ and \KwBreak\algolinelabel{line:GFTfixed-enter-phase-II}}
}

\vspace{2mm}
\tcc{exploration phase II - identify near-optimal \SinglePriceMech}
\For{each round $\ts=\tsBreak +1$ to $\timeHorizon$}{
    Set index $\buyerGapLevelAt{\ts}\gets \lfloor 1+\log\log(\abs{\hvecs_\ts})^{-1}\rfloor$\algolinelabel{line:GFTfixed-index}

\If{$\abs{\hvecs_\ts} < 1/\timeHorizon$\algolinelabel{line:GFTfixed-phase-II-stop}}{
        Set exploitation price $\sellerPrice\primed\gets \LB{\phi}_{\ts}$ and \KwBreak
    }

    Implement \SinglePriceMech with price $\postPriceAt{\ts}\gets\LB{\phi}_{\ts}+2^{-2^{\buyerGapLevelAt{\ts}}}$ in round $\ts$ \label{line:GFTfixed-KL-search}

    Observe feedback $\feedbackAt{\ts} = (\feedbackSellerAt{\ts}, {\feedbackBuyerAt{\ts}})$
      
    \lIf{$\num{\feedbackBuyerAt{\ts}}=0$}{update $\hvecsAt{\ts+1} \gets [\LB{\phi}_{\ts},\LB{\phi}_{\ts}+2^{-2^{\buyerGapLevelAt{\ts}}}]$}
    \lElseIf{$\num{\feedbackBuyerAt{\ts}}=1$}{set exploitation price $\sellerPrice\primed\gets \postPriceAt{\ts}$ and \KwBreak}
    \lElseIf{$\num{\feedbackBuyerAt{\ts}} > 1$}{update $\hvecsAt{\ts+1} \gets [\LB{\phi}_{\ts}+2^{-2^{\buyerGapLevelAt{\ts}}},\UB{\phi}_{\ts}]$\algolinelabel{line:GFTfixed-two-accept}}
    }
\vspace{2mm}
\tcc{exploitation phase}
Implement \SinglePriceMech with price $\sellerPrice\primed$ for all remaining rounds

\end{algorithm}
 
In \underline{exploration phase I}, the algorithm aims to identify a price $\postPrice$ that \emph{induces a trade}. Such a price need not be unique: any price between the seller's cost $\sellerValue$ and the highest buyer value $\buyerValueAt{(1)}$ induces a trade. To find such a price, the algorithm performs a binary search
over an uncertainty interval initialized as $\hvecs_1 = [\LB{\phi}_1,\UB{\phi}_1] \gets [0,1]$. In each round $\ts$, it posts the midpoint price $\postPriceAt{\ts}=(\LB{\phi}_\ts+\UB{\phi}_\ts)/2$ of the current uncertainty interval $\hvecs_\ts = [\LB{\phi}_\ts,\UB{\phi}_\ts]$ and updates the interval to $\hvecs_{\ts+1}$ based on traders' feedback. This phase ends when all $\timeHorizon$ rounds have elapsed or a trade occurs at the current price $\postPriceAt{\ts}$. In the latter case, the algorithm proceeds to the exploration phase II.

In \underline{exploration phase II}, the algorithm \emph{refines the price} to achieve GFT within $1/\timeHorizon$ of the optimum GFT benchmark $\OPT$. To achieve this, it suffices to find a price $\postPrice\geq \sellerValue$ such that either (i) exactly one buyer accepts the price, or (ii) the price lies within $1/\timeHorizon$ below the highest buyer value, i.e., $\buyerValueAt{(1)} - 1/\timeHorizon\leq \postPrice\leq \buyerValueAt{(1)}$. In case (i), the accepting buyer must have the highest value. In case (ii), at least one buyer accepts, and every accepting buyer has a value close to the highest buyer value, so every realized matching achieves near-optimal GFT. In bilateral trade, case (i) already holds when exploration phase I ends with a trade, so exploration phase II requires no additional rounds.

To identify the desired price in exploration phase II, the algorithm continues searching within the uncertainty interval inherited from exploration phase I. Instead of continuing the classic binary search, it switches to conservative binary search (\emph{\ref{eq:KL-search}}\footnote{\ref{eq:KL-search} is named after \citet{KL-03}, who introduced this method for revenue maximization in single-buyer dynamic pricing.} for short), which increases the price from the lower endpoint $\LB{\phi}_{\ts}$ to (line~\ref{line:GFTfixed-KL-search} in \Cref{alg:GFTfixed})
\begin{align}
\label{eq:KL-search}
\tag{\textsc{KL Search}}
    \postPriceAt{\ts}= \LB{\phi}_{\ts} + 2^{-2^{\buyerGapLevelAt{\ts}}}
    \;\;
    \mbox{where}
    \;\;
    \buyerGapLevelAt{\ts} = \bfloor{1+\log \log (\UB{\phi}_{\ts} - \LB{\phi}_{\ts})^{-1}}.
\end{align}
More intuition about \ref{eq:KL-search} will be given at the end of this section. Exploration phase II ends when all $\timeHorizon$ rounds have elapsed, exactly one buyer accepts the price (case (i)), or the uncertainty interval has width less than $1/\timeHorizon$ (case (ii)). If any rounds remain, the algorithm proceeds to the exploitation phase.

In the \underline{exploitation phase}, the algorithm posts the near-optimal price identified in exploration phase II in all remaining rounds.

\xhdr{Theoretical guarantees.} We now present the optimal regret guarantees of {\OneToManyGFT}.

\begin{restatable}{theorem}{thmUniformPricing}
\label{thm:uniform-pricing}
In one-to-many markets, the regret of {\OneToManyGFT} (\Cref{alg:GFTfixed}) using {\SinglePriceMech} is $O(\log\log\timeHorizon)$. In bilateral trade, the regret of {\OneToManyGFT} is at most $1$. 
\end{restatable}

The formal proof of \Cref{thm:uniform-pricing} appears in \Cref{app:uniform-pricing}. Several remarks are in order. First, in bilateral trade, {\OneToManyGFT} achieves constant regret, independent of the time horizon $\timeHorizon$, and is therefore order-optimal in its dependence on $\timeHorizon$ among all online learning algorithms. Second, in one-to-many markets, its regret grows only doubly logarithmically with $\timeHorizon$. As we show in \Cref{thm:mismatch} in \Cref{sec:mismatch-difficulty}, this dependence on the time horizon is also order-optimal among all online learning algorithms using {\TwoPriceMech} (which includes {\SinglePriceMech} as its special case). Thus, even allowing different prices for sellers and buyers does not improve the optimal regret rate of {\OneToManyGFT} in one-to-many markets. Finally, as we explain below, the difference between the two regret rates arises from exploration phase II, which requires additional price refinement in general one-to-many markets but no additional rounds in bilateral trade.

\xhdr{Algorithmic intuition and regret analysis.}
We conclude this section by explaining the intuition behind {\OneToManyGFT} and sketching its regret analysis. 

GFT maximization is easier in bilateral trade than in one-to-many markets because any price that induces a trade achieves the optimum GFT benchmark $\OPT$. Consequently, in bilateral trade, {\OneToManyGFT} incurs regret only while searching for such a price in exploration phase I. The total regret in this phase is at most 1, (and constant even in one-to-many markets). To understand this bound, consider two observations: one about the regret structure and the other about the binary search procedure. First, when no trade occurs, the immediate regret equals the difference between the highest buyer value $\buyerValueAt{(1)}$ and the seller's cost $\sellerValue$. If this difference $\buyerValueAt{(1)}-\sellerValue$ is large, the interval of prices that induce trade is wide, making such a price easier to find. If the difference is small, the immediate regret is correspondingly small.\footnote{This property differs from the regret structure in classical single-buyer dynamic pricing for revenue maximization \citep{KL-03}. Instead, it resembles the regret structure of symmetric loss minimization in contextual search \citep{LS-18}.} Second, the difference $\buyerValueAt{(1)}-\sellerValue$ is always bounded above by the width of the current uncertainty interval $\hvecs_{\ts}$, which starts at one and is halved after each binary search update.

In exploration phase II, the algorithm uses the conservative binary search rule (\ref{eq:KL-search}), which balances the reduction in uncertainty against the immediate regret. If at least two buyers accept, the seller may trade with a buyer other than the highest-value buyer. The immediate regret is small---at most the width of the current uncertainty interval, $\abs{\hvecs_{\ts}}\le 2^{-2^{\buyerGapLevelAt{\ts}-1}}$. Hence, at each value of the index $\buyerGapLevelAt{\ts}$, the number of rounds in which at least two buyers accept is at most $2^{-2^{\buyerGapLevelAt{\ts}-1}}/2^{-2^{\buyerGapLevelAt{\ts}}}=2^{2^{\buyerGapLevelAt{\ts}-1}}$, so the total regret from these rounds is at most $2^{2^{\buyerGapLevelAt{\ts}-1}}\cdot 2^{-2^{\buyerGapLevelAt{\ts}-1}}=1$. If all buyers reject, no trade occurs, and the immediate regret can be large, up to $1$.\footnote{Unlike in exploration phase I, the difference $\buyerValueAt{(1)}-\sellerValue$ in exploration phase II is no longer bounded above by the width of the current uncertainty interval $\hvecs_{\ts}$.} However, the reduction in uncertainty is also substantial: the interval width drops to $2^{-2^{\buyerGapLevelAt{\ts}}}$. For each value of the index $\buyerGapLevelAt{\ts}$, all buyers reject in at most one round. Hence, the total regret at each value of the index $\buyerGapLevelAt{\ts}$ is at most two: one from the rounds in which at least two buyers accept and one from the round in which all buyers reject. Since there are only $O(\log\log \timeHorizon)$ distinct values of the index $\buyerGapLevelAt{\ts}$, summing over these values yields the stated regret bound.\footnote{A similar asymmetric regret structure arises in classical single-buyer dynamic pricing for revenue maximization \citep{KL-03}.}

Combining the regret bounds for the two exploration phases with the exploitation-phase regret of at most 1, which follows from the near-optimality of the price identified during exploration, establishes the regret guarantee in \Cref{thm:uniform-pricing}.

\section{Mismatch and the Limits of Learning in Larger Markets}
\label{sec:mismatch-difficulty}

The preceding section shows that {\OneToManyGFT}, using {\SinglePriceMech}, achieves constant regret in bilateral trade and deteriorates to doubly logarithmic regret in one-to-many markets. This raises two natural questions. First, can a better learning algorithm using {\SinglePriceMech} avoid this deterioration as the market expands? Second, can more flexible pricing schemes overcome this deterioration? In this section, we establish learning barriers that persist even under {\TwoPriceMech}, as formalized in the following theorem.
\begin{restatable}{theorem}{thmMismatch}
    \label{thm:mismatch}
    Consider online learning algorithms using {\TwoPriceMech} in every round. The following regret lower bounds hold:
    \begin{enumerate}[label=(\roman*)]
        \item\label{item:mismatch:one-to-many} \underline{\emph{(Barrier in one-to-many markets)}} There exists a family of one-seller, two-buyer instances such that every online learning algorithm incurs $\Omega(\log\log \timeHorizon)$ regret on some instance in this family.
        \item\label{item:mismatch:many-to-many} \underline{\emph{(Barrier in many-to-many markets)}} There exists a family of three-seller, three-buyer instances such that every online learning algorithm incurs $\Omega(\timeHorizon)$ regret on some instance in this family.
    \end{enumerate}
    Moreover, the optimum GFT benchmark $\OPT$ can be achieved under {\SinglePriceMech} on every instance in these families.
\end{restatable}
Throughout this section, we focus on {\TwoPriceMech}, which allows different prices for sellers and buyers and includes {\SinglePriceMech} as a special case. The formal proof of \Cref{thm:mismatch} appears in \Cref{app:mismatch}. We interpret \Cref{thm:mismatch} as follows.

A central implication is that learning becomes harder as the market expands: the optimal worst-case regret increases from constant in bilateral trade to doubly logarithmic in one-to-many markets and linear in many-to-many markets. Linear worst-case regret means that an algorithm cannot guarantee vanishing average regret; even a scheme that never induces trade incurs regret of the same order. Two qualifications are important. First, both lower bounds hold with only a few traders: one seller and two buyers in \Cref{thm:mismatch}\ref{item:mismatch:one-to-many}, and three traders on each side in \Cref{thm:mismatch}\ref{item:mismatch:many-to-many}. Thus, these barriers arise even when the number of unknown types is fixed, rather than growing with the time horizon $\timeHorizon$. Second, both lower bounds use instances on which the optimum GFT benchmark $\OPT$ is achievable under {\SinglePriceMech}. The barriers therefore do not arise from an inability of {\SinglePriceMech} or {\TwoPriceMech} to implement an efficient outcome under full information. Instead, they arise from the need to learn suitable prices without knowing traders' types.

\Cref{thm:mismatch}\ref{item:mismatch:one-to-many} shows that even when the principal can post different prices to sellers and buyers using {\TwoPriceMech}, every learning algorithm incurs $\Omega(\log\log \timeHorizon)$ regret on some instance. Together with \Cref{thm:uniform-pricing}, this establishes that {\OneToManyGFT} is order-optimal in its dependence on time horizon $\timeHorizon$ in one-to-many markets, even among algorithms using {\TwoPriceMech}.

Both lower bounds in \Cref{thm:mismatch} stem from a common difficulty, which we call the \emph{mismatch phenomenon}. Intuitively, when the numbers of accepting sellers and buyers differ, some accepting traders must remain unmatched. Under the randomized matching rule, less efficient traders may be matched in place of more efficient traders who would participate in an optimal matching attaining $\OPT$. In our hard instances, such displacement can occur with constant probability and cause substantial regret. Under full information, the optimum GFT benchmark $\OPT$ is achievable on these instances by choosing a common price that equalizes the numbers of accepting sellers and buyers. When traders' types are unknown, however, avoiding mismatch during learning becomes difficult.

Although mismatch can arise in both one-to-many and many-to-many markets, its consequences differ substantially. In one-to-many markets, it raises the optimal regret from constant to doubly logarithmic, a theoretically interesting deterioration that nevertheless leaves the dependence on the time horizon mild. In many-to-many markets, it leads to linear worst-case regret even under {\TwoPriceMech}, precluding a guarantee of vanishing average regret. In \Cref{sec:segmented-pricing}, we overcome this barrier by allowing the additional pricing flexibility of {\SegPriceMech}, which enables us to control the regret caused by mismatch and recover doubly logarithmic dependence on the time horizon $\timeHorizon$.

In the remainder of this section, we present the two hard-instance constructions underlying \Cref{thm:mismatch} and explain how each exploits the mismatch phenomenon. In both \Cref{sec:mismatch:one-to-many,sec:mismatch:many-to-many}, we construct hard instances by defining an \emph{adversary} that interacts with any online learning algorithm $\ALG$ using {\TwoPriceMech}. In each round, the adversary observes the two prices posted by $\ALG$ and determines each trader's accept/reject feedback based on these prices and the history of previous rounds. After all $\timeHorizon$ rounds have elapsed, the adversary fixes all traders' types, thereby determining the regret of $\ALG$. The construction must ensure that every feedback response is consistent with these final types. Thus, the interaction corresponds to a single fixed market instance in which traders' types remain constant throughout all rounds.

\subsection{One-to-Many Markets: Constant Regret Is No Longer Achievable}
\label{sec:mismatch:one-to-many}

In this section, we show that the constant-regret guarantee for bilateral trade cannot be extended to one-to-many markets, even when the principal uses {\TwoPriceMech}. The difficulty is due to the mismatch phenomenon. Specifically, in one-to-many market, inducing a trade no longer guarantees efficient trade: when multiple buyers accept, the randomized matching rule may select a buyer whose value is below the highest buyer value. The principal must therefore learn a price that attracts the highest-value buyer while excluding less efficient buyers. We show that this additional requirement leads to the doubly logarithmic regret lower bound in \Cref{thm:mismatch}\ref{item:mismatch:one-to-many}.

Our construction uses a market with one zero-cost seller and two buyers whose unknown values satisfy $1/2 \leq \buyerValueAt{1} < \buyerValueAt{2} \leq 3/4$. The optimum GFT is $\buyerValueAt{2}$, achieved by any common price in $(\buyerValueAt{1},\buyerValueAt{2}]$. If the seller and both buyers accept, the seller trades with each buyer with equal probability, resulting in expected immediate regret $(\buyerValueAt{2}-\buyerValueAt{1})/2$. If both buyers reject, no trade occurs, and the immediate regret is $\buyerValueAt{2}$. Allowing a separate seller-side price does not resolve this tension: it neither distinguishes between the buyers nor changes the randomized choice between them when both accept. Building on this observation, we construct buyer values that force the principal either to incur substantial regret while narrowing the uncertainty interval or to accumulate regret through repeated inefficient trades. We formalize this construction idea below. 

\xhdr{Adversary construction.} We now describe our adversary that interacts with an online algorithm $\ALG$ in the setting introduced above, with one zero-cost seller and two buyers. The adversary determines the buyers' feedback in each round and ultimately selects buyer values $\buyerValueAt{1},\buyerValueAt{2}$ consistent with all feedback provided. A formal description appears in \Cref{adv:GFTfixed}.

\setcounter{adversarycf}{0}
\begin{adversary}[t]
\caption{{Hard instance for
{\TwoPriceMechs} in a one-to-two market}}
\label[adversary]{adv:GFTfixed}
\SetAlgoLined
\SetNoFillComment
\DontPrintSemicolon
\KwIn{online algorithm {\ALG} implementing {\TwoPriceMechs}, sufficiently large time horizon $\timeHorizon$}
\KwOut{seller cost $\sellerValueAt{1}$, buyer values $\buyerValueAt{1},\buyerValueAt{2}$, and traders' per-round feedback}

\tcc{$\hvecsAt{\ts}$ is the uncertainty interval in round $\ts$; $\hvecsPhase{\phaseNum}$ is the phase-start interval in phase $\phaseNum$; $\phaseCount$ is the number of rounds in which both buyers reject}
\vspace{2mm}
Initialize $\sellerValueAt{1}\gets 0$, $\hvecsAt{1}\gets[1/2,3/4]$, $\phaseNum\gets 1$, $\hvecsPhase{1}\gets[1/2,3/4]$, and $\phaseCount\gets 0$

\For{each round $\ts=1$ \KwTo $\timeHorizon$}{
Receive {\TwoPriceMech} with price pair $(\sellerPriceAt{\ts},\buyerPriceAt{\ts})$ from algorithm $\ALG$

    The zero-cost seller accepts if and only if $\sellerPriceAt{\ts}\ge0$

    Set $\priceIncreasingThresholdAt{\phaseNum}\gets \abs{\hvecsPhase{\phaseNum}}^2/\bigl(64\cdot\floor{\frac14\cdot\log_2\log_2\timeHorizon}\bigr)$

\If{$\buyerPriceAt{\ts}>\min\hvecsAt{\ts}+\priceIncreasingThresholdAt{\phaseNum}$}{
Both buyers reject; no trade occurs

        Update $\hvecsAt{\ts+1}\gets[\min\hvecsAt{\ts},\min\hvecsAt{\ts}+\priceIncreasingThresholdAt{\phaseNum}]$

        Set $\hvecsPhase{\phaseNum+1}\gets\hvecsAt{\ts+1}$, $\phaseNum\gets\phaseNum+1$, and $\phaseCount\gets\phaseCount+1$

        \lIf{$\phaseCount=\floor{\frac14\cdot\log_2\log_2\timeHorizon}$}{\KwBreak \tcp*[h]{stop updating by \ref{rule:single-price-tight-reject}}}
    }
    \Else{
Both buyers accept; if the seller also accepts, the randomized matching rule matches the seller to a uniformly random buyer

Update $\hvecsAt{\ts+1}\gets\hvecsAt{\ts}\cap[\buyerPriceAt{\ts},1]$

        \lIf{$\abs{\hvecsAt{\ts+1}}\leq\abs{\hvecsPhase{\phaseNum}}/4$}{\KwBreak \tcp*[h]{stop updating by \ref{rule:single-price-tight-accept}}}
    }
}
For the remaining rounds, set $\buyerValueAt{1}\gets\min\hvecsAt{\ts+1}$ and $\buyerValueAt{2}\gets\max\hvecsAt{\ts+1}$
\end{adversary}

In each round $\ts$, the adversary maintains an \emph{uncertainty interval} $\hvecsAt{\ts}$ containing possible values for both buyers that are consistent with the posted prices and feedback observed before round $\ts$.
Some buyer-value pairs consistent with the observed feedback may lie outside $\hvecsAt{\ts}$, because the adversary need not retain every feasible pair.

The rounds are divided into phases. Phase $\phaseNum$ begins immediately after phase $\phaseNum-1$ ends, with phase $1$ starting at round $1$. Each phase ends when both buyers reject, the adversary stops updating the uncertainty interval, or the time horizon ends.
Let $\hvecsPhase{\phaseNum}$ denote the \emph{phase-start interval}, namely the uncertainty interval at the beginning of phase $\phaseNum$.
We initialize $\hvecsAt{1}=\hvecsPhase{1}=[1/2,3/4]$.
The adversary also maintains a counter $\phaseCount$, initialized to zero, that records the number of rounds in which both buyers reject.
We assume that $\timeHorizon$ is sufficiently large that $\floor{\frac14\cdot\log_2\log_2\timeHorizon}\geq 1$.
For each phase $\phaseNum$, define the \emph{price-increase threshold} $\priceIncreasingThresholdAt{\phaseNum}$ by
\begin{align}
    \label{eq:single-price-tight-threshold}
    \priceIncreasingThresholdAt{\phaseNum}
    \deq \frac{\abs{\hvecsPhase{\phaseNum}}^2}
    {64\cdot\floor{\frac14\cdot\log_2\log_2\timeHorizon}}.
\end{align}
Suppose that round $\ts$ belongs to phase $\phaseNum$.
The algorithm posts a seller-side price $\sellerPriceAt{\ts}$ and a buyer-side price $\buyerPriceAt{\ts}$ satisfying weak budget balance.
The zero-cost seller accepts if and only if $\sellerPriceAt{\ts}\geq 0$.
The adversary determines the buyers' feedback and updates the uncertainty interval and phase index according to the following two rules:
\begin{enumerate}[label=\textbf{Rule~\arabic*.}, ref=Rule~\arabic*, leftmargin=*]
    \item\label{rule:single-price-tight-reject} If $\buyerPriceAt{\ts}>\min\hvecsAt{\ts}+\priceIncreasingThresholdAt{\phaseNum}$, the adversary makes both buyers reject, updates the uncertainty interval to $\hvecsAt{\ts+1}=[\min\hvecsAt{\ts},\min\hvecsAt{\ts}+\priceIncreasingThresholdAt{\phaseNum}]$, and increments the counter $\phaseCount$ by one. If $\phaseCount$ reaches $\floor{\frac14\cdot\log_2\log_2\timeHorizon}$, the adversary stops updating the interval. Otherwise, it begins phase $\phaseNum+1$ with phase-start interval $\hvecsPhase{\phaseNum+1}=\hvecsAt{\ts+1}$, provided that rounds remain.
    \item\label{rule:single-price-tight-accept} If $\buyerPriceAt{\ts}\le\min\hvecsAt{\ts}+\priceIncreasingThresholdAt{\phaseNum}$, the adversary makes both buyers accept and updates the uncertainty interval to $\hvecsAt{\ts+1}=\hvecsAt{\ts}\cap[\buyerPriceAt{\ts},1]$. Once the interval width is at most $\abs{\hvecsPhase{\phaseNum}}/4$, the adversary stops updating the interval. This rule also covers negative buyer-side prices, for which both buyers accept and the interval remains unchanged.
\end{enumerate}
When either stopping condition is met or the time horizon ends, the adversary fixes the two buyers' values at the endpoints of the resulting interval. Specifically, if the final update occurs after round $\ts$, it sets $\buyerValueAt{1}=\min\hvecsAt{\ts+1}$ and $\buyerValueAt{2}=\max\hvecsAt{\ts+1}$. If any rounds remain, the buyers respond according to these fixed values. The construction ensures that these values are consistent with all preceding feedback, so the entire interaction corresponds to a single fixed market instance.

\xhdr{Theoretical guarantee and proof sketch.} We now state the regret lower bound established by \Cref{adv:GFTfixed} and explain how the mismatch phenomenon drives the construction.

\begin{restatable}[Restatement of \Cref{thm:mismatch}\ref{item:mismatch:one-to-many}]{lemma}{lemMismatchOneToMany}
    Given any online algorithm $\ALG$ using {\TwoPriceMech}, its regret on the one-seller two-buyer instance generated by \Cref{adv:GFTfixed} is $\Omega(\log\log \timeHorizon)$.
\end{restatable}

The key difficulty preventing constant regret is that inducing trade does not necessarily eliminate regret. Whenever the seller and both buyers accept (\ref{rule:single-price-tight-accept} in \Cref{adv:GFTfixed}), the randomized matching rule selects each buyer with probability one-half, resulting in expected immediate regret $(\buyerValueAt{2}-\buyerValueAt{1})/2$. In contrast, whenever both buyers reject (\ref{rule:single-price-tight-reject} in \Cref{adv:GFTfixed}), the immediate regret is $\buyerValueAt{2}\geq 1/2$.

The adversary exploits this regret structure by controlling how quickly the algorithm can narrow the uncertainty interval. An aggressive buyer-side price makes both buyers reject, incurring constant regret. A more conservative price makes both buyers accept but increases the lower endpoint of the uncertainty interval by at most the price-increase threshold. Such rounds reveal information slowly while exposing the algorithm to mismatch. There are three possible cases. First, if the counter $\phaseCount$ reaches $\totalPhaseNum\triangleq \floor{\frac14\cdot\log_2\log_2\timeHorizon}$, there have been $\totalPhaseNum$ rounds in which both buyers reject. Each contributes regret of at least one-half, giving total regret of at least $\totalPhaseNum/2=\Theta(\log\log\timeHorizon)$. Second, suppose the adversary stops because the interval width falls to at most $\abs{\hvecsPhase{\phaseNum}}/4$. By Eqn.~\eqref{eq:single-price-tight-threshold} and \ref{rule:single-price-tight-accept}, each update in this phase increases the lower endpoint by at most $\priceIncreasingThresholdAt{\phaseNum}$. Thus, the phase contains at least $3\abs{\hvecsPhase{\phaseNum}}/(4\priceIncreasingThresholdAt{\phaseNum})$ rounds in which both buyers accept. The stopping rule also ensures that the final buyer-value gap remains sufficiently large. Consequently, the cumulative expected regret is bounded by $3\abs{\hvecsPhase{\phaseNum}}/(4\priceIncreasingThresholdAt{\phaseNum})\cdot (\abs{\hvecsPhase{\phaseNum}}/4-\priceIncreasingThresholdAt{\phaseNum})/2=\Theta(\log\log\timeHorizon)$. Finally, if the horizon ends before either stopping condition is met, the final uncertainty interval still has width $\Omega(\sqrt{1/\timeHorizon})$, and both buyers accept in at least $\timeHorizon-\totalPhaseNum$ rounds. These rounds yield cumulative expected regret $\Omega(\sqrt{\timeHorizon})$, which exceeds the required doubly logarithmic lower bound for sufficiently large horizons.

Our construction shares the phase-based information--regret tradeoff underlying the doubly logarithmic lower bound of \citet{KL-03}, but requires a different treatment of trading rounds. In their fixed-valuation setting, an accepted offer incurs revenue regret equal to the buyer's value minus the posted price. In our setting, a price that induces trade only with the higher-value buyer achieves zero regret, whereas acceptance by the seller and both buyers incurs expected mismatch regret equal to half the buyer-value gap. Thus, uncertainty about the location of the values alone does not ensure substantial regret: nearly equal values make mismatch nearly costless. Our construction addresses this difficulty by selecting the final interval endpoints as buyer values and using the stopping rule to preserve a sufficiently large gap, while maintaining consistency with all preceding feedback.

\subsection{Many-to-Many Markets: Sublinear Regret Is No Longer Achievable}

\label{sec:mismatch:many-to-many}

In this section, we show that the learning barrier becomes more severe when both sides contain multiple traders. In one-to-many markets, {\OneToManyGFT} can limit the loss from selecting a lower-value buyer by narrowing the uncertainty interval. In many-to-many markets, however, uncertainty about a narrow interval of marginal trader types can cause the displacement of traders whose contributions to GFT remain large. We exploit this distinction to establish the linear regret lower bound in \Cref{thm:mismatch}\ref{item:mismatch:many-to-many}.

Our construction uses three sellers and three buyers. One seller has cost $\sellerValueAt{1}=0$, two buyers have values $\buyerValueAt{1}=\buyerValueAt{2}=1$, and the remaining types satisfy $1/4 \leq \sellerValueAt{2}=\sellerValueAt{3}<\buyerValueAt{3}\leq 3/4$. A common price in $[\sellerValueAt{3},\buyerValueAt{3}]$ induces all three trades and achieves the optimum GFT. Outside this critical interval, side-specific pricing cannot avoid a constant loss: either too few sellers participate, leaving beneficial trades unrealized, or too few buyers participate, allowing the randomized matching rule to leave the zero-cost seller unmatched. We construct the unknown types so that, under weak budget balance, the principal never posts a seller-side and buyer-side price pair that induces all traders to accept. The resulting expected regret is bounded below by a positive constant in every round, yielding linear cumulative regret despite the existence of an efficient common price.

\xhdr{Adversary construction and theoretical guarantee.}
We now describe an adversary that interacts with an online algorithm $\ALG$ using {\TwoPriceMech} in a market with three sellers and three buyers. 
The adversary fixes $\sellerValueAt{1}=0$ and $\buyerValueAt{1}=\buyerValueAt{2}=1$, and maintains an uncertainty interval $\hvecsAt{\ts}$, initialized to $\hvecsAt{1}=[1/4,3/4]$, for the common cost of sellers 2 and 3 and the value of buyer 3. A formal description appears in \Cref{adv:GFTtwoprice}.

\setcounter{adversarycf}{1}
\begin{adversary}[t]
\caption{{Hard instance for
{\TwoPriceMechs} in a three-to-three market}}
\label[adversary]{adv:GFTtwoprice}
\SetAlgoLined
\SetNoFillComment
\KwIn{{online algorithm {\ALG} implementing {\TwoPriceMechs}}, time horizon~\timeHorizon}
\KwOut{seller costs $\sellerValueAt{1}, \sellerValueAt{2},\sellerValueAt{3}$, buyers values $\buyerValueAt{1},\buyerValueAt{2},\buyerValueAt{3}$, and traders' per-round feedback}

\vspace{2mm}
Initialize {$\UB{\sellerPrice} \gets {1}/{4}$, $\LB{\buyerPrice}\gets {3}/{4}$}

\For{each round $\ts=1$ to $\timeHorizon$}{
    Receive {\TwoPriceMech} with price pair $(\sellerPriceAt{\ts}, \buyerPriceAt{\ts})$ from algorithm $\ALG$

    \If{$\sellerPriceAt{\ts}< \LB{\buyerPrice}$}{
    Only the seller $\sellerValueAt{1}$ accepts the price and update $\UB{\sellerPrice} =\max\{\UB{\sellerPrice}, \sellerPriceAt{\ts}\}$
    }
     \Else{
     All three sellers accept the price
     }

     \If{$\buyerPriceAt{\ts} > \UB{\sellerPrice}$}{
    Only two buyers $\buyerValueAt{1}$ and $\buyerValueAt{2}$ accept the price and update $\LB{\buyerPrice} =\min\{\LB{\buyerPrice}, \buyerPriceAt{\ts}\}$
     }
     \Else{
     All three buyers accept the price
     }
}
Set
$\sellerValueAt{1}\gets 0$, $\buyerValueAt{1}=\buyerValueAt{2}\gets 1$,
$\sellerValueAt{2}=\sellerValueAt{3} \gets {(2\UB{\sellerPrice}+\LB{\buyerPrice})}/{3}$ and $ \buyerValueAt{3}\gets{(\UB{\sellerPrice}+2\LB{\buyerPrice})}/{3}$.
\end{adversary} 
In round $\ts$, the algorithm posts seller-side price $\sellerPriceAt{\ts}$ and buyer-side price $\buyerPriceAt{\ts}$ with $\sellerPriceAt{\ts}\leq \buyerPriceAt{\ts}$. Seller 1 and buyers 1 and 2 respond according to their fixed types. The adversary initializes $\hvecsAt{\ts+1}\gets\hvecsAt{\ts}$ and determines the remaining traders' feedback through two sequential updates:
First, if seller-side price $\sellerPriceAt{\ts}<\max\hvecsAt{\ts}$, sellers 2 and 3 reject, and the adversary updates
$\hvecsAt{\ts+1}\gets
    [\max\{\min\hvecsAt{\ts},\sellerPriceAt{\ts}\},\max\hvecsAt{\ts}]$.
Otherwise, both sellers accept and the interval remains unchanged.
Next, using the updated interval, if buyer-side price $\buyerPriceAt{\ts}>\min\hvecsAt{\ts+1}$, buyer 3 rejects, and the adversary updates $\hvecsAt{\ts+1}\gets
[\min\hvecsAt{\ts+1},\min\{\max\hvecsAt{\ts+1},\buyerPriceAt{\ts}\}]$. Otherwise, buyer 3 accepts and the interval remains unchanged.
These rules preserve a nested interval of positive width. After all $\timeHorizon$ rounds, the adversary sets
$\sellerValueAt{2}=\sellerValueAt{3}=({2\min\hvecsAt{\timeHorizon+1}+\max\hvecsAt{\timeHorizon+1}})/{3}$ and $\buyerValueAt{3}=({\min\hvecsAt{\timeHorizon+1}+2\max\hvecsAt{\timeHorizon+1}})/{3}$.
The resulting types are consistent with all preceding feedback and therefore define a single fixed market instance.

\Cref{adv:GFTtwoprice} prevents all six traders from accepting simultaneously: if sellers 2 and 3 accept, the seller-side price $\sellerPriceAt{\ts}$ is at least $\max\hvecsAt{\ts}$, so the buyer-side price $\buyerPriceAt{\ts}\geq\sellerPriceAt{\ts}$ must exceed the lower endpoint of the interval, causing buyer 3 to reject. If at most one seller accepts, the realized GFT is at most 1, whereas the optimum GFT is $2+\buyerValueAt{3}-2\sellerValueAt{2}$, yielding regret of at least $1/4$. If all three sellers accept and trade occurs, only buyers 1 and 2 accept, and the randomized matching rule leaves the zero-cost seller unmatched with probability one-third. The resulting expected regret is $\buyerValueAt{3}-\frac23\sellerValueAt{2}\geq 1/12$. Rounds without trade also incur regret of at least $1/12$. Thus, the expected regret in every round is at least $1/12$, giving cumulative regret of at least $\timeHorizon/12$.

\begin{restatable}[Restatement of \Cref{thm:mismatch}\ref{item:mismatch:many-to-many}]{lemma}{lemMismatchManyToMany}
    Given any online algorithm $\ALG$ using {\TwoPriceMech}, its regret on the three-seller three-buyer instance generated by \Cref{adv:GFTtwoprice} is $\Omega(\timeHorizon)$.
\end{restatable}

\section{Overcoming Mismatch through Segmented Pricing}
\label{sec:segmented-pricing}

The preceding section shows that even {\TwoPriceMech} cannot guarantee sublinear regret in many-to-many markets. Although an efficient uniform price exists on the hard instances, searching for it can repeatedly displace efficient traders and incur substantial regret. In this section, we show that limited price differentiation within each side of the market overcomes this barrier. Specifically, {\SegPriceMech}, which partitions each side into \emph{at most two groups} and assigns a price to each group, suffices to recover doubly logarithmic regret in the time horizon. We formalize this result in the following theorem. To simplify notation, for a market with $\sellerNum$ sellers and $\buyerNum$ buyers, we write $\totalNum\triangleq\sellerNum+\buyerNum$ for the total number of traders and $\maxTradeNum\triangleq\min\{\sellerNum,\buyerNum\}$ for the maximum possible number of trades in a round. 

\begin{restatable}{theorem}{thmSegmentedPricing}
\label{thm:segmented-pricing}
In a two-sided market with $\sellerNum$ sellers and $\buyerNum$ buyers, the regret of {\ConservativeSegmentedPricing} (\Cref{alg:conservative-segmented-pricing}) using {\SegPriceMech} is $O((\sellerNum+\buyerNum)\log\log\timeHorizon
+(\sellerNum+\buyerNum)\min\{\sellerNum,\buyerNum\})$. Moreover, every online learning algorithm using {\SegPriceMech} incurs $\Omega((\sellerNum+\buyerNum)\log\log\timeHorizon)$ regret on some instance.
\end{restatable}

We interpret \Cref{thm:segmented-pricing} as follows.
A central implication is that limited price differentiation can improve market efficiency when the platform must learn traders' types. Even when the platform seeks to maximize GFT rather than its own profit, the flexibility to differentiate prices within each side can be valuable. Specifically, {\SegPriceMech} enables vanishing average regret for any fixed market size, whereas pricing schemes that assign a common price to all traders on each side, including {\SinglePriceMech} and {\TwoPriceMech}, face a linear worst-case regret barrier in many-to-many markets. At a high level, this flexibility allows the platform to guarantee a minimum number of trades through a common price while simultaneously offering exploratory prices to learn about the remaining traders.

Like {\OneToManyGFT} in one-to-many markets, {\ConservativeSegmentedPricing} achieves regret that grows only doubly logarithmically with the time horizon $\timeHorizon$. This dependence is order-optimal in both $\timeHorizon$ and $\totalNum$ among all online learning algorithms using {\SegPriceMech}.

In the remainder of this section, we present the details of {\ConservativeSegmentedPricing} and explain the intuition behind its key components in \Cref{sec:segmented-pricing:algorithm}. We then provide its regret analysis in \Cref{sec:segmented-pricing:analysis}. The lower bound of \Cref{thm:segmented-pricing} is stated and proved in \Cref{app:segmented-pricing-lower-bound}.

\subsection{Algorithm Description}
\label{sec:segmented-pricing:algorithm}

Our learning algorithm, {\ConservativeSegmentedPricing}, maintains a separate uncertainty interval for each trader and applies the same pricing and update rules in every round. At a high level, it uses these intervals to compute a pessimistic lower bound on the efficient trade size. It then selects that many traders on each side, referred to as \emph{selected traders}, who are guaranteed to accept a common price. The remaining traders, referred to as \emph{outside traders}, receive an exploratory price on each side to help reduce uncertainty about their types. Below, we explain how the algorithm computes the pessimistic trade size, selects the groups, determines the segmented prices, and updates the uncertainty intervals. A formal description appears in \Cref{alg:conservative-segmented-pricing}.

\begin{algorithm}
\caption{\ConservativeSegmentedPricing}
\label{alg:conservative-segmented-pricing}
\SetAlgoLined
\SetNoFillComment
\DontPrintSemicolon
\KwIn{Time horizon $\timeHorizon$; seller and buyer sets $\sellerSet,\buyerSet$}
\KwOut{{\SegPriceMech} in each round $\ts\in[\timeHorizon]$}

Initialize $[\sellerValueLbOfAt{\sellerIndex}{1},\sellerValueUbOfAt{\sellerIndex}{1}]\gets[0,1]$ for all $\sellerIndex\in\sellerSet$ and
$[\buyerValueLbOfAt{\buyerIndex}{1},\buyerValueUbOfAt{\buyerIndex}{1}]\gets[0,1]$ for all $\buyerIndex\in\buyerSet$\;
Set $\maxTradeNum\gets\min\{\abs{\sellerSet},\abs{\buyerSet}\}$ and $\segPrecision\gets1/(4\timeHorizon^2)$\;
Fix an order of trader identities on each side for tie breaking\;

\For{each round $\ts=1$ to $\timeHorizon$}{
    \tcp{Compute the pessimistic trade size}
    Sort seller upper bounds increasingly and buyer lower bounds decreasingly, using the fixed orders to break ties\;
    Compute $\selectedTradeNumAt{\ts}$ by Eqn.~\eqref{eq:segmented-selected-size}\;

    \tcp{Select the trading groups}
    Let $\sellerSubsetAt{1}{\ts}$ and $\buyerSubsetAt{1}{\ts}$ contain the first $\selectedTradeNumAt{\ts}$ sellers and buyers in their respective sorted lists\;
    Set $\sellerSubsetAt{2}{\ts}\gets\sellerSet\setminus\sellerSubsetAt{1}{\ts}$ and $\buyerSubsetAt{2}{\ts}\gets\buyerSet\setminus\buyerSubsetAt{1}{\ts}$\;

    \tcp{Determine the segmented prices}
    Compute $\sellerIntGapOfAt{\sellerIndex}{\ts}\gets
        \sellerValueUbOfAt{\sellerIndex}{\ts}
        -\sellerValueLbOfAt{\sellerIndex}{\ts}$ for all $\sellerIndex\in\sellerSet$
    and $\buyerIntGapOfAt{\buyerIndex}{\ts}\gets
        \buyerValueUbOfAt{\buyerIndex}{\ts}
        -\buyerValueLbOfAt{\buyerIndex}{\ts}$ for all $\buyerIndex\in\buyerSet$\;

    Set $\sellerExplorePriceOfAt{\sellerIndex}{\ts}\gets
        \max\left\{0,\sellerValueLbOfAt{\sellerIndex}{\ts}
        +\tfrac12\left(\sellerIntGapOfAt{\sellerIndex}{\ts}\right)^{3/2}
        -\segPrecision\right\}$ for all $\sellerIndex\in\sellerSet$\;
    Set $\buyerExplorePriceOfAt{\buyerIndex}{\ts}\gets
        \min\left\{1,\buyerValueUbOfAt{\buyerIndex}{\ts}
        -\tfrac12\left(\buyerIntGapOfAt{\buyerIndex}{\ts}\right)^{3/2}
        +\segPrecision\right\}$ for all $\buyerIndex\in\buyerSet$\;

    Set $\sellerExplorePriceAt{\ts}\gets
        \min_{\sellerIndex\in\sellerSubsetAt{2}{\ts}}
        \sellerExplorePriceOfAt{\sellerIndex}{\ts}$ if
        $\sellerSubsetAt{2}{\ts}\neq\emptyset$, and
        $\sellerExplorePriceAt{\ts}\gets1$ otherwise\;
    Set $\buyerExplorePriceAt{\ts}\gets
        \max_{\buyerIndex\in\buyerSubsetAt{2}{\ts}}
        \buyerExplorePriceOfAt{\buyerIndex}{\ts}$ if
        $\buyerSubsetAt{2}{\ts}\neq\emptyset$, and
        $\buyerExplorePriceAt{\ts}\gets0$ otherwise\;

    Set $(\commonPriceLbAt{\ts},\commonPriceUbAt{\ts})\gets(0,1)$\;
    \If{$\selectedTradeNumAt{\ts}>0$}{
        Set $(\commonPriceLbAt{\ts},\commonPriceUbAt{\ts})\gets
        (\sellerValueOrdUbOfAt{\selectedTradeNumAt{\ts}}{\ts},
         \buyerValueOrdLbOfAt{\selectedTradeNumAt{\ts}}{\ts})$\;
    }
    Set $\commonPriceAt{\ts}\gets
        \min\left\{\commonPriceUbAt{\ts},
        \max\left\{\commonPriceLbAt{\ts},
        (\sellerExplorePriceAt{\ts}+\buyerExplorePriceAt{\ts})/2
        \right\}\right\}$\;

    Post $\sellerPriceAtOfseg{\ts}{1}
        =\buyerPriceAtOfseg{\ts}{1}
        =\commonPriceAt{\ts}$ to all selected sellers and buyers\;
    Post $\sellerPriceAtOfseg{\ts}{2}
        =\min\{\sellerExplorePriceAt{\ts},\commonPriceAt{\ts}\}$
        to all sellers in $\sellerSubsetAt{2}{\ts}$\;
    Post $\buyerPriceAtOfseg{\ts}{2}
        =\max\{\buyerExplorePriceAt{\ts},\commonPriceAt{\ts}\}$
        to all buyers in $\buyerSubsetAt{2}{\ts}$\;

    \tcp{Update the uncertainty intervals}
    Observe all traders' accept/reject decisions and copy all current interval endpoints to round $\ts+1$\;
    \ForEach{seller $\sellerIndex\in\sellerSet$ receiving price $\offeredPrice$}{
        \lIf{seller $\sellerIndex$ accepts}{
            set $\sellerValueUbOfAt{\sellerIndex}{\ts+1}\gets
            \min\{\sellerValueUbOfAt{\sellerIndex}{\ts},\offeredPrice\}$
        }
        \lElse{
            set $\sellerValueLbOfAt{\sellerIndex}{\ts+1}\gets
            \max\{\sellerValueLbOfAt{\sellerIndex}{\ts},\offeredPrice\}$
        }
    }
    \ForEach{buyer $\buyerIndex\in\buyerSet$ receiving price $\offeredPrice$}{
        \lIf{buyer $\buyerIndex$ accepts}{
            set $\buyerValueLbOfAt{\buyerIndex}{\ts+1}\gets
            \max\{\buyerValueLbOfAt{\buyerIndex}{\ts},\offeredPrice\}$
        }
        \lElse{
            set $\buyerValueUbOfAt{\buyerIndex}{\ts+1}\gets
            \min\{\buyerValueUbOfAt{\buyerIndex}{\ts},\offeredPrice\}$
        }
    }
}
\end{algorithm} 
At the beginning of round $\ts$, the algorithm maintains an uncertainty interval $[\sellerValueLbOfAt{\sellerIndex}{\ts},\sellerValueUbOfAt{\sellerIndex}{\ts}]$ for every seller $\sellerIndex\in\sellerSet$ and $[\buyerValueLbOfAt{\buyerIndex}{\ts},\buyerValueUbOfAt{\buyerIndex}{\ts}]$ for every buyer $\buyerIndex\in\buyerSet$. All intervals are initialized to $[0,1]$.

\xhdr{Computing the pessimistic trade size.}
The algorithm computes a pessimistic lower bound on the efficient trade size by treating seller upper bounds as costs and buyer lower bounds as values. Sort the seller upper bounds in increasing order and the buyer lower bounds in decreasing order, breaking ties according to fixed orders of trader identities. Let $\sellerValueOrdUbOfAt{\tradeRank}{\ts}$ and $\buyerValueOrdLbOfAt{\tradeRank}{\ts}$ denote the corresponding order statistics. These orderings depend on the current bounds and may change across rounds. Define pessimistic trade size $\selectedTradeNumAt{\ts}$ as 
\begin{align}
    \selectedTradeNumAt{\ts}\triangleq
    \max\Bigl(\{0\}\cup
    \left\{\tradeRank\in[\maxTradeNum]:
        \sellerValueOrdUbOfAt{\tradeRank}{\ts}
        \leq\buyerValueOrdLbOfAt{\tradeRank}{\ts}\right\}\Bigr).
    \label{eq:segmented-selected-size}
\end{align}
Because these bounds overestimate sellers' costs and underestimate buyers' values, the resulting trade size is a lower bound on the true efficient trade size.

\xhdr{Selecting the trading groups.}
Let $\sellerSubsetAt{1}{\ts}$ contain the first $\selectedTradeNumAt{\ts}$ sellers in the sorted seller list and $\buyerSubsetAt{1}{\ts}$ the first $\selectedTradeNumAt{\ts}$ buyers in the sorted buyer list. The remaining sellers and buyers form the second segments $\sellerSubsetAt{2}{\ts}$ and $\buyerSubsetAt{2}{\ts}$, respectively. We refer to traders in the first segments as \emph{selected traders} and traders in the second segments as \emph{outside traders}.
For any given number of sellers, choosing those with the smallest upper bounds minimizes the price needed to guarantee acceptance. Symmetrically, choosing buyers with the largest lower bounds maximizes the price they are guaranteed to accept. Consequently, the pessimistic trade size $\selectedTradeNumAt{\ts}$ gives the largest common group size for which acceptance at a common price can be guaranteed.

\xhdr{Determining the segmented prices.}
The algorithm assigns a common price to the selected traders and an exploratory price to the outside traders on each side.

We first construct the exploratory price proposals. Define the uncertainty interval widths by $\sellerIntGapOfAt{\sellerIndex}{\ts}
\triangleq\sellerValueUbOfAt{\sellerIndex}{\ts}-\sellerValueLbOfAt{\sellerIndex}{\ts},
\buyerIntGapOfAt{\buyerIndex}{\ts}
\triangleq\buyerValueUbOfAt{\buyerIndex}{\ts}-\buyerValueLbOfAt{\buyerIndex}{\ts}$.
Set $\segPrecision\triangleq 1/(4\timeHorizon^2)$. The algorithm defines an \emph{individual exploratory price} for each seller and buyer, respectively, by
\begin{align*}
    \sellerExplorePriceOfAt{\sellerIndex}{\ts}
    \triangleq\max\left\{0,\sellerValueLbOfAt{\sellerIndex}{\ts}
        +\tfrac12\left(\sellerIntGapOfAt{\sellerIndex}{\ts}\right)^{3/2}-\segPrecision\right\}
        \;\;\mbox{and}\;\;
    \buyerExplorePriceOfAt{\buyerIndex}{\ts}
    \triangleq\min\left\{1,\buyerValueUbOfAt{\buyerIndex}{\ts}
        -\tfrac12\left(\buyerIntGapOfAt{\buyerIndex}{\ts}\right)^{3/2}+\segPrecision\right\}.
\end{align*}
Each seller's individual exploratory price lies near the lower cost bound, whereas each buyer's lies near the upper value bound. Acceptance at such a price can therefore sharply reduce the corresponding uncertainty interval.
For the outside traders, the algorithm takes the smallest individual seller price as the seller-side proposal and the largest individual buyer price as the buyer-side proposal:
\begin{align*}
    \sellerExplorePriceAt{\ts}
    \triangleq\min\nolimits_{\sellerIndex\in\sellerSubsetAt{2}{\ts}}
        \sellerExplorePriceOfAt{\sellerIndex}{\ts},
    \;\;\mbox{and}\;\;
    \buyerExplorePriceAt{\ts}
    \triangleq\max\nolimits_{\buyerIndex\in\buyerSubsetAt{2}{\ts}}
        \buyerExplorePriceOfAt{\buyerIndex}{\ts}.
\end{align*}

The common price $\commonPriceAt{\ts}$ offered to all selected traders is the midpoint of the two outside-group price proposals $\sellerExplorePriceAt{\ts}$ and $\buyerExplorePriceAt{\ts}$, clipped to the interval $[\sellerValueOrdUbOfAt{\selectedTradeNumAt{\ts}}{\ts}, \buyerValueOrdLbOfAt{\selectedTradeNumAt{\ts}}{\ts}]$, i.e.,
\begin{align}
\commonPriceAt{\ts}\triangleq
\min\left\{\buyerValueOrdLbOfAt{\selectedTradeNumAt{\ts}}{\ts},
\max\left\{\sellerValueOrdUbOfAt{\selectedTradeNumAt{\ts}}{\ts},
(\sellerExplorePriceAt{\ts}+\buyerExplorePriceAt{\ts})/2\right\}\right\}.
\label{eq:segmented-common-price}
\end{align}
Finally, the algorithm posts the common price $\commonPriceAt{\ts}$ to all selected sellers and buyers, whose acceptance is guaranteed by construction. For the outside traders, it caps the seller proposal at $\commonPriceAt{\ts}$ and floors the buyer proposal at $\commonPriceAt{\ts}$, ensuring weak budget balance:
\begin{align*}
    \sellerPriceAtOfseg{\ts}{1}
    =\buyerPriceAtOfseg{\ts}{1}
    &\triangleq\commonPriceAt{\ts},
    &\sellerPriceAtOfseg{\ts}{2}
    &\triangleq\min\{\sellerExplorePriceAt{\ts},\commonPriceAt{\ts}\},
    &\buyerPriceAtOfseg{\ts}{2}
    &\triangleq\max\{\buyerExplorePriceAt{\ts},\commonPriceAt{\ts}\}.
\end{align*}

\Cref{fig:segmented-price-selection} illustrates the price selection for four sellers and four buyers, with $\selectedTradeNumAt{\ts}=2$. In this example, the midpoint of the two outside proposals lies in the interval guaranteeing acceptance by all selected traders, so the clipping operations leave the midpoint and both outside proposals unchanged.

\begin{figure}
    \centering
    \begin{tikzpicture}[x=5.2cm,y=0.72cm,font=\small,
        price interval/.style={draw=black!65,line width=1pt},
        price cap/.style={draw=black!65,line width=0.7pt},
        price guide/.style={draw=black!45,densely dashed,line width=0.5pt}]
\foreach \shift in {0,8} {
            \begin{scope}[xshift=\shift cm]
                \fill[black!5] (0.32,-0.4) rectangle (0.85,3.95);
                \draw[black!25,line width=0.35pt] (0.32,-0.4) -- (0.32,3.95);
                \draw[black!25,line width=0.35pt] (0.85,-0.4) -- (0.85,3.95);
                \draw[price guide] (0.445,2.45) -- (0.445,3.95);
                \node[above=2pt] at (0.445,3.95) {$\commonPrice$};
                \draw[black!25,densely dotted] (0,1.95) -- (1,1.95);
                \draw[->,>=stealth,black!60,line width=0.5pt]
                    (-0.025,-0.4) -- (1.045,-0.4);
                \foreach \x in {0,1} {
                    \draw[black!55] (\x,-0.45) -- (\x,-0.35);
                    \node[below=3pt] at (\x,-0.4) {$\x$};
                }
                \foreach \x in {0.32,0.85}
                    \draw[black!55] (\x,-0.45) -- (\x,-0.35);
                \node[below=3pt] at (0.32,-0.4)
                    {$\sellerValueOrdUb{\selectedTradeNum}$};
                \node[below=3pt] at (0.85,-0.4)
                    {$\buyerValueOrdLb{\selectedTradeNum}$};
                \foreach \trader/\row in {1/3.5,2/2.7,3/1.25,4/0.45}
                    \node[anchor=east] at (-0.025,\row) {$\trader$};
                \draw[black!60] (-0.10,3.72) -- (-0.12,3.72)
                    -- (-0.12,2.48) -- (-0.10,2.48);
                \draw[black!60] (-0.10,1.47) -- (-0.12,1.47)
                    -- (-0.12,0.23) -- (-0.10,0.23);
            \end{scope}
        }

        \node[font=\small\bfseries] at (0.5,5.35) {Sellers: cost intervals};
        \node[anchor=east] at (-0.15,3.1) {$\sellerSubset{1}$};
        \node[anchor=east] at (-0.15,0.85) {$\sellerSubset{2}$};
        \foreach \lo/\hi/\row in {0.04/0.2/3.5,0.12/0.32/2.7,0.02/0.42/1.25,0.38/0.9/0.45} {
            \draw[price interval] (\lo,\row) -- (\hi,\row);
            \draw[price cap] (\lo,{\row-0.09}) -- (\lo,{\row+0.09});
            \draw[price cap] (\hi,{\row-0.09}) -- (\hi,{\row+0.09});
        }
        \draw[price guide] (0.15,0.2) -- (0.15,1.5);
        \foreach \x/\row in {0.15/1.25,0.57/0.45}
            \filldraw[fill=white,draw=black,line width=0.7pt]
                (\x,\row) circle[radius=3.3pt];
        \foreach \x/\row in {0.445/3.5,0.445/2.7,0.15/1.25,0.15/0.45}
            \fill (\x,\row) circle[radius=1.8pt];
        \node[above=2pt] at (0.15,1.5) {$\sellerTwoPrices{2}$};
        \node at (0.5,-1.7)
            {$\sellerTwoPrices{2}=\sellerExplorePrice
              =\min\{\sellerExplorePriceOf{3},\sellerExplorePriceOf{4}\}$};

        \begin{scope}[xshift=8cm]
            \node[font=\small\bfseries] at (0.5,5.35) {Buyers: value intervals};
            \node[anchor=east] at (-0.15,3.1) {$\buyerSubset{1}$};
            \node[anchor=east] at (-0.15,0.85) {$\buyerSubset{2}$};
            \foreach \lo/\hi/\row in {0.91/0.99/3.5,0.85/0.98/2.7,0.4/0.94/1.25,0.1/0.6/0.45} {
                \draw[price interval] (\lo,\row) -- (\hi,\row);
                \draw[price cap] (\lo,{\row-0.09}) -- (\lo,{\row+0.09});
                \draw[price cap] (\hi,{\row-0.09}) -- (\hi,{\row+0.09});
            }
            \draw[price guide] (0.74,0.2) -- (0.74,1.5);
            \foreach \x/\row in {0.74/1.25,0.42/0.45}
                \filldraw[fill=white,draw=black,line width=0.7pt]
                    (\x,\row) circle[radius=3.3pt];
            \foreach \x/\row in {0.445/3.5,0.445/2.7,0.74/1.25,0.74/0.45}
                \fill (\x,\row) circle[radius=1.8pt];
            \node[above=2pt] at (0.74,1.5) {$\buyerTwoPrices{2}$};
            \node at (0.5,-1.7)
                {$\buyerTwoPrices{2}=\buyerExplorePrice
                  =\max\{\buyerExplorePriceOf{3},\buyerExplorePriceOf{4}\}$};
        \end{scope}

        \draw[price interval] (0.4cm,-1.95cm) -- (0.95cm,-1.95cm);
        \draw[price cap] (0.4cm,-2.01cm) -- (0.4cm,-1.89cm);
        \draw[price cap] (0.95cm,-2.01cm) -- (0.95cm,-1.89cm);
        \node[anchor=west] at (1.05cm,-1.95cm) {Uncertainty interval};
        \filldraw[fill=white,draw=black,line width=0.7pt]
            (4.8cm,-1.95cm) circle[radius=3.3pt];
        \node[anchor=west] at (5.0cm,-1.95cm) {Individual exploratory price};
        \fill (10.2cm,-1.95cm) circle[radius=1.8pt];
        \node[anchor=west] at (10.4cm,-1.95cm) {Posted price};
    \end{tikzpicture}
    \caption{Price selection with
        $\protect\selectedTradeNum=2$ and $\commonPrice=(\sellerExplorePrice+\buyerExplorePrice)/2$.
        Hollow markers show their individual exploratory prices, and filled markers show the prices actually posted.
        A filled center in a hollow marker indicates that the two prices coincide.}
    \label{fig:segmented-price-selection}
\end{figure}

\xhdr{Updating the uncertainty intervals.}
After observing the accept/reject decisions, the algorithm updates each trader's uncertainty interval using that trader's response to the actual posted price. Specifically, if seller $\sellerIndex$ receives price $\offeredPrice$, acceptance updates the upper bound to $\min\{\sellerValueUbOfAt{\sellerIndex}{\ts},\offeredPrice\}$, while rejection updates the lower bound to $\max\{\sellerValueLbOfAt{\sellerIndex}{\ts},\offeredPrice\}$. If buyer $\buyerIndex$ receives price $\offeredPrice$, acceptance updates the lower bound to $\max\{\buyerValueLbOfAt{\buyerIndex}{\ts},\offeredPrice\}$, while rejection updates the upper bound to $\min\{\buyerValueUbOfAt{\buyerIndex}{\ts},\offeredPrice\}$. In each case, the other endpoint remains unchanged.

\subsection{Regret Analysis}
\label{sec:segmented-pricing:analysis}

In this section, we analyze the regret of {\ConservativeSegmentedPricing} and establish the regret upper bound in \Cref{thm:segmented-pricing}. The proof of the matching regret lower bound appears in \Cref{app:segment-pricing}.

\xhdr{Overcoming three-seller, three-buyer barrier.}
To build intuition, fix any instance from the family underlying \Cref{thm:mismatch}\ref{item:mismatch:many-to-many}. Recall that $\sellerValueAt{1}=0$, $\buyerValueAt{1}=\buyerValueAt{2}=1$, and $1/4\leq\sellerValueAt{2}=\sellerValueAt{3}<\buyerValueAt{3}\leq3/4$. Under {\TwoPriceMech}, the two positive-cost sellers always make the same acceptance decision. If both reject, at most one trade occurs; if both accept while buyer $3$ rejects, the zero-cost seller may be left unmatched. These losses remain bounded away from zero even when the critical gap $\buyerValueAt{3}-\sellerValueAt{2}$ is arbitrarily small.

Segmented pricing allows the two equal-cost sellers to receive different prices. For illustration, consider a round in which the current bounds lead {\ConservativeSegmentedPricing} to select sellers $1$ and $2$ and buyers $1$ and $2$. These traders receive the common price, while seller $3$ and buyer $3$ receive exploratory prices. If both outside traders reject, the selected traders execute two trades, achieving GFT of $2-\sellerValueAt{2}$. Since the optimum offline benchmark $\OPT=2+\buyerValueAt{3}-2\sellerValueAt{2}$, the regret is only $\buyerValueAt{3}-\sellerValueAt{2}$. Thus, in this configuration, excluding the outside traders sacrifices only the marginal surplus, which vanishes as the critical gap $\buyerValueAt{3}-\sellerValueAt{2}$ shrinks.

Outside acceptances can still cause mismatch: seller $3$ may displace the zero-cost seller, or buyer $3$ may displace a unit-value buyer. However, such acceptances sharply reduce uncertainty about the accepting trader's type and therefore cannot cause constant losses indefinitely. If both outside traders reject, the forgone marginal surplus is bounded by the gap between their posted prices, which the analysis controls through reductions in uncertainty. This illustrates the two sources of regret studied below: displacement caused by outside acceptances and unrealized gains from additional trades. The formal analysis applies to every round and does not assume that the algorithm always selects these particular groups. In particular, it establishes $\asymptoticO(\log\log\timeHorizon)$ regret on every fixed instance in this family.

\xhdr{Our proof strategy.}
The exploratory prices balance these two sources of regret. First, a potentially harmful outside acceptance sharply reduces the corresponding uncertainty. We show that each trader can cause only $\asymptoticO(\log\log\timeHorizon)$ such events. Second, \Cref{lem:segmented-interval-round-regret} bounds the remaining immediate regret by $\maxTradeNum$ times the gap $\buyerPriceAtOfseg{\ts}{2}-\sellerPriceAtOfseg{\ts}{2}$ between the prices posted to outside buyers and sellers. Choosing the common price as the clipped midpoint of the outside proposals ensures that this gap can be bounded using uncertainty quantities associated with at most two traders who receive their own individual exploratory prices. For these traders, decreases in the square roots of their interval widths account for the remaining loss.

This argument explains the exponent $3/2$ in the individual exploratory prices. When the exploratory offset exceeds twice the precision correction, rejection at an individual exploratory price decreases an interval of width $\traderIntGap$ by order $\traderIntGap^{3/2}$; the corresponding decrease in $\sqrt{\traderIntGap}$ is of order $\traderIntGap$, which is the scale of the uncertainty loss. Acceptance reduces the width to at most $\traderIntGap^{3/2}/2$, so the number of costly acceptance events grows only doubly logarithmically with the horizon. The correction $\segPrecision$ ensures that acceptance when $\traderIntGap^{3/2}/2<\segPrecision$ can only occur at the endpoints $0$ or $1$, where an outside trader causes no loss through displacement. We make these observations precise below.

\xhdr{Immediate regret and outside acceptances.}
For the analysis, define the exploratory offsets
$\sellerExplorationOffsetOfAt{\sellerIndex}{\ts}
\triangleq \tfrac12(\sellerIntGapOfAt{\sellerIndex}{\ts})^{3/2}$
and
$\buyerExplorationOffsetOfAt{\buyerIndex}{\ts}
\triangleq \tfrac12(\buyerIntGapOfAt{\buyerIndex}{\ts})^{3/2}$.
In each round $\ts$, let $\countedAcceptanceNumAt{\ts}$ count the outside sellers who accept and satisfy $\sellerExplorationOffsetOfAt{\sellerIndex}{\ts}\geq\segPrecision$, together with the outside buyers who accept and satisfy $\buyerExplorationOffsetOfAt{\buyerIndex}{\ts}\geq\segPrecision$. We call these \emph{counted acceptances}. An accepting outside seller with $\sellerExplorationOffsetOfAt{\sellerIndex}{\ts}<\segPrecision$ must have cost zero and receive price zero. Indeed, its individual exploratory price is either zero or strictly below its lower cost bound, and its actual posted price is no larger. The symmetric argument shows that an accepting outside buyer with $\buyerExplorationOffsetOfAt{\buyerIndex}{\ts}<\segPrecision$ has value one and receives price one. Such uncounted acceptances can replace selected traders on the same side without decreasing GFT.

\begin{lemma}
    \label{lem:segmented-interval-round-regret}
    For {\ConservativeSegmentedPricing}, in every round $\ts\in[\timeHorizon]$, every maximum-cardinality matching $\matchAt{\ts}$ between accepting traders satisfies
    \begin{align*}
        \OPT-\gft[\matchAt{\ts}]
        \leq 2\countedAcceptanceNumAt{\ts}
        +\maxTradeNum\left(\buyerPriceAtOfseg{\ts}{2}-\sellerPriceAtOfseg{\ts}{2}\right).
\end{align*}
\end{lemma}
The first term in the upper bound of \Cref{lem:segmented-interval-round-regret} charges the loss from counted outside acceptances to the individual accepting traders. This avoids multiplying the number of such events by the maximum trade size $\maxTradeNum$. The second term captures the remaining loss through the gap between the two outside prices. In the illustrative configuration above, the first term accounts for possible displacement of the zero-cost seller or a unit-value buyer. When both outside traders reject, the second term bounds the forgone marginal surplus. Thus, \Cref{lem:segmented-interval-round-regret} generalizes the distinction between these two losses to arbitrary markets and selected groups.

\xhdr{Relating the price gap to information acquired.}
When the group of outside sellers $\sellerSubsetAt{2}{\ts}$ is nonempty, let $\sellerRepresentativeAt{\ts}$ be an outside seller with the smallest individual exploratory price, so that $\sellerExplorePriceOfAt{\sellerRepresentativeAt{\ts}}{\ts}=\sellerExplorePriceAt{\ts}$. Similarly, when the group of outside buyers $\buyerSubsetAt{2}{\ts}$ is nonempty, let $\buyerRepresentativeAt{\ts}$ be an outside buyer with the largest individual exploratory price, so that $\buyerExplorePriceOfAt{\buyerRepresentativeAt{\ts}}{\ts}=\buyerExplorePriceAt{\ts}$. For each representative that exists, define the associated radius by
\begin{align*}
\sellerRadiusAt{\ts}
    \triangleq\sellerValueUbOfAt{\sellerRepresentativeAt{\ts}}{\ts}-\sellerExplorePriceOfAt{\sellerRepresentativeAt{\ts}}{\ts}
    \;\;\mbox{and}\;\;
    \buyerRadiusAt{\ts}
    \triangleq\buyerExplorePriceOfAt{\buyerRepresentativeAt{\ts}}{\ts}-\buyerValueLbOfAt{\buyerRepresentativeAt{\ts}}{\ts}.
\end{align*}
Each radius is nonnegative and is at most the corresponding interval width plus $\segPrecision$.

\begin{lemma}
    \label{lem:segmented-interval-price-gap}
    For {\ConservativeSegmentedPricing},
    in every round $\ts\in[\timeHorizon]$, the gap $\buyerPriceAtOfseg{\ts}{2}-\sellerPriceAtOfseg{\ts}{2}$ is bounded by the sum of at most two of the radii defined as $\sellerRadiusAt{\ts}
    \triangleq\sellerValueUbOfAt{\sellerRepresentativeAt{\ts}}{\ts}-\sellerExplorePriceOfAt{\sellerRepresentativeAt{\ts}}{\ts}$ and $
    \buyerRadiusAt{\ts}
    \triangleq\buyerExplorePriceOfAt{\buyerRepresentativeAt{\ts}}{\ts}-\buyerValueLbOfAt{\buyerRepresentativeAt{\ts}}{\ts}.$ 
    . Every trader whose radius is used receives its individual exploratory price in that round.
\end{lemma}

In the illustrative configuration above, the only possible representatives are seller $3$ and buyer $3$. \Cref{lem:segmented-interval-price-gap} extends this control to larger markets: regardless of the number of outside traders, at most two representatives suffice, and each representative used in the bound receives its exact individual exploratory price.

For the remainder of the analysis, let $\chargedTraderSetAt{\ts}\subseteq\agentSet$ be the set of at most two traders whose radii are used in \Cref{lem:segmented-interval-price-gap}. For a generic trader $\genericTraderIndex\in\agentSet$, write $\traderIntGapOfAt{\genericTraderIndex}{\ts}$ for its interval width and $\explorationOffsetOfAt{\genericTraderIndex}{\ts}\triangleq\traderIntGapOfAt{\genericTraderIndex}{\ts}^{3/2}/2$. When $\genericTraderIndex\in\chargedTraderSetAt{\ts}$, let $\traderRadiusOfAt{\genericTraderIndex}{\ts}$ be its radius. Combining the two preceding lemmas yields
\begin{align}
    \OPT-\gft[\matchAt{\ts}]
    \leq 2\countedAcceptanceNumAt{\ts}+\maxTradeNum\sum\nolimits_{\genericTraderIndex\in\chargedTraderSetAt{\ts}}\traderRadiusOfAt{\genericTraderIndex}{\ts}.
    \label{eq:segmented-interval-charging}
\end{align}
The next lemma bounds the total number of counted acceptances and, consequently, the cumulative contribution of the first term (i.e., $\sum_{\ts\in[\timeHorizon]}\countedAcceptanceNumAt{\ts}$) in Eqn.~\eqref{eq:segmented-interval-charging}.

\begin{lemma}
    \label{lem:segmented-interval-acceptances}
    For {\ConservativeSegmentedPricing},
$\sum\nolimits_{\ts\in[\timeHorizon]}\countedAcceptanceNumAt{\ts}
        \leq \totalNum\left(1+\log_{3/2}\left(
            1+\frac{\ln(2\timeHorizon^2)}{3\ln2}\right)\right)$.
\end{lemma}

\Cref{lem:segmented-interval-acceptances} limits how often the costly outside acceptances illustrated by the three-seller, three-buyer example can occur. It remains to bound the cumulative radius term in Eqn.~\eqref{eq:segmented-interval-charging}. The proof below tracks decreases in the square root of each trader's interval width and bounds the contribution of small intervals separately.

\begin{proof}[Proof of regret upper bound in \Cref{thm:segmented-pricing}]
Fix a trader $\genericTraderIndex\in\agentSet$ and consider a round $\ts$ with $\genericTraderIndex\in\chargedTraderSetAt{\ts}$ and $\explorationOffsetOfAt{\genericTraderIndex}{\ts}>2\segPrecision$. The trader's individual exploratory price lies strictly inside its uncertainty interval. Because the trader receives this exact price, rejection decreases the width by $\explorationOffsetOfAt{\genericTraderIndex}{\ts}-\segPrecision\geq\explorationOffsetOfAt{\genericTraderIndex}{\ts}/2$. Acceptance leaves width at most $\explorationOffsetOfAt{\genericTraderIndex}{\ts}-\segPrecision\leq\traderIntGapOfAt{\genericTraderIndex}{\ts}/2$ and hence also decreases the width by at least $\explorationOffsetOfAt{\genericTraderIndex}{\ts}/2$. It follows that
\begin{align}
    \sqrt{\traderIntGapOfAt{\genericTraderIndex}{\ts}}-\sqrt{\traderIntGapOfAt{\genericTraderIndex}{\ts+1}}
    &=\frac{\traderIntGapOfAt{\genericTraderIndex}{\ts}-\traderIntGapOfAt{\genericTraderIndex}{\ts+1}}
        {\sqrt{\traderIntGapOfAt{\genericTraderIndex}{\ts}}+\sqrt{\traderIntGapOfAt{\genericTraderIndex}{\ts+1}}}
    \geq\frac{\explorationOffsetOfAt{\genericTraderIndex}{\ts}}{4\sqrt{\traderIntGapOfAt{\genericTraderIndex}{\ts}}}
    \overset{}{=}\frac{\traderIntGapOfAt{\genericTraderIndex}{\ts}}8.
    \label{eq:segmented-sqrt-potential}
\end{align}
where the last equality follows from $\explorationOffsetOfAt{\genericTraderIndex}{\ts}=\traderIntGapOfAt{\genericTraderIndex}{\ts}^{3/2}/2$. The radius is at most $\traderIntGapOfAt{\genericTraderIndex}{\ts}$ because the price is interior. The square root of the interval width is initially one and never increases, including in rounds when the trader does not belong to $\chargedTraderSetAt{\ts}$. Telescoping Eqn.~\eqref{eq:segmented-sqrt-potential} therefore gives
$ \sum\nolimits_{\ts:\genericTraderIndex\in\chargedTraderSetAt{\ts}, \explorationOffsetOfAt{\genericTraderIndex}{\ts}>2\segPrecision} \traderRadiusOfAt{\genericTraderIndex}{\ts}\leq8. $
Summing over all traders bounds the contribution of these radii by $8\totalNum$.

For each remaining radius, $\explorationOffsetOfAt{\genericTraderIndex}{\ts}\leq2\segPrecision$, so
$\traderIntGapOfAt{\genericTraderIndex}{\ts}\leq(4\segPrecision)^{2/3}=\timeHorizon^{-4/3}$.
Using $\traderRadiusOfAt{\genericTraderIndex}{\ts}\leq\traderIntGapOfAt{\genericTraderIndex}{\ts}+\segPrecision$ and $|\chargedTraderSetAt{\ts}|\leq2$, the total contribution of these small intervals is at most $
    2\timeHorizon(\timeHorizon^{-4/3}+\segPrecision)
    =2\timeHorizon^{-1/3}+\frac1{2\timeHorizon}
    \leq\frac52$.
This bound also covers singleton intervals and prices clipped to zero or one. Combining the two contributions yields
$\sum\nolimits_{\ts\in[\timeHorizon]}
        \sum\nolimits_{\genericTraderIndex\in\chargedTraderSetAt{\ts}}\traderRadiusOfAt{\genericTraderIndex}{\ts}
    \leq8\totalNum+\frac52$.
Summing Eqn.~\eqref{eq:segmented-interval-charging} (implied by \Cref{lem:segmented-interval-round-regret,lem:segmented-interval-price-gap}) over all rounds and applying \Cref{lem:segmented-interval-acceptances} and the upper bound of $8\totalNum+\frac52$ derived above, we guarantee that for every sequence of maximum-cardinality matchings:
    \begin{align*}
        \sum\nolimits_{\ts\in[\timeHorizon]}
            \left(\OPT-\gft[\matchAt{\ts}]\right)
        \leq8\totalNum\maxTradeNum+\frac52\maxTradeNum
        +2\totalNum\left(1+\log_{3/2}\left(
            1+\frac{\ln(2\timeHorizon^2)}{3\ln2}\right)\right).
\end{align*}
    Taking expectations under the randomized matching rule gives the same bound on the regret of {\ConservativeSegmentedPricing}. For $\timeHorizon\geq3$, this is $\asymptoticO(\totalNum\log\log\timeHorizon+\totalNum\maxTradeNum)$. This proves the regret upper bound of \Cref{thm:segmented-pricing}.
\end{proof}
 
\section{Extensions: Contextual Markets and Profit Maximization}
\label{sec:extensions}

In this section, we extend the baseline model in two directions. First, in \Cref{sec:extension:contextual-market}, we allow traders' costs and values to vary across rounds with an observed market context, while retaining gains-from-trade as the principal's objective. Second, in \Cref{sec:extension:profit}, we consider a principal who seeks to maximize profit and establish regret guarantees in the contextual settings. These extensions examine how the learning problem changes when trader types vary with observable features or the principal pursues a different objective.

\subsection{Contextual Markets}
\label{sec:extension:contextual-market}

The baseline model assumes that each trader's private type remains constant across rounds. In this section, we allow costs and values to depend on observable features, such as product characteristics or market conditions. We formalize this extension below.

Each seller $\sellerIndex\in\sellerSet$ (resp.\ buyer $\buyerIndex\in\buyerSet$) is associated with a fixed, private \emph{$\dimension$-dimensional feature vector} $\sellerVecAt{\sellerIndex}\in\ball_\dimension$ (resp.\ $\buyerVecAt{\buyerIndex}\in\ball_\dimension$) in the unit ball $\ball_\dimension$. These feature vectors are unknown to the principal. In each round $\ts\in[\timeHorizon]$, an adversarially chosen \emph{product context} $\contextAt{\ts}\in \ball_\dimension$ is revealed to the principal. Each seller's cost and each buyer's value are given by the inner product of their private feature vector and the product context, i.e., $\sellerValueAt{\sellerIndex,\ts} \triangleq  \innerproduct{\sellerVecAt{\sellerIndex},\contextAt{\ts}}$ and $\buyerValueAt{\buyerIndex,\ts} \triangleq  \innerproduct{\buyerVecAt{\buyerIndex},\contextAt{\ts}}$. We assume that the feature vectors and contexts are such that all costs and values are nonnegative, i.e., $\sellerValueAt{\sellerIndex,\ts}\geq0$ and $\buyerValueAt{\buyerIndex,\ts}\geq0$ for all $\sellerIndex\in\sellerSet$, $\buyerIndex\in\buyerSet$ and $\ts\in[\timeHorizon]$. Since feature vectors and contexts lie in the unit ball, costs and values then lie in $[0,1]$, as in the baseline model. The principal observes the context before posting prices and continues to learn from traders' accept/reject decisions. Although costs and values may vary across rounds, the underlying feature vectors remain fixed, allowing past feedback to inform future pricing decisions.

We define the cumulative optimum GFT benchmark $\OPT_{[\timeHorizon]}$ as the sum of the per-round optimum GFT benchmarks $\OPT_\ts$ over all rounds $\ts\in[\timeHorizon]$. These benchmarks are evaluated using the round-dependent type profiles $\{(\sellerValueAt{1,\ts}, \dots, \sellerValueAt{\sellerNum,\ts},\buyerValueAt{1,\ts}, \dots, \buyerValueAt{\buyerNum,\ts})\}_{t\in[T]}$. Specifically,
\begin{align*}
    \OPT_\ts \triangleq \sum\nolimits_{\ell\in [k_{\ts}]} \paren{\buyerValueAt{(\ell),\ts} - \sellerValueAt{(\ell),\ts}}
    \;\;
    \mbox{and}
    \;\;
    \OPT_{[\timeHorizon]} \triangleq \sum\nolimits_{\ts\in[\timeHorizon]}\OPT_\ts
\end{align*}
where $\buyerValueAt{(\ell),\ts}$ (resp.\ $\sellerValueAt{(\ell),\ts}$) denotes the $\ell$-th largest buyer value (resp.\ $\ell$-th smallest seller cost) in round $\ts$. The efficient trade size $k_{\ts} \in [0:\min\{\buyerNum,\sellerNum\}]$ is the largest index for which the corresponding buyer value is at least the corresponding seller cost, i.e., $\buyerValueAt{(k_\ts),\ts} \ge \sellerValueAt{(k_\ts),\ts}$; if no such index exists, then $k_\ts=0$ and $\OPT_\ts=0$. Regret is defined as in \Cref{sec:prelim}, with $\gftOpt$ replaced by $\OPT_\ts$.

Compared with the baseline model, the main challenge is that reducing uncertainty for one context need not resolve uncertainty for another. We address this challenge by extending the baseline algorithms from uncertainty intervals to multidimensional convex uncertainty sets: {\ContextualTradeThenRefine} extends the trade-search then cautious-refinement structure of {\OneToManyGFT} in \Cref{sec:uniform-pricing}, and {\QuantileTradeAndExplore} extends the conservative selection and exploratory-pricing structure of {\ConservativeSegmentedPricing} in \Cref{sec:segmented-pricing:algorithm}. The resulting regret bounds scale with the dimension $\dimension$, but the dependence on the time horizon $\timeHorizon$ is the same as in the baseline model (\Cref{thm:uniform-pricing,thm:segmented-pricing}). Specifically, under {\SinglePriceMech}, the regret remains independent of $\timeHorizon$ in contextual bilateral trade and doubly logarithmic in $\timeHorizon$ in contextual one-to-many markets; under {\SegPriceMech}, the regret remains doubly logarithmic in $\timeHorizon$ in contextual two-sided markets. We state these guarantees below and defer the algorithms and their regret analyses to Appendix~\ref{app:extension:gft}.

\begin{restatable}{theorem}{thmContextualGFT}
\label{thm:extension:contextual-gft}
Consider $\dimension$-dimensional contextual markets
with $\sellerNum$ sellers and $\buyerNum$ buyers. The following regret upper bounds hold:
\begin{enumerate}[label=(\roman*)]
    \item\label{item:extension:contextual-gft:uniform-pricing}
    {\ContextualTradeThenRefine} using {\SinglePriceMech}
    achieves regret $O(\dimension\log\dimension)$
    in bilateral trade and
    $O\bigl((\sellerNum+\buyerNum)\dimension
    (\log\log\timeHorizon+\log\dimension)\bigr)$
    in one-to-many markets.
    \item\label{item:extension:contextual-gft:segmented-pricing}
    In general two-sided markets,
    {\QuantileTradeAndExplore} using {\SegPriceMech}
    achieves regret
    $O\bigl((\sellerNum+\buyerNum)\dimension
    (\log\log\timeHorizon+
    \min\{\sellerNum,\buyerNum\}\log\dimension)\bigr)$.
\end{enumerate}
Each bound is order-optimal in time horizon
$\timeHorizon$ among all online learning algorithms
using the corresponding pricing scheme.
\end{restatable}

\subsection{Profit Maximization}
\label{sec:extension:profit}

Our analysis so far focuses on maximizing gains-from-trade. In this section, we consider a principal who seeks to maximize profit, defined as the difference between the payments collected from matched buyers and those made to matched sellers. We study this objective directly in the contextual model of \Cref{sec:extension:contextual-market}, which includes the baseline model as a special case.

The traders' private feature vectors, the observed product contexts, and the accept/reject feedback remain as defined in the preceding subsection. Since {\SinglePriceMech} always generates zero profit, we focus on {\TwoPriceMech}, the simplest pricing scheme in our framework that permits positive profit. In each round $\ts\in[\timeHorizon]$, the principal posts a seller-side price $\sellerPriceAt{\ts}$ and a buyer-side price $\buyerPriceAt{\ts}$ satisfying weak budget balance. Under the randomized matching rule, the resulting profit is
\begin{align*}
    \operatorname{Profit}_{\ts}
    \triangleq
    \min\left\{
        \left|\left\{\sellerIndex\in\sellerSet:
        \sellerValueAt{\sellerIndex,\ts}\leq\sellerPriceAt{\ts}\right\}\right|,
        \left|\left\{\buyerIndex\in\buyerSet:
        \buyerValueAt{\buyerIndex,\ts}\geq\buyerPriceAt{\ts}\right\}\right|
    \right\}
    \left(\buyerPriceAt{\ts}-\sellerPriceAt{\ts}\right).
\end{align*}
Thus, profit equals the number of trades multiplied by the price spread.

We define the cumulative optimum profit benchmark $\operatorname{Profit}^{\star}_{[\timeHorizon]}$ as the sum of the maximum profits achievable under {\TwoPriceMech} in each round when traders' types are known. Specifically,
\begin{align*}
    \operatorname{Profit}^{\star}_{\ts}
    \triangleq
    \max\left\{
        0,\;
        \max_{k\in[\min\{\buyerNum,\sellerNum\}]}
        k\left(\buyerValueAt{(k),\ts}-\sellerValueAt{(k),\ts}\right)
    \right\}
    \;\;
    \mbox{and}\;\;
    \operatorname{Profit}^{\star}_{[\timeHorizon]}
    \triangleq
    \sum\nolimits_{\ts\in[\timeHorizon]}
    \operatorname{Profit}^{\star}_{\ts},
\end{align*}
where the order statistics are defined as in the preceding subsection. The zero option allows the principal to forgo trade. We measure an algorithm's cumulative regret as the difference between this benchmark and its total profit over the horizon.

Compared with GFT maximization, profit maximization changes the role of mismatch. Under {\TwoPriceMech}, once the accepting sets are fixed, profit depends only on the number of trades and the price spread, regardless of which traders are matched. Consequently, the displacement of efficient traders does not directly reduce profit. The principal must nevertheless learn prices that balance the number of trades against the profit from each trade. By extending our search-based approach to this objective, we obtain doubly logarithmic regret in the time horizon using {\TwoPriceMech}, even in contextual many-to-many markets. We state the guarantee below and defer the algorithm
and its regret analysis to Appendix~\ref{app:extension:profit}.
\begin{restatable}{theorem}{thmContextualProfit}
\label{thm:extension:contextual-profit}
In $\dimension$-dimensional contextual two-sided markets
with $\sellerNum$ sellers and $\buyerNum$ buyers,
the regret of {\MSSPM}
(\Cref{alg:-contextual-two-sided-market-profit})
using {\TwoPriceMech} is
$O((\sellerNum+\buyerNum)\dimension\log\log\timeHorizon
+(\sellerNum+\buyerNum)\min\{\sellerNum,\buyerNum\}
\dimension\log\dimension)$.
This dependence on the time horizon $\timeHorizon$
is order-optimal among all online learning algorithms
using {\TwoPriceMech}.
\end{restatable}

\section{Conclusion and Future Directions}
\label{sec:conclusion}
We studied dynamic pricing for gains-from-trade maximization in two-sided markets under weak budget balance and accept/reject feedback. For deterministic learning policies, our results reveal how market composition and pricing flexibility jointly determine the difficulty of learning. Uniform pricing achieves constant regret in bilateral trade and doubly logarithmic regret in one-to-many markets, whereas even side-specific pricing faces a linear regret barrier in many-to-many markets. We identified the mismatch phenomenon underlying these barriers and showed that segmented pricing, with at most two groups on each side, restores doubly logarithmic regret. These results demonstrate that limited price differentiation can support efficient learning, even on instances where a common price achieves the optimum under full information. We also extended our framework to contextual markets and showed that side-specific pricing suffices for doubly logarithmic regret under profit maximization.

\xhdr{Future directions.}
Several questions remain for future research. First, our lower bounds concern deterministic learning policies. It would be interesting to determine whether randomizing posted prices can improve the attainable regret rates under the different pricing schemes. Second, traders' costs and values are fixed in our baseline model and are linear functions of the observed context in our contextual extension. A natural next step is to allow these costs and values to vary with unobserved shocks, as in models of one-sided contextual dynamic pricing with noisy buyer valuations \citep{LSS-22,TGMP-24}. Even at the same price and in the same context, a trader may then accept in one round and reject in another. How does this uncertainty affect the value of price differentiation for a platform seeking to maximize expected gains-from-trade? Third, our analysis assumes that traders are myopic. Forward-looking traders may alter their current accept/reject decisions to influence future prices, changing the information available to the platform. Designing learning policies with strong regret guarantees while accounting for these intertemporal incentives is another natural extension.
Finally, another direction is to study the joint design of pricing and matching decisions. Could a platform use historical feedback to select which accepting traders are matched and thereby reduce the need for price differentiation?

\newpage

\bibliographystyle{informs2014}
\bibliography{mybibfile}

\newpage
\renewcommand{\theHchapter}{A\arabic{chapter}}
\renewcommand{\theHsection}{A\arabic{section}}

\ECSwitch
\ECDisclaimer

\makeatletter
\renewcommand{\theHsection}{appendix.section.\thesection}
\providecommand*{\theHsubsection}{\thesubsection}
\renewcommand{\theHsubsection}{appendix.subsection.\thesubsection}

\providecommand*{\theHequation}{\theequation}
\renewcommand{\theHequation}{appendix.equation.\theequation}

\providecommand*{\theHfigure}{\thefigure}
\renewcommand{\theHfigure}{appendix.figure.\thefigure}

\providecommand*{\theHtable}{\thetable}
\renewcommand{\theHtable}{appendix.table.\thetable}

\providecommand*{\theHtheorem}{\thetheorem}
\renewcommand{\theHtheorem}{appendix.theorem.\thetheorem}

\providecommand*{\theHassumption}{\theassumption}
\renewcommand{\theHassumption}{appendix.assumption.\theassumption}

\providecommand*{\theHlemma}{\thelemma}
\renewcommand{\theHlemma}{appendix.lemma.\thelemma}

\providecommand*{\theHproposition}{\theproposition}
\renewcommand{\theHproposition}{appendix.proposition.\theproposition}

\providecommand*{\theHcorollary}{\thecorollary}
\renewcommand{\theHcorollary}{appendix.corollary.\thecorollary}

\providecommand*{\theHexample}{\theexample}
\renewcommand{\theHexample}{appendix.example.\theexample}
\makeatother

\section{Applications and Pricing Practices in Two-Sided Markets}
\label{app:applications}

In this section, we describe applications that motivate our study of pricing in two-sided markets. We first discuss their broader economic relevance and then present three examples illustrating different pricing practices.

\xhdr{Applications and economic scale.}
Two-sided online marketplaces facilitate exchange between distinct groups, including consumers and merchants (e.g., Amazon Marketplace and eBay), riders and drivers (e.g., Uber), guests and hosts (e.g., Airbnb), and clients and freelancers (e.g., Upwork). Across these settings, platforms coordinate participation on both sides, although their control over prices and other transaction terms varies.

Aggregate evidence illustrates the scale of the broader economic activities supported by these platforms. A UNCTAD report estimates that business e-commerce sales across 43 economies, accounting for around three quarters of global GDP, reached almost 27 trillion U.S.\ dollars in 2022 \citep{link-unctad-b2b-ecommerce-2024}. For the gig economy, the World Economic Forum reports an estimated market size of approximately 557 billion U.S.\ dollars in 2024 and cites a projection of 1{,}847 billion U.S.\ dollars by 2032 \citep{link-wef-gigeconomy-2024}. These figures cover broader e-commerce and gig-economy activities, rather than only transactions intermediated by platforms that set prices on both sides. They nevertheless provide context for the economic importance of coordinating supply and demand in digital markets.

\xhdr{Pricing practices.}
The following examples illustrate how platforms and market operators determine the terms of exchange. The first two distinguish payments to suppliers from prices charged to customers. The third illustrates a common clearing price for sellers and buyers within a specified market area and delivery interval.

\begin{example}[Two-sided pricing in online retail marketplaces]
\label{ex:application:retail}
    Managed retail marketplaces provide an example of centralized control over the terms offered to suppliers and consumers. Consider Temu's fully managed model, under which the platform takes responsibility for major downstream activities, including retail pricing, promotion, and fulfillment \citep{link-reuters-temu-fullymanaged-2025,link-temu-fullymanaged-sfc2024}. Suppliers deliver products to designated warehouses, while the platform manages their sale to consumers.

    The supplier-side terms and the consumer-facing retail price are determined separately. Accounts of Temu's sourcing process describe competition among manufacturers through supplier bids, together with platform control over final retail prices \citep{link-diplomat-temu-pricing-2025}. This arrangement illustrates the economic distinction captured by side-specific pricing: the amount a consumer pays need not equal the amount a supplier receives. The difference helps cover the platform's fulfillment and other operating costs and may contribute to its profit.

    This example concerns the fully managed arrangement; Temu also operates other models with different allocations of responsibility. Moreover, supplier bidding provides richer information than the binary feedback assumed in our model.
\end{example}

\begin{example}[Two-sided pricing in ride-hailing platforms]
\label{ex:application:ride-hailing}
    Ride-hailing platforms provide another example of distinct prices on the two sides of a market. Uber offers riders an upfront price before they confirm a trip and, in markets with driver upfront fares, separately presents drivers with an earnings offer. Uber explicitly distinguishes rider prices from driver earnings; its service fee varies across trips and reflects their difference after excluding specified pass-through amounts \citep{link-uber-marketplace-pricing}.

    Prices also respond to market conditions. Surge pricing adjusts offers when demand is high relative to available drivers, influencing both riders' willingness to request trips and drivers' incentives to supply service \citep{link-uber-surge-works}. These practices illustrate how a platform can use different terms on the two sides to coordinate participation. Actual offers can also vary across trips and locations, so side-specific pricing in our model captures one aspect of this richer pricing environment.
\end{example}

\begin{example}[Uniform pricing in wholesale electricity markets]
\label{ex:application:electricity}
    Wholesale electricity day-ahead markets provide an example of uniform clearing prices within a specified pricing area and delivery interval. Generators submit offers to supply electricity, while buyers submit orders to purchase it for delivery during specified intervals on the following day. The market operator aggregates these orders and determines clearing prices and quantities.

    For example, Nord Pool applies a common area price to accepted purchases and sales for the same delivery interval: producers receive that price, and consumers pay it \citep{link-nordpool-price-calculation}. Prices can differ across bidding areas when transmission constraints prevent price convergence. Similarly, U.S.\ markets such as PJM use locational marginal prices for energy purchases and sales, with prices reflecting the location at which electricity is delivered or received \citep{link-pjm-lmp}.

    The relevant uniformity is therefore across the two sides at a given location and time, rather than across an entire interconnected electricity system. These markets illustrate the common-price structure of {\SinglePriceMech}, although their prices are determined through bid-based market clearing rather than the repeated posted-price learning process studied in this paper.
\end{example}

\xhdr{Connection to our framework.}
These applications motivate our comparison of common transaction prices and differentiated terms of exchange. Our model isolates the learning problem faced by a principal who controls posted prices but observes only individual accept/reject decisions. Uniform pricing and side-specific pricing capture two basic pricing structures illustrated above. Segmented pricing provides a parsimonious extension that allows differentiation within each side while restricting the number of groups to two. The purpose of this restriction is to identify how little additional pricing flexibility suffices to overcome the learning barriers in our model. 
\section{Missing Technical Details and Analysis of \texorpdfstring{{\OneToManyGFT}}{Trade-then-Refine}}
\label{app:uniform-pricing}
\label{app:section3-upper-bound}In this section, we provide all missing technical details of {\OneToManyGFT} and the formal proofs of \Cref{thm:uniform-pricing}.
We first verify that the uncertainty interval maintained by {\OneToManyGFT} (\Cref{alg:GFTfixed}) always contains an optimal GFT-maximizing price in \Cref{lem:-GFTfixed-consistency}. We then bound the regret incurred in each phase to complete the proof.

\begin{lemma}
\label{lem:-GFTfixed-consistency}
Let $\buyerValueAt{(1)}$ denote the largest value among $\{\buyerValueAt{\buyerIndex}\}_{\buyerIndex\in[\traderNum]}$ and $\sellerValue$ denote the seller's cost.
For each $\ts\in\tsSet$ and the uncertainty interval $\hvecsAt{\ts}$ maintained by {\OneToManyGFT} (\Cref{alg:GFTfixed}):
(i) in exploration phase I, $[\sellerValue, \buyerValueAt{(1)}]\subseteq \hvecsAt{\ts}$;
(ii) in exploration phase II, $\sellerValue\le\LB{\phi}_{\ts}\le\buyerValueAt{(1)}\le\UB{\phi}_{\ts}$, so the seller always accepts the posted price.
\end{lemma}

\begin{proof}[Proof of \Cref{lem:-GFTfixed-consistency}]
For part (i), because $\hvecsAt{\ts}$ is an interval, it suffices to show that $\hvecsAt{\ts}$ contains $\sellerValue$ and $\buyerValueAt{(1)}$, and we proceed by induction on $\ts$. In round $1$, the uncertainty interval $\hvecsAt{1}=[0,1]$ contains $\sellerValue$ and $\buyerValueAt{(1)}$. Consider a round $\ts$ of exploration phase I in which $\hvecsAt{\ts}$ contains $\sellerValue$ and $\buyerValueAt{(1)}$, and the algorithm posts the midpoint price $\postPriceAt{\ts}=(\LB{\phi}_{\ts}+\UB{\phi}_{\ts})/2$. If the seller rejects and at least one buyer accepts, then $\sellerValue>\postPriceAt{\ts}$ and $\buyerValueAt{(1)}\ge\postPriceAt{\ts}$. Therefore, the updated $\hvecsAt{\ts+1}=[\postPriceAt{\ts},\UB{\phi}_{\ts}]$ contains $\sellerValue$ and $\buyerValueAt{(1)}$. If the seller accepts and all buyers reject, then $\sellerValue\le\postPriceAt{\ts}$ and $\buyerValueAt{(1)}<\postPriceAt{\ts}$, so the updated $\hvecsAt{\ts+1}=[\LB{\phi}_{\ts},\postPriceAt{\ts}]$ contains $\sellerValue$ and $\buyerValueAt{(1)}$. The seller and all buyers cannot reject together, since that would require $\buyerValueAt{(1)}<\postPriceAt{\ts}<\sellerValue$. In the remaining case, the seller and at least one buyer accept, so a trade occurs and exploration phase I ends. This proves part (i).

For part (ii), according to \Cref{alg:GFTfixed}, the algorithm enters exploration phase II only after a round $\tsBreak$ of exploration phase I in which the seller and at least two buyers accept the midpoint price $\postPriceAt{\tsBreak}$, and the uncertainty interval is then updated to $\hvecsAt{\tsBreak+1}=[\postPriceAt{\tsBreak},\UB{\phi}_{\tsBreak}]$. Therefore, we have $\sellerValue\le\postPriceAt{\tsBreak}=\LB{\phi}_{\tsBreak+1}\le\buyerValueAt{(1)}$, and part (i) gives $\buyerValueAt{(1)}\le\UB{\phi}_{\tsBreak}=\UB{\phi}_{\tsBreak+1}$. After entering exploration phase II, the lower endpoint $\LB{\phi}_{\ts}$ is non-decreasing, because each update either keeps the lower endpoint or raises it to the posted price. Hence $\sellerValue\le\LB{\phi}_{\ts}$ for every round $\ts$ in exploration phase II. Since the posted price is no less than $\LB{\phi}_{\ts}$, the seller always accepts the posted price in exploration phase II.
Now consider a round $\ts$ of exploration phase II in which $\LB{\phi}_{\ts}\le\buyerValueAt{(1)}\le\UB{\phi}_{\ts}$ and the algorithm posts the conservative binary search price $\postPriceAt{\ts}$. If all buyers reject, then $\buyerValueAt{(1)}<\postPriceAt{\ts}=\UB{\phi}_{\ts+1}$ and $\LB{\phi}_{\ts+1}=\LB{\phi}_{\ts}$, so $\LB{\phi}_{\ts+1}\le\buyerValueAt{(1)}\le\UB{\phi}_{\ts+1}$. If at least two buyers accept, then $\buyerValueAt{(1)}\ge\postPriceAt{\ts}=\LB{\phi}_{\ts+1}$ and $\UB{\phi}_{\ts+1}=\UB{\phi}_{\ts}$, so again $\LB{\phi}_{\ts+1}\le\buyerValueAt{(1)}\le\UB{\phi}_{\ts+1}$. If exactly one buyer accepts, exploration phase II ends. By induction, $\sellerValue\le\LB{\phi}_{\ts}\le\buyerValueAt{(1)}\le\UB{\phi}_{\ts}$ in every round $\ts$ of exploration phase II, which proves part (ii).
\end{proof}

\thmUniformPricing*

\begin{proof}[Proof of \Cref{thm:uniform-pricing}]
Let $\tsBreak$ denote the last round of exploration phase I, with $\tsBreak=\timeHorizon$ if exploration phase I lasts for all $\timeHorizon$ rounds. When exploration phase I ends because the seller and at least two buyers accept, $\tsBreak$ is the round recorded by \Cref{alg:GFTfixed}. We bound the regret of exploration phase I, exploration phase II, and the exploitation phase in turn, and then combine these bounds for bilateral trade and for one-to-many markets.

We first bound the regret of exploration phase I. Summing the immediate regret over the $\tsBreak$ rounds of exploration phase I gives \zqedit{$\sum\nolimits_{\ts\in[\tsBreak]}\regAt{\ts} \overset{(a)}{\le} \tsBreak\cdot\OPT \overset{(b)}{\le} \tsBreak\cdot\abs{\hvecsAt{\tsBreak}} \overset{(c)}{=} \tsBreak\cdot2^{-(\tsBreak-1)} \le 1$,}
where inequality (a) holds because the immediate regret in each round is no more than the optimum GFT benchmark $\OPT$; (b) holds because $\OPT=\buyerValueAt{(1)}-\sellerValue$, as the single seller trades with at most one buyer, and the uncertainty interval $\hvecsAt{\tsBreak}$ contains $[\sellerValue,\buyerValueAt{(1)}]$ by \Cref{lem:-GFTfixed-consistency}(i), so the width $\abs{\hvecsAt{\tsBreak}}$ is at least $\OPT$; and (c) holds because the uncertainty interval starts with width $1$ and is halved after each of the $\tsBreak-1$ binary search updates before round $\tsBreak$. This bound holds in both bilateral trade and one-to-many markets.

Two properties of the index $\buyerGapLevelAt{\ts}$ in exploration phase II are used below. First, the index is a positive integer that never decreases. Indeed, the interval inherited from exploration phase I is the upper half of an interval contained in $[0,1]$ (see \cref{line:GFTfixed-enter-phase-II}), so its width is at most $1/2$. Whenever $\abs{\hvecsAt{\ts}}\le1/2$, the index $\buyerGapLevelAt{\ts}=\bfloor{1+\log\log\abs{\hvecsAt{\ts}}^{-1}}$ is a positive integer that exceeds $\log\log\abs{\hvecsAt{\ts}}^{-1}$, so $2^{-2^{\buyerGapLevelAt{\ts}}}<\abs{\hvecsAt{\ts}}$, the posted price lies strictly inside $\hvecsAt{\ts}$, and the update yields a subinterval of $\hvecsAt{\ts}$. By induction, the width is at most $1/2$ and never increases during exploration phase II, so the index never decreases. Second, exploration phase II continues only while the width is at least $1/\timeHorizon$ (see \cref{line:GFTfixed-phase-II-stop}), so $\buyerGapLevelAt{\ts}\le1+\log\log\timeHorizon$ in every round in which the algorithm posts a price.

For exploration phase II, we first focus on the rounds where all buyers reject and claim that such rounds can happen at most $1+\log\log \timeHorizon$ times. To see this, consider such a round $\ts$ and the resulting update of the uncertainty interval from round $\ts$ to $\ts+1$. By construction, the index $\buyerGapLevelAt{\ts+1}=\buyerGapLevelAt{\ts}+1$. Thus, the feedback where all buyers reject can happen at most once for each value of the index $\buyerGapLevelAt{\ts}$. Therefore, the total regret due to rounds where all buyers reject in exploration phase II is at most $1+\log\log \timeHorizon$ since the immediate regret is at most $1$.

For rounds in exploration phase II where at least two buyers accept, whichever accepting buyer the randomized matching rule selects, the immediate regret at each such round $\ts$ is at most $\abs{\hvecsAt{\ts}}\le 2^{-{2^{\buyerGapLevelAt{\ts}-1}}}$, where the inequality holds by the definition of $\buyerGapLevelAt{\ts}$ (see \cref{line:GFTfixed-index}). For each value of the index $\buyerGapLevelAt{\ts}$, the number of rounds where at least two buyers accept is at most $2^{2^{\buyerGapLevelAt{\ts}-1}}$, since the width of the uncertainty interval satisfies $\abs{\hvecsAt{\ts}}\le 2^{-2^{\buyerGapLevelAt{\ts}-1}}$ and shrinks by $2^{-2^{\buyerGapLevelAt{\ts}}}$ after each such round (see \cref{line:GFTfixed-two-accept}). Thus, the total regret incurred for each such index value is at most $2^{2^{\buyerGapLevelAt{\ts}-1}} \cdot 2^{-2^{\buyerGapLevelAt{\ts}-1}} = 1$. Since the index is at most $1+\log\log \timeHorizon$, the total regret due to rounds where at least two buyers accept in exploration phase II is at most $1+\log\log \timeHorizon$. For rounds in exploration phase II where exactly one buyer accepts, the regret is zero since the optimal price is found. Thus, the total regret in exploration phase II is at most $O(\log\log \timeHorizon)$.

For the exploitation phase, the price is either the optimal price that makes exactly one buyer accept or the price $\sellerPrice\primed=\LB{\phi}_{\ts}$ with $\abs{\hvecsAt{\ts}}<1/\timeHorizon$. In the first case, there is no regret.
In the second case, the total regret is at most $1/\timeHorizon \cdot \timeHorizon =1$.

In bilateral trade, exploration phase II does not happen, and the exploitation phase has no regret by the first case above, so the regret of {\OneToManyGFT} in bilateral trade is at most the regret of exploration phase I, which is at most $1$.
In one-to-many markets, by combining the regret from exploration phase I, II and exploitation phase, we finish our proof.
\end{proof}

\section{Missing Technical Details and Proof of \texorpdfstring{\Cref{thm:mismatch}}{Learning Barrier Characterization}}
\label{app:mismatch}

In this section, we provide all missing technical details of two adversary constructions in \Cref{sec:mismatch-difficulty} and the formal proof of \Cref{thm:mismatch}.
The barrier in one-to-many markets in \Cref{thm:mismatch}\ref{item:mismatch:one-to-many} is presented in \Cref{app:mismatch:one-to-many}, and the barrier in many-to-many markets in \Cref{thm:mismatch}\ref{item:mismatch:many-to-many} is presented in \Cref{app:mismatch:many-to-many}.

\subsection{Learning Barrier of Side-Specific Pricing in One-to-Many Markets}
\label{app:mismatch:one-to-many}

The following lemma shows that \Cref{adv:GFTfixed} is well defined and consistent with a single fixed instance.
\begin{lemma}
    \label{lem:single-price-tight-consistency-copy}
    For any online learning algorithm $\ALG$ implementing {\TwoPriceMech}, the construction in \Cref{adv:GFTfixed} satisfies:
    \begin{enumerate}[label=(\roman*)]
        \item For every round $\ts$ in which the adversary updates the uncertainty interval, $\hvecsAt{\ts+1}\subseteq\hvecsAt{\ts}\subseteq[1/2,3/4]$ and $\abs{\hvecsAt{\ts+1}}>0$. Moreover, for every round $\ts$ of phase $\phaseNum$ in which \ref{rule:single-price-tight-accept} applies, $\abs{\hvecsAt{\ts+1}}>\abs{\hvecsPhase{\phaseNum}}/4-\priceIncreasingThresholdAt{\phaseNum}$.
        \item The buyer values $\buyerValueAt{1}$ and $\buyerValueAt{2}$ output by the adversary satisfy $1/2\le\buyerValueAt{1}<\buyerValueAt{2}\le3/4$, and running $\ALG$ on the resulting instance reproduces the same posted prices and accept/reject feedback in every round.
    \end{enumerate}
\end{lemma}

\begin{proof}
Consider a round $\ts$ in phase $\phaseNum$, with seller-side price $\sellerPriceAt{\ts}$, buyer-side price $\buyerPriceAt{\ts}$, current uncertainty interval $\hvecsAt{\ts}$, phase-start interval $\hvecsPhase{\phaseNum}$, and price-increase threshold $\priceIncreasingThresholdAt{\phaseNum}$ defined in \cref{eq:single-price-tight-threshold}.
We verify by induction on $\ts$ that $\hvecsAt{\ts+1}\subseteq\hvecsAt{\ts}\subseteq[1/2,3/4]$ and $\abs{\hvecsAt{\ts+1}}>0$ in every round in which the adversary updates the interval.
Since the phase-start interval has length at most $1/4$, the threshold satisfies $0<\priceIncreasingThresholdAt{\phaseNum}\le\abs{\hvecsPhase{\phaseNum}}/64$.
At the beginning of phase $\phaseNum$, the current interval equals $\hvecsPhase{\phaseNum}$, and within the phase only \ref{rule:single-price-tight-accept} updates the interval until the interval-length stopping condition is met, so $\abs{\hvecsAt{\ts}}>\abs{\hvecsPhase{\phaseNum}}/4$ in every round $\ts$ of phase $\phaseNum$.
If \ref{rule:single-price-tight-reject} applies, then $\buyerPriceAt{\ts}>\min\hvecsAt{\ts}+\priceIncreasingThresholdAt{\phaseNum}$, and \ref{rule:single-price-tight-reject} retains an interval of length $\priceIncreasingThresholdAt{\phaseNum}$ at the lower end of the current interval.
This new interval lies inside the current interval since $\abs{\hvecsAt{\ts}}>\abs{\hvecsPhase{\phaseNum}}/4>\priceIncreasingThresholdAt{\phaseNum}$, and its upper endpoint is strictly below the buyer-side price $\buyerPriceAt{\ts}$, so every buyer value it contains rejects this price.
If \ref{rule:single-price-tight-accept} applies, then $\buyerPriceAt{\ts}\le\min\hvecsAt{\ts}+\priceIncreasingThresholdAt{\phaseNum}$, where the buyer-side price may be negative, and \ref{rule:single-price-tight-accept} increases the lower endpoint by at most $\priceIncreasingThresholdAt{\phaseNum}$ and leaves the upper endpoint unchanged, so every buyer value in the updated interval accepts the buyer-side price $\buyerPriceAt{\ts}$, and $\abs{\hvecsAt{\ts+1}} \ge\abs{\hvecsAt{\ts}}-\priceIncreasingThresholdAt{\phaseNum} >\frac{\abs{\hvecsPhase{\phaseNum}}}{4}-\priceIncreasingThresholdAt{\phaseNum} >0$.
Starting from the initial interval $[1/2,3/4]$, these observations establish the claimed interval properties for every update, which proves part (i).

For part (ii), suppose that the construction ends after round $\ts$, either because the adversary stops updating the interval or because $\ts=\timeHorizon$, so that $\buyerValueAt{1}=\min\hvecsAt{\ts+1}$ and $\buyerValueAt{2}=\max\hvecsAt{\ts+1}$.
By part (i), $1/2\le\buyerValueAt{1}<\buyerValueAt{2}\le3/4$.
Since every later interval is contained in each earlier interval, $\buyerValueAt{1}$ and $\buyerValueAt{2}$ lie in every earlier uncertainty interval, and hence reproduce every earlier acceptance or rejection decision of the two buyers, including acceptance at equality.
The zero-cost seller's decisions are also reproduced, since in every round the seller accepts if and only if the seller-side price is nonnegative.
Since $\ALG$ observes only the acceptance decisions, induction over the rounds shows that it also posts the same prices on this fixed instance, and it therefore reproduces the same prices and feedback in every round.
\end{proof}

We now restate the regret lower bound in \Cref{thm:mismatch}\ref{item:mismatch:one-to-many} and prove it using \Cref{adv:GFTfixed} and \Cref{lem:single-price-tight-consistency-copy}.
\lemMismatchOneToMany*

\begin{proof}
Fix an online learning algorithm $\ALG$.
Let $\reg[\ALG]$ denote the algorithm's cumulative GFT regret over the $\timeHorizon$ rounds.
By \Cref{lem:single-price-tight-consistency-copy}, the construction is consistent with one fixed market instance, in which $1/2\le\buyerValueAt{1}<\buyerValueAt{2}\le3/4$. We bound the regret according to which stopping condition is met, including exhaustion of the horizon.

Before the buyer values are fixed, a round in \ref{rule:single-price-tight-reject} has no trade and incurs regret $\buyerValueAt{2}\ge1/2$, whereas a round in \ref{rule:single-price-tight-accept} incurs expected regret at least $(\buyerValueAt{2}-\buyerValueAt{1})/2$.
Indeed, both buyers accept in such a round. If the seller-side price is nonnegative, the zero-cost seller also accepts, and the randomized matching rule selects each buyer with probability one-half, so the expected regret is $(\buyerValueAt{2}-\buyerValueAt{1})/2$. Otherwise, the seller rejects, no trade occurs, and the regret is $\buyerValueAt{2}$.
By construction, the counter $\phaseCount$ counts exactly the rounds in \ref{rule:single-price-tight-reject}, in which both buyers reject.
Moreover, the phase index and counter start at $1$ and $0$, respectively, and increase together only in \ref{rule:single-price-tight-reject}, so $\phaseNum=\phaseCount+1$ after every update.
Any rounds remaining after the buyer values are fixed have nonnegative regret.
There are three possible termination cases: (1) the counter reaches its threshold in \ref{rule:single-price-tight-reject}, (2) the interval-length condition is met in \ref{rule:single-price-tight-accept}, or (3) all $\timeHorizon$ rounds are used without either condition being met.

For case (1), the counter reaches $\floor{\frac14\cdot\log_2\log_2\timeHorizon}$.
There have been exactly this many rounds in which both buyers reject, each contributing regret at least $1/2$.
Consequently, $\reg[\ALG]\ge\frac12\cdot\floor{\frac14\cdot\log_2\log_2\timeHorizon}$.

For case (2), suppose that \ref{rule:single-price-tight-accept} ends the construction after round $\ts$ in phase $\phaseNum$ because $\abs{\hvecsAt{\ts+1}}\le\abs{\hvecsPhase{\phaseNum}}/4$.
Every round of this phase uses \ref{rule:single-price-tight-accept}, since a rejection would end the phase and start a new one.
The upper endpoint stays fixed, while each update increases the lower endpoint by at most $\priceIncreasingThresholdAt{\phaseNum}$.
Reducing the interval length from $\abs{\hvecsPhase{\phaseNum}}$ to at most $\abs{\hvecsPhase{\phaseNum}}/4$ therefore requires at least $3\cdot\abs{\hvecsPhase{\phaseNum}}/(4\cdot\priceIncreasingThresholdAt{\phaseNum})$ both-accept rounds.
At the same time, the last update cannot make the final buyer-value gap arbitrarily small:
$\buyerValueAt{2}-\buyerValueAt{1}=\abs{\hvecsAt{\ts+1}}\overset{(a)}{>}\frac{\abs{\hvecsPhase{\phaseNum}}}{4}-\priceIncreasingThresholdAt{\phaseNum}\overset{(b)}{\ge}\frac{15}{64}\cdot\abs{\hvecsPhase{\phaseNum}}$,
where inequality (a) uses \Cref{lem:single-price-tight-consistency-copy}(i), and inequality (b) uses the threshold bound $\priceIncreasingThresholdAt{\phaseNum}\le\abs{\hvecsPhase{\phaseNum}}/64$, which follows from \cref{eq:single-price-tight-threshold} since $\abs{\hvecsPhase{\phaseNum}}\le1/4$.
Multiplying the number of both-accept rounds by a lower bound on their expected regret gives
\begin{align*}
    \reg[\ALG]
    \overset{(a)}{>}\frac{3\cdot\abs{\hvecsPhase{\phaseNum}}}{4\cdot\priceIncreasingThresholdAt{\phaseNum}}
\cdot\frac{15\cdot\abs{\hvecsPhase{\phaseNum}}}{128}
\overset{(b)}{=}\frac{45}{8}\cdot\floor{\frac14\cdot\log_2\log_2\timeHorizon}
    \ge\frac12\cdot\floor{\frac14\cdot\log_2\log_2\timeHorizon},
\end{align*}
where inequality (a) uses the lower bound on the final buyer-value gap $\buyerValueAt{2}-\buyerValueAt{1}\geq \frac{15}{64}\cdot\abs{\hvecsPhase{\phaseNum}}$ and the lower bound $(\buyerValueAt{2}-\buyerValueAt{1})/2$ on the expected regret of each both-accept round, and equality (b) follows from the threshold definition in \cref{eq:single-price-tight-threshold}.

For case (3), neither stopping condition is met during the $\timeHorizon$ rounds.
Let $\phaseNum$ now denote the phase index after the final update, including the newly initialized phase if both buyers reject in round $\timeHorizon$.
The final counter satisfies $\phaseCount<\floor{\frac14\cdot\log_2\log_2\timeHorizon}$, so the identity $\phaseNum=\phaseCount+1$ gives $\phaseNum\le\floor{\frac14\cdot\log_2\log_2\timeHorizon}$.
The initial phase-start interval has length $1/4$, and each preceding phase ends with \ref{rule:single-price-tight-reject}.
Thus, for each preceding phase index $\phaseNumVar\in\{1,\ldots,\phaseNum-1\}$, $\abs{\hvecsPhase{1}}=\frac14, \abs{\hvecsPhase{\phaseNumVar+1}} =\priceIncreasingThresholdAt{\phaseNumVar} =\frac{\abs{\hvecsPhase{\phaseNumVar}}^2} {64\cdot\floor{\frac14\cdot\log_2\log_2\timeHorizon}}$.
Taking logarithms of the reciprocals and iterating this recurrence gives
\begin{align*}
\log_2\frac{1}{\abs{\hvecsPhase{\phaseNum}}}
    &=2^{\phaseNum-1}\left(8+\log_2\floor{\frac14\log_2\log_2\timeHorizon}\right)
      -\left(6+\log_2\floor{\frac14\log_2\log_2\timeHorizon}\right)
\overset{(a)}{\le}\frac12\sqrt{\log_2\timeHorizon}
      \le\frac12\log_2\timeHorizon,
\end{align*}
where inequality~(a) uses the bound $\phaseNum\le\floor{\frac14\cdot\log_2\log_2\timeHorizon}$ on the final phase index and drops the nonnegative subtracted term, and inequality (b) uses $8+\log_2\floor{\frac14\cdot\log_2\log_2\timeHorizon}\le(\log_2\timeHorizon)^{1/4}$, which holds for $\timeHorizon\ge2^{2^{16}}$.
It follows that $\abs{\hvecsPhase{\phaseNum}}\ge\timeHorizon^{-1/2}$.
If both buyers reject in the final round, the newly initialized phase-start interval is exactly $\hvecsAt{\timeHorizon+1}$, so the final buyer-value gap equals $\abs{\hvecsPhase{\phaseNum}}$.
If both buyers accept in the final round, the interval-length stopping condition is not met, so the final gap is greater than $\abs{\hvecsPhase{\phaseNum}}/4$.
In either event, the final gap exceeds $\abs{\hvecsPhase{\phaseNum}}/4$.
Exactly $\phaseCount$ rounds have both buyers rejecting, leaving $\timeHorizon-\phaseCount>\timeHorizon-\floor{\frac14\cdot\log_2\log_2\timeHorizon}$ both-accept rounds.
Their cumulative expected regret satisfies
\begin{align*}
    \reg[\ALG]
    \ge\left(\timeHorizon-\floor{\frac14\cdot\log_2\log_2\timeHorizon}\right)
\cdot\frac{\abs{\hvecsPhase{\phaseNum}}}{8}
\overset{(a)}{\ge}\frac{\timeHorizon-\floor{\frac14\cdot\log_2\log_2\timeHorizon}}{8\cdot\sqrt{\timeHorizon}}
\overset{(b)}{\ge}\frac{\sqrt{\timeHorizon}}{16}
    \overset{(c)}{\ge}\frac12\cdot\floor{\frac14\cdot\log_2\log_2\timeHorizon},
\end{align*}
where the first inequality uses the lower bound $(\buyerValueAt{2}-\buyerValueAt{1})/2>\abs{\hvecsPhase{\phaseNum}}/8$ on the expected regret of each both-accept round, inequality (a) uses the lower bound on the final phase-start interval length, inequality (b) uses $\timeHorizon-\floor{\frac14\cdot\log_2\log_2\timeHorizon}\ge\timeHorizon/2$, and inequality (c) uses $\sqrt{\timeHorizon}\ge8\cdot\floor{\frac14\cdot\log_2\log_2\timeHorizon}$.
The bounds used in (b) and (c) hold for $\timeHorizon\ge2^{2^{16}}$.

Each termination case therefore gives $\reg[\ALG]\ge\frac12\cdot\floor{\frac14\cdot\log_2\log_2\timeHorizon}$.
For $\timeHorizon\ge2^{2^{16}}$, the floor satisfies $\floor{\frac14\cdot\log_2\log_2\timeHorizon}\ge\frac18\cdot\log_2\log_2\timeHorizon$.
Thus $\reg[\ALG]\ge\frac1{16}\cdot\log_2\log_2\timeHorizon$, which completes the proof of \Cref{thm:mismatch}\ref{item:mismatch:one-to-many}.
\end{proof}

\subsection{Learning Barrier of Side-Specific Pricing in Many-to-Many Markets}
\label{app:mismatch:many-to-many}

First, we verify that the instance output by \Cref{adv:GFTtwoprice} is consistent with the feedback observed by the algorithm.
\begin{lemma}
    \label{adv:GFTtwo-consistency-copy}
    For any online learning algorithm $\ALG$ implementing {\TwoPriceMechs}, running $\ALG$ on the instance output by \Cref{adv:GFTtwoprice} reproduces the same posted prices and accept/reject feedback in every round.
\end{lemma}

\begin{proof}
The interval $[\UB{\sellerPrice},\LB{\buyerPrice}]$ never becomes empty: initially $\UB{\sellerPrice}=\tfrac14<\tfrac34=\LB{\buyerPrice}$, and each update preserves $\UB{\sellerPrice}<\LB{\buyerPrice}$, since $\UB{\sellerPrice}$ becomes $\max\{\UB{\sellerPrice},\sellerPriceAt{\ts}\}$ only when $\sellerPriceAt{\ts}<\LB{\buyerPrice}$, and $\LB{\buyerPrice}$ becomes $\min\{\LB{\buyerPrice},\buyerPriceAt{\ts}\}$ only when $\buyerPriceAt{\ts}>\UB{\sellerPrice}$.
Since the adversary finally sets $\sellerValueAt{2}=\sellerValueAt{3}=\frac{2\UB{\sellerPrice}+\LB{\buyerPrice}}{3}$ and $\buyerValueAt{3}=\frac{\UB{\sellerPrice}+2\LB{\buyerPrice}}{3}$, the final values satisfy $\tfrac14\le\UB{\sellerPrice}<\sellerValueAt{2}=\sellerValueAt{3}<\buyerValueAt{3}<\LB{\buyerPrice}\le\tfrac34$.

Fix a round $\ts$.
If the adversary lets only seller $1$ accept, then $\UB{\sellerPrice}$ becomes $\max\{\UB{\sellerPrice},\sellerPriceAt{\ts}\}\ge\sellerPriceAt{\ts}$, so $\sellerPriceAt{\ts}\le\UB{\sellerPrice}<\sellerValueAt{2}=\sellerValueAt{3}$ for the final $\UB{\sellerPrice}$, and sellers $2$ and $3$ indeed reject.
Otherwise, $\sellerPriceAt{\ts}\ge\LB{\buyerPrice}$ in round $\ts$, and $\LB{\buyerPrice}$ only decreases, so $\sellerPriceAt{\ts}>\sellerValueAt{2}=\sellerValueAt{3}$ and all three sellers indeed accept.
Symmetrically, $\buyerPriceAt{\ts}>\buyerValueAt{3}$ if the adversary lets only buyers $1$ and $2$ accept, and $\buyerPriceAt{\ts}<\buyerValueAt{3}$ otherwise.
Moreover, since $\sellerValueAt{1}=0$ and $\buyerValueAt{1}=\buyerValueAt{2}=1$, seller $1$ and buyers $1$ and $2$ accept every price in $[0,1]$.
Hence every reported decision agrees with the final instance.
Since $\ALG$ observes only these decisions, induction over the rounds shows that it posts the same prices on the final instance, and it therefore reproduces the same prices and feedback in every round.
\end{proof}

We now restate the regret lower bound in \Cref{thm:mismatch}\ref{item:mismatch:many-to-many} and prove it using \Cref{adv:GFTtwoprice} and \Cref{adv:GFTtwo-consistency-copy}.
\lemMismatchManyToMany*

\begin{proof}
    According to \Cref{adv:GFTtwoprice} and \Cref{adv:GFTtwo-consistency-copy}, there is one seller $\sellerValueAt{1}$ or all three sellers accepting the price $\sellerPriceAt{\ts}$, and two buyers $\buyerValueAt{1}, \buyerValueAt{2}$ or all three buyers accepting the price $\buyerPriceAt{\ts}$ for each round $\ts$. In addition, the situation where all three sellers and three buyers are willing to trade cannot happen because of the budget balance constraint $\sellerPriceAt{\ts} \leq \buyerPriceAt{\ts}$. Therefore, the number of sellers accepting $\sellerPriceAt{\ts}$ never equal to the number of buyers accepting $\buyerPriceAt{\ts}$. Concretely, we have three cases: $(1)$ one seller $\sellerValueAt{1}$ and two buyers $\buyerValueAt{1},\buyerValueAt{2}$, $(2)$ one seller $\sellerValueAt{1}$ and all three buyers, $(3)$ all three sellers and two buyers $\buyerValueAt{1},\buyerValueAt{2}$. For case $(1)$ and case $(2)$, there is only one trading pair while the optimal matching consists of three pairs, making $\ALG$ suffer regret at least ${1}/{4}$. For case $(3)$, a direct computation under the randomized matching rule shows that the expected regret in such a round is at least $1/12$. Thus, in all three cases, the expected regret in each round is at least $1/12$.
The total expected regret in $\timeHorizon$ rounds is at least $\timeHorizon/12$. This finishes the proof.
\end{proof}
 
\section{Missing Technical Details and Analysis of \texorpdfstring{{\ConservativeSegmentedPricing}}{Trade-and-Explore}}
\label{app:segment-pricing}

This section proves the supporting lemmas in \Cref{sec:segmented-pricing} and establishes a regret lower bound for algorithms using {\SegPriceMech}. We use the notation introduced in \Cref{sec:prelim,sec:segmented-pricing}.

\subsection{Proofs of the Supporting Lemmas}
\label{app:segmented-pricing-upper-bound}

In the first two proofs, we fix a round and suppress the round index.

\begin{proof}[Proof of \Cref{lem:segmented-interval-round-regret}]
    Let $\selectedGft\triangleq \sum_{\buyerIndex\in\buyerSubset{1}}\buyerValueAt{\buyerIndex} -\sum_{\sellerIndex\in\sellerSubset{1}}\sellerValueAt{\sellerIndex}$
denote the GFT obtained by matching all selected traders. Decompose the GFT of each matched pair into
    $(\commonPrice-\sellerValueAt{\sellerIndex})+(\buyerValueAt{\buyerIndex}-\commonPrice)$.
    Both contributions are nonnegative for accepting traders. Since the realized matching contains at least $\selectedTradeNum$ pairs, each unmatched selected seller can be associated with a distinct matched outside seller. An uncounted accepting outside seller has cost zero and contributes at least as much as the selected seller it replaces. A counted replacement reduces GFT by at most one. Applying the symmetric argument to buyers and retaining the nonnegative contributions of any additional matched traders yields $\gft[\match]\geq \selectedGft-\countedAcceptanceNum$.

    Now remove the $\countedAcceptanceNum$ counted accepting outside traders, and let $\remainingOpt$ denote the optimum GFT in the remaining market. Removing their incident pairs from an optimal matching discards at most $\countedAcceptanceNum$ pairs, each contributing at most one to GFT. Thus, $\OPT\leq\remainingOpt+\countedAcceptanceNum$. Every remaining outside seller has cost at least $\sellerTwoPrices{2}$: a rejecting seller has strictly higher cost, while an uncounted accepting seller has both cost and posted price equal to zero. Similarly, every remaining outside buyer has value at most $\buyerTwoPrices{2}$. Using the same decomposition relative to $\commonPrice$, the selected traders in any matching of the remaining market contribute at most $\selectedGft$, because their individual contributions are nonnegative. Each outside seller contributes at most $\commonPrice-\sellerTwoPrices{2}$, and each outside buyer contributes at most $\buyerTwoPrices{2}-\commonPrice$. Since at most $\maxTradeNum$ traders on each side can be matched,
    $\remainingOpt \leq \selectedGft+\maxTradeNum(\commonPrice-\sellerTwoPrices{2}) +\maxTradeNum(\buyerTwoPrices{2}-\commonPrice) =\selectedGft+\maxTradeNum(\buyerTwoPrices{2}-\sellerTwoPrices{2})$.
    Combining these inequalities with $\OPT\leq\remainingOpt+\countedAcceptanceNum$ proves \Cref{lem:segmented-interval-round-regret}.
\end{proof}

\begin{proof}[Proof of \Cref{lem:segmented-interval-price-gap}]
    We distinguish four cases according to whether the seller-side and buyer-side gaps relative to the common price $\commonPrice$ are positive.

    If both gaps are positive, then $\sellerExplorePrice<\commonPrice<\buyerExplorePrice$. Both outside groups are nonempty, so $\selectedTradeNum<\maxTradeNum$. Let $\sellerIndex$ and $\buyerIndex$ denote the outside traders attaining the seller-side and buyer-side proposals, respectively. Maximality of $\selectedTradeNum$ and the ordering of the bounds imply $\sellerValueUbAt{\sellerIndex} \geq\sellerValueOrdUb{\selectedTradeNum+1} >\buyerValueOrdLb{\selectedTradeNum+1} \geq\buyerValueLbAt{\buyerIndex}$.
    Hence,
\begin{align*}
\buyerTwoPrices{2}-\sellerTwoPrices{2}=\buyerExplorePriceOf{\buyerIndex}-\sellerExplorePriceOf{\sellerIndex}
        &=(\buyerExplorePriceOf{\buyerIndex}-\buyerValueLbAt{\buyerIndex})
          +(\buyerValueLbAt{\buyerIndex}-\sellerValueUbAt{\sellerIndex})
          +(\sellerValueUbAt{\sellerIndex}-\sellerExplorePriceOf{\sellerIndex})
        \leq  (\buyerExplorePriceOf{\buyerIndex}-\buyerValueLbAt{\buyerIndex})
          +(\sellerValueUbAt{\sellerIndex}-\sellerExplorePriceOf{\sellerIndex})=\buyerRadius+\sellerRadius.
\end{align*}
Both traders receive their individual exploratory prices.

    If only the seller-side gap is positive, then $\sellerExplorePrice<\commonPrice$ and $\buyerExplorePrice\leq\commonPrice$. The midpoint of the two proposals is strictly below $\commonPrice$, so \cref{eq:segmented-common-price} clips it to the lower endpoint: $\commonPrice=\commonPriceLb$. This implies $\selectedTradeNum>0$, since otherwise $\commonPriceLb=0$, contradicting the nonnegativity of the seller-side proposal. Let $\sellerIndex$ be an outside seller attaining that proposal. The ordering of seller upper bounds gives $\commonPriceLb\leq\sellerValueUbAt{\sellerIndex}$, and therefore
    $\buyerTwoPrices{2}-\sellerTwoPrices{2}=\commonPriceLb-\sellerExplorePriceOf{\sellerIndex}\leq \sellerRadius$.
    Seller $\sellerIndex$ receives $\sellerExplorePriceOf{\sellerIndex}$.

    If only the buyer-side gap is positive, the midpoint of the two proposals is strictly above $\commonPrice$, so $\commonPrice=\commonPriceUb$ and $\selectedTradeNum>0$. Let $\buyerIndex$ be an outside buyer attaining the buyer-side proposal. The ordering of buyer lower bounds gives $\commonPriceUb\geq\buyerValueLbAt{\buyerIndex}$, yielding
    $\buyerTwoPrices{2}-\sellerTwoPrices{2}=\buyerExplorePriceOf{\buyerIndex}-\commonPriceUb\leq \buyerRadius$.
    Buyer $\buyerIndex$ receives $\buyerExplorePriceOf{\buyerIndex}$.

    Finally, if both gaps are zero, no radius is needed. The empty-group conventions ensure that a positive gap on either side can arise only when the corresponding outside group is nonempty.
\end{proof}

\begin{proof}[Proof of \Cref{lem:segmented-interval-acceptances}]
    Consider a counted acceptance by an outside seller $\sellerIndex$ in round $\ts$. Since $\sellerExplorationOffsetOfAt{\sellerIndex}{\ts}\geq\segPrecision$, the seller's individual exploratory price equals
    $\sellerValueLbOfAt{\sellerIndex}{\ts}+\sellerExplorationOffsetOfAt{\sellerIndex}{\ts}-\segPrecision$.
    The actual posted price is no larger. Acceptance therefore leaves an interval of width at most
    $\sellerExplorationOffsetOfAt{\sellerIndex}{\ts}-\segPrecision
    \leq(\sellerIntGapOfAt{\sellerIndex}{\ts})^{3/2}/2$.
    The same contraction holds for a counted acceptance by an outside buyer. All other responses weakly decrease interval widths.

    Fix a trader $\genericTraderIndex$. Starting from width one, after $\acceptanceCount$ counted acceptances, the trader's interval width is at most $2^{-2((3/2)^\acceptanceCount-1)}$.
This follows by induction from the contraction $\traderIntGap\mapsto\traderIntGap^{3/2}/2$. For an additional acceptance to be counted, the width immediately before that acceptance must satisfy $\traderIntGap^{3/2}/2\geq\segPrecision$, or equivalently,
    $\traderIntGap\geq(2\timeHorizon^2)^{-2/3}$.
    Since the width cannot increase between counted acceptances, combining these bounds gives the necessary condition $(3/2)^\acceptanceCount\leq1+\frac{\ln(2\timeHorizon^2)}{3\ln2}$.
Applying this condition immediately before the last counted acceptance shows that each trader contributes at most
    $1+\log_{3/2}(1+\ln(2\timeHorizon^2)/(3\ln2))$
    counted acceptances. Summing over all $\totalNum$ traders proves the lemma.
\end{proof}

\subsection{Regret Lower Bound for Segmented Pricing}
\label{app:segmented-pricing-lower-bound}

We now establish the regret lower bound underlying the optimality statement in \Cref{thm:segmented-pricing}. We show that, in balanced markets, the worst-case regret grows at least linearly in the population size and doubly logarithmically in the horizon. As in \Cref{sec:prelim}, we consider deterministic learning algorithms, and regret takes expectation over the randomized matching rule.

\begin{theorem}
    \label{thm:segmented-population-lower-bound}
    For every population size $\sellerNum=\buyerNum\geq4$, horizon $\timeHorizon\geq\ceil{e^{256}}$, and online learning algorithm $\ALG$ using {\SegPriceMech}, there exists a fixed $\sellerNum$-seller, $\buyerNum$-buyer instance on which
    \begin{align}
        \reg[\ALG]\geq
        \frac{\floor{\buyerNum/4}}{6}
        \floor{\frac12\log_2\ln\timeHorizon}.
        \label{eq:segmented-population-lower-bound}
    \end{align}
    Moreover, the optimum GFT benchmark $\OPT$ can be achieved under {\SinglePriceMech} on this instance. Consequently, every algorithm using {\SegPriceMech} incurs $\Omega(\buyerNum\log\log\timeHorizon)$ regret on some balanced-market instance.
\end{theorem}

Since $\totalNum=2\buyerNum$ in balanced markets, this result establishes the necessity of the $\totalNum\log\log\timeHorizon$ term in \Cref{thm:segmented-pricing}. The construction allows the algorithm to assign traders with identical types to different segments.

\xhdr{Hard-instance family.}
Let $r\triangleq\floor{\buyerNum/4}$ and $\rho\triangleq\buyerNum-4r$. There are $2r$ zero-cost sellers, $r$ sellers of cost $\sellerValue_-$, $r$ sellers of cost $\sellerValue_+$, and $\rho$ sellers of cost one, where
$1/3\leq\sellerValue_-<\sellerValue_+\leq2/3$.
On the buyer side, $3r$ buyers have value one and the remaining $r+\rho$ buyers have value zero. \Cref{fig:segmented-population-hard-instance} illustrates this distribution. Let $\mathcal U\subseteq\sellerSet$ denote the set of $2r$ sellers with intermediate costs. The adversary will choose both cost levels and the subset $\mathcal Z\subseteq\mathcal U$ of $r$ sellers with cost $\sellerValue_-$; the other sellers in $\mathcal U$ have cost $\sellerValue_+$. The lower bound continues to hold even if the algorithm knows $\mathcal U$, the number of sellers at each cost level, and the identities and types of all traders outside $\mathcal U$.
\begin{figure}[htbp]
    \centering
    \begin{tikzpicture}[x=9.4cm,y=1cm,font=\small,
        population axis/.style={draw=black!65,line width=0.5pt,->,>=stealth},
        population tick/.style={draw=black!55,line width=0.45pt}]
\fill[black!5] (0.33333,0.85) rectangle (0.66667,1.83);
        \draw[black!40,densely dashed] (0.33333,0.85) -- (0.33333,1.83);
        \draw[black!40,densely dashed] (0.66667,0.85) -- (0.66667,1.83);

        \node[anchor=east,align=right] at (-0.055,1) {Sellers\\(cost)};
        \node[anchor=east,align=right] at (-0.055,-0.55) {Buyers\\(value)};
        \draw[population axis] (-0.025,1) -- (1.045,1);
        \draw[population axis] (-0.025,-0.55) -- (1.045,-0.55);

        \foreach \x in {0,0.33333,0.43,0.58,0.66667,1}
            \draw[population tick] (\x,0.94) -- (\x,1.06);
        \foreach \x in {0,1}
            \draw[population tick] (\x,-0.61) -- (\x,-0.49);

        \filldraw[fill=black,draw=black] (0,1) circle[radius=2.6pt];
        \filldraw[fill=black,draw=black] (0.43,1) circle[radius=2.6pt];
        \filldraw[fill=white,draw=black,line width=0.7pt] (0.58,1) circle[radius=2.6pt];
        \filldraw[fill=white,draw=black,line width=0.7pt] (1,1) circle[radius=2.6pt];
        \filldraw[fill=white,draw=black,line width=0.7pt] (0,-0.55) circle[radius=2.6pt];
        \filldraw[fill=black,draw=black] (1,-0.55) circle[radius=2.6pt];

        \node[above=5pt] at (0,1) {$2r$};
        \node[above=5pt] at (0.43,1) {$r$};
        \node[above=5pt] at (0.58,1) {$r$};
        \node[above=5pt] at (1,1) {$\rho$};
        \node[above=5pt] at (0,-0.55) {$r+\rho$};
        \node[above=5pt] at (1,-0.55) {$3r$};

        \node[below=5pt] at (0,1) {$0$};
        \node[below=5pt] at (0.33333,1) {$1/3$};
        \node[below=5pt] at (0.43,1) {$\sellerValue_-$};
        \node[below=5pt] at (0.58,1) {$\sellerValue_+$};
        \node[below=5pt] at (0.66667,1) {$2/3$};
        \node[below=5pt] at (1,1) {$1$};
        \node[below=5pt] at (0,-0.55) {$0$};
        \node[below=5pt] at (1,-0.55) {$1$};

        \draw[black!65,line width=0.5pt]
            (0.43,1.55) -- (0.43,1.72) -- (0.58,1.72) -- (0.58,1.55);
        \node[above=3pt] at (0.505,1.72) {$\mathcal U$ ($2r$ sellers)};
    \end{tikzpicture}
    \caption{Hard-instance family for the lower bound, with $r=\floor{\buyerNum/4}$ and $\rho=\buyerNum-4r$. Labels above markers give group sizes. }
    \label{fig:segmented-population-hard-instance}
\end{figure}
For every instance in this family, the optimal matching serves the $3r$ unit-value buyers using the $2r$ zero-cost sellers and the $r$ sellers of cost $\sellerValue_-$. Hence,
\begin{align}
    \OPT=3r-r\sellerValue_-.
    \label{eq:segmented-population-opt}
\end{align}
Any common price in $[\sellerValue_-,\sellerValue_+)$ induces exactly these accepting traders and achieves $\OPT$.

The learning difficulty is to identify which sellers in $\mathcal U$ have the lower cost. We first quantify the regret when the observed feedback does not distinguish these two cost levels. We then construct feedback consistent with a fixed instance for which this uncertainty persists long enough to yield the claimed lower bound.

\xhdr{Immediate regret under indistinguishable costs.}
For a fixed round $\ts$, let $\sellerAccAt{\ts}$ and $\buyerAccAt{\ts}$ denote the accepting seller and buyer sets. Under the randomized matching rule,
\begin{align}
    \expect{\gft[\matchAt{\ts}]}
    =\min\{\abs{\sellerAccAt{\ts}},\abs{\buyerAccAt{\ts}}\}
    \left(
        \frac{\sum_{\buyerIndex\in\buyerAccAt{\ts}}\buyerValueAt{\buyerIndex}}
             {\abs{\buyerAccAt{\ts}}}
        -\frac{\sum_{\sellerIndex\in\sellerAccAt{\ts}}\sellerValueAt{\sellerIndex}}
              {\abs{\sellerAccAt{\ts}}}
    \right),
    \label{eq:segmented-population-matching-gft}
\end{align}
with value zero if either accepting set is empty.

Fix the two cost levels and write $\Delta\triangleq\sellerValue_+-\sellerValue_-$ and $\mu\triangleq(\sellerValue_-+\sellerValue_+)/2$. For the following comparison, choose $\mathcal Z$ uniformly from the $r$-element subsets of $\mathcal U$. We write $\expect[\mathcal Z]{\cdot}$ for this additional expectation.

\begin{lemma}
    \label{lem:segmented-population-round-regret}
    Suppose that the posted prices and all accept/reject decisions in round $\ts$ are the same for every assignment $\mathcal Z$. Then $\expect[\mathcal Z]{\regAt{\ts}}\geq r\Delta/2$. If all sellers in $\mathcal U$ accept, then $\regAt{\ts}\geq r/6$ for every assignment. If all sellers in $\mathcal U$ reject, then $\regAt{\ts}\geq r/3$ for every assignment.
\end{lemma}

\begin{proof}[Proof of \Cref{lem:segmented-population-round-regret}]
    Fix the accepting sets and an assignment $\mathcal Z$, and let $a\triangleq\abs{\sellerAccAt{\ts}\cap\mathcal U}$. To upper bound the expected GFT, retain precisely the accepting sellers from $\mathcal U$, add all zero-cost sellers and unit-value buyers to the accepting sets, and remove all cost-one sellers and zero-value buyers. Each change weakly increases the expected GFT in Eqn.~\eqref{eq:segmented-population-matching-gft}. For example, adding a zero-cost seller either increases the number of trades by one without increasing the total matched cost, or leaves the number of trades unchanged and reduces the average cost of matched sellers. Adding a unit-value buyer is symmetric. Removing a cost-one seller or a zero-value buyer also weakly increases expected GFT, because all costs and values lie in $[0,1]$. These modified accepting sets need not be induced by posted prices; they are used only to upper bound the algorithm's expected GFT.

    Since the accepting subset of $\mathcal U$ is the same for every $\mathcal Z$, its expected total cost under the uniform assignment is $a\mu$. The modified accepting sets therefore give $\expect[\mathcal Z]{\expect{\gft[\matchAt{\ts}]}} \leq\min\{2r+a,3r\}\left(1-\frac{a\mu}{2r+a}\right)$.
    Subtracting this upper bound from Eqn.~\eqref{eq:segmented-population-opt} gives two cases.
    For $0\leq a\leq r$, $\expect[\mathcal Z]{\regAt{\ts}}\geq r\Delta/2+(r-a)(1-\mu)$, while for $r\leq a\leq2r$, $\expect[\mathcal Z]{\regAt{\ts}}\geq r\Delta/2+2r\mu(a-r)/(2r+a)$.
Both additional terms are nonnegative, proving the first claim. If $a=2r$, the total cost of accepting sellers in $\mathcal U$ equals $r(\sellerValue_-+\sellerValue_+)$ for every assignment. The same comparison then gives $\regAt{\ts}\geq\frac r4(3\sellerValue_+-\sellerValue_-) \geq\frac r2\sellerValue_-\geq\frac r6$.
If $a=0$, the expected GFT is at most $2r$, so
    $\regAt{\ts}\geq r(1-\sellerValue_-)\geq r/3$ since $\sellerValue_-\leq2/3$.
\end{proof}

\xhdr{Adversary construction.}
Fix an arbitrary online learning algorithm $\ALG$ using {\SegPriceMech}, and set
$\totalPhaseNum\triangleq\floor{\frac12\log_2\ln\timeHorizon}$.
As in \Cref{sec:mismatch-difficulty}, the adversary maintains an uncertainty interval $\hvecsAt{\ts}$ containing possible values of both unknown costs. Write
$\hvecsAt{\ts}=[\LB{\phi}_{\ts},\UB{\phi}_{\ts}]\subseteq[1/3,2/3]$.
The construction proceeds in at most $\totalPhaseNum$ phases, with $\hvecsPhase{\phaseNum}$ denoting the interval at the beginning of phase $\phaseNum$. Initially, $\hvecsAt{1}=\hvecsPhase{1}=[1/3,2/3]$. At the beginning of phase $\phaseNum$, set the threshold
\begin{align}
    \priceIncreasingThresholdAt{\phaseNum}\triangleq
        \frac{\abs{\hvecsPhase{\phaseNum}}}
        {8\ceil{4\totalPhaseNum/\abs{\hvecsPhase{\phaseNum}}}}.
    \label{eq:segmented-population-phase-parameters}
\end{align}
This threshold remains fixed throughout the phase.

In round $\ts$, let $\sellerPriceAtOf{\ts}{\sellerIndex}$ denote the price offered to seller $\sellerIndex$, and define $\LB{\sellerPrice}_{\ts}\triangleq \min_{\sellerIndex\in\mathcal U}\sellerPriceAtOf{\ts}{\sellerIndex}$ and $\UB{\sellerPrice}_{\ts}\triangleq \max_{\sellerIndex\in\mathcal U}\sellerPriceAtOf{\ts}{\sellerIndex}$.
Since {\SegPriceMech} permits at most two seller prices, every seller in $\mathcal U$ receives either $\LB{\sellerPrice}_{\ts}$ or $\UB{\sellerPrice}_{\ts}$. These two prices may coincide. All traders outside $\mathcal U$ respond according to their fixed types. For sellers in $\mathcal U$, the adversary uses the first applicable rule below, where $\phaseNum$ is the current phase:
\begin{enumerate}[label=\textbf{Rule~\arabic*.}, ref=Rule~\arabic*, leftmargin=*]
    \item\label{rule:segmented-population-accept} If $\LB{\sellerPrice}_{\ts}\geq\LB{\phi}_{\ts}+2\priceIncreasingThresholdAt{\phaseNum}$, all sellers in $\mathcal U$ accept. The adversary updates $\hvecsAt{\ts+1} =[\LB{\phi}_{\ts},\LB{\phi}_{\ts}+\priceIncreasingThresholdAt{\phaseNum}]$
and ends the phase.
    \item\label{rule:segmented-population-reject} Otherwise, if $\UB{\sellerPrice}_{\ts}\leq\UB{\phi}_{\ts}-2\priceIncreasingThresholdAt{\phaseNum}$, all sellers in $\mathcal U$ reject. The adversary updates $\hvecsAt{\ts+1} =[\UB{\phi}_{\ts}-\priceIncreasingThresholdAt{\phaseNum},\UB{\phi}_{\ts}]$
and ends the phase.
    \item\label{rule:segmented-population-split} Otherwise, sellers offered $\LB{\sellerPrice}_{\ts}$ reject and those offered $\UB{\sellerPrice}_{\ts}$ accept. The adversary updates  $\hvecsAt{\ts+1} =\hvecsAt{\ts}\cap [\LB{\sellerPrice}_{\ts}+\priceIncreasingThresholdAt{\phaseNum}, \UB{\sellerPrice}_{\ts}-\priceIncreasingThresholdAt{\phaseNum}]$.
The phase continues unless this rule has been used $\ceil{4\totalPhaseNum/\abs{\hvecsPhase{\phaseNum}}}$ times in the phase, in which case the adversary stops updating the interval.
\end{enumerate}
When \ref{rule:segmented-population-accept} or \ref{rule:segmented-population-reject} ends phase $\phaseNum$, the adversary starts phase $\phaseNum+1$ with $\hvecsPhase{\phaseNum+1}=\hvecsAt{\ts+1}$, unless $\phaseNum=\totalPhaseNum$, in which case it stops updating the interval. We show below that a stopping condition is reached before the horizon ends. If the last update occurs after round $\tau$, the adversary fixes the two costs at $\sellerValue_-\triangleq\min\hvecsAt{\tau+1}$ and $\sellerValue_+\triangleq\max\hvecsAt{\tau+1}$.
The assignment $\mathcal Z$ will be chosen after bounding the cumulative regret. The following lemma ensures that any such assignment reproduces all preceding feedback.

\begin{lemma}
    \label{lem:segmented-population-consistency}
    Every interval maintained by the adversary is nested within the preceding interval and has positive width. After any application of \ref{rule:segmented-population-split} in phase $\phaseNum$, its width exceeds $\abs{\hvecsPhase{\phaseNum}}/4$. 
The final costs $\sellerValue_-$ and $\sellerValue_+$ are 
    consistent with all posted prices and feedback through round $\tau$, for every assignment $\mathcal Z\subseteq\mathcal U$ of size $r$. Consequently, the segment assignments, prices, feedback, and stopping time through round $\tau$ are independent of $\mathcal Z$.
\end{lemma}

\begin{proof}[Proof of \Cref{lem:segmented-population-consistency}]
    Whenever \ref{rule:segmented-population-split} applies, failure of the first two conditions gives $\LB{\sellerPrice}_{\ts}<\LB{\phi}_{\ts}+2\priceIncreasingThresholdAt{\phaseNum}$ and $\UB{\sellerPrice}_{\ts}>\UB{\phi}_{\ts}-2\priceIncreasingThresholdAt{\phaseNum}$.
Thus, \ref{rule:segmented-population-split} increases the lower endpoint by less than $3\priceIncreasingThresholdAt{\phaseNum}$ and decreases the upper endpoint by less than $3\priceIncreasingThresholdAt{\phaseNum}$. After $s\leq\ceil{4\totalPhaseNum/\abs{\hvecsPhase{\phaseNum}}}$ applications in this phase, the interval width is therefore greater than $\abs{\hvecsPhase{\phaseNum}}-6s\priceIncreasingThresholdAt{\phaseNum} \geq\abs{\hvecsPhase{\phaseNum}} -\frac34\abs{\hvecsPhase{\phaseNum}} =\frac{\abs{\hvecsPhase{\phaseNum}}}{4}$,
where the inequality follows from Eqn.~\eqref{eq:segmented-population-phase-parameters}. This positive lower bound also guarantees that each such update leaves a nonempty interval.
    In particular, $\LB{\sellerPrice}_{\ts}=\UB{\sellerPrice}_{\ts}$ cannot occur under \ref{rule:segmented-population-split}, because failure of the first two conditions with equal prices would require an interval width below $4\priceIncreasingThresholdAt{\phaseNum}$. Before an application of either of the first two rules, the width exceeds $\abs{\hvecsPhase{\phaseNum}}/4>\priceIncreasingThresholdAt{\phaseNum}$, so its retained interval is also nested and nonempty.

    Under \ref{rule:segmented-population-accept}, the retained upper endpoint is strictly below $\LB{\sellerPrice}_{\ts}$, so every cost in the retained interval accepts every price offered to $\mathcal U$. Under \ref{rule:segmented-population-reject}, the retained lower endpoint is strictly above $\UB{\sellerPrice}_{\ts}$, so every such cost rejects every offered price. Under \ref{rule:segmented-population-split}, every retained cost is strictly above $\LB{\sellerPrice}_{\ts}$ and strictly below $\UB{\sellerPrice}_{\ts}$, giving exactly the specified feedback. Because these are the only prices offered to $\mathcal U$, no feedback distinguishes the two final cost levels. Nestedness makes every earlier response consistent with both final costs.

    The feedback of all traders outside $\mathcal U$ is fixed by their known types. Since $\ALG$ is deterministic and observes only accept/reject decisions, induction over rounds shows that it chooses the same segments and prices on every final assignment $\mathcal Z$ and receives exactly the constructed feedback. The unobserved matching randomness does not affect its subsequent decisions. Hence the entire interaction through round $\tau$, including the stopping time, is independent of $\mathcal Z$.
\end{proof}

\begin{lemma}
    \label{lem:segmented-population-stopping}
    For $\timeHorizon\geq\ceil{e^{256}}$, the adversary stops updating the uncertainty interval before round $\timeHorizon$.
\end{lemma}

\begin{proof}[Proof of \Cref{lem:segmented-population-stopping}]
    Each phase $\phaseNum$ contains at most $\ceil{4\totalPhaseNum/\abs{\hvecsPhase{\phaseNum}}}$ rounds: it ends with one of the first two rules, or when \ref{rule:segmented-population-split} reaches this number of applications. Whenever phase $\phaseNum+1$ is reached, $\abs{\hvecsPhase{\phaseNum+1}}=\priceIncreasingThresholdAt{\phaseNum}$. Since $\abs{\hvecsPhase{\phaseNum}}\leq1/3$ and $\totalPhaseNum\geq1$, Eqn.~\eqref{eq:segmented-population-phase-parameters} gives $\ceil{\frac{4\totalPhaseNum}{\abs{\hvecsPhase{\phaseNum}}}} \leq\frac{5\totalPhaseNum}{\abs{\hvecsPhase{\phaseNum}}}$ and $\frac1{\abs{\hvecsPhase{\phaseNum+1}}} \leq40\totalPhaseNum\left(\frac1{\abs{\hvecsPhase{\phaseNum}}}\right)^2$.
Starting from $\abs{\hvecsPhase{1}}=1/3$, induction gives $\frac1{\abs{\hvecsPhase{\phaseNum}}} \leq(40\totalPhaseNum)^{2^{\phaseNum-1}-1}3^{2^{\phaseNum-1}} \leq(120\totalPhaseNum)^{2^{\phaseNum-1}}$.
Combining these bounds, every reached phase $\phaseNum$ contains at most $5\totalPhaseNum(120\totalPhaseNum)^{2^{\phaseNum-1}}$ rounds. Since at most $\totalPhaseNum$ phases are reached, summing these upper bounds gives
    \begin{align}
        \tau
        &\leq5\totalPhaseNum\sum_{\phaseNum=1}^{\totalPhaseNum} (120\totalPhaseNum)^{2^{\phaseNum-1}}      \leq5\totalPhaseNum\cdot\totalPhaseNum
            (120\totalPhaseNum)^{2^{\totalPhaseNum-1}}    =5\totalPhaseNum^2(120\totalPhaseNum)^{2^{\totalPhaseNum-1}}.
        \label{eq:segmented-population-stopping-time}
    \end{align}
    The second inequality bounds each of the $\totalPhaseNum$ summands by its value at $\phaseNum=\totalPhaseNum$, since $120\totalPhaseNum>1$ and $2^{\phaseNum-1}\leq2^{\totalPhaseNum-1}$.
    To compare this bound with the horizon, put $x\triangleq\ln\timeHorizon\geq256$. The definition of $\totalPhaseNum$ gives $\totalPhaseNum\leq\sqrt{x}/2$ and $2^{\totalPhaseNum-1}\leq\sqrt{x}/2$. Moreover, $\ln(60\sqrt{x})\leq\sqrt{x}/2$ and $\ln(5x/4)\leq x/4$ for $x\geq256$.
Both inequalities hold at $x=256$, and the difference between the right and left sides increases thereafter. Taking logarithms in Eqn.~\eqref{eq:segmented-population-stopping-time} therefore yields $\ln\tau \leq\ln(5\totalPhaseNum^2)+2^{\totalPhaseNum-1}\ln(120\totalPhaseNum) \leq x/4+x/4=x/2$.
Since $x=\ln\timeHorizon$, exponentiating this inequality gives $\tau\leq e^{x/2} =e^{(\ln\timeHorizon)/2} =\sqrt{\timeHorizon}<\timeHorizon$,  where the last inequality uses $\timeHorizon>1$.
\end{proof}

\xhdr{Cumulative regret and a fixed instance.}
We now combine the immediate-regret bounds with the adversary's stopping rule. The averaging over cost assignments is applied to the cumulative regret, so a single assignment suffices for the whole interaction.

\begin{proof}[Proof of \Cref{thm:segmented-population-lower-bound}]
    Fix an arbitrary online learning algorithm $\ALG$ using {\SegPriceMech}. By \Cref{lem:segmented-population-consistency,lem:segmented-population-stopping}, the adversary stops at some $\tau<\timeHorizon$ and chooses two distinct costs in $[1/3,2/3]$. Average first over a uniformly chosen assignment $\mathcal Z$. Every assignment produces the same segments, prices, and feedback through round $\tau$, so \Cref{lem:segmented-population-round-regret} applies throughout this interaction.

    If the adversary stops after $\totalPhaseNum$ phases ending with \ref{rule:segmented-population-accept} or \ref{rule:segmented-population-reject}, the final round of each phase has immediate regret at least $r/6$. Hence,
    $
        \expect[\mathcal Z]{\sum_{\ts=1}^{\tau}\regAt{\ts}}
        \geq r\totalPhaseNum/6$.
    Otherwise, the adversary stops after $\ceil{4\totalPhaseNum/\abs{\hvecsPhase{\phaseNum}}}$ applications of \ref{rule:segmented-population-split} in some phase $\phaseNum$. By \Cref{lem:segmented-population-consistency}, the final cost gap satisfies $\Delta>\abs{\hvecsPhase{\phaseNum}}/4$. Each such round has expected immediate regret at least $r\Delta/2$ by \Cref{lem:segmented-population-round-regret}. Thus,
    $
        \expect[\mathcal Z]{\sum_{\ts=1}^{\tau}\regAt{\ts}}
        \geq\ceil{\frac{4\totalPhaseNum}{\abs{\hvecsPhase{\phaseNum}}}}
            \frac{r\Delta}{2}>\frac{4\totalPhaseNum}{\abs{\hvecsPhase{\phaseNum}}}
            \frac{r\abs{\hvecsPhase{\phaseNum}}}{8}
        =\frac{r\totalPhaseNum}{2}
        \geq\frac{r\totalPhaseNum}{6}$.
    The other rounds have nonnegative regret because $\OPT$ is the optimum over all feasible matchings.

    In either case, at least one fixed assignment $\mathcal Z$ has cumulative regret through round $\tau$ of at least $r\totalPhaseNum/6$, with expectation only over the matching rule. Choose that assignment and fix all traders' types before running $\ALG$. By \Cref{lem:segmented-population-consistency}, this fixed instance reproduces the entire constructed interaction through round $\tau$. The remaining rounds have nonnegative regret, so $\reg[\ALG]\geq r\totalPhaseNum/6$, which is Eqn.~\eqref{eq:segmented-population-lower-bound}. Finally, $r\geq\buyerNum/8$ for $\buyerNum\geq4$ and $\totalPhaseNum=\Theta(\log\log\timeHorizon)$ in the stated horizon range.
\end{proof}

\section{Missing Technical Details and Analysis for Contextual Markets}
\label{app:extension}

In this section, we provide all missing technical details of two extension models and the formal proofs of \Cref{thm:extension:contextual-gft,thm:extension:contextual-profit} in \Cref{app:extension:gft,app:extension:profit}, respectively.

\subsection{Contextual GFT Maximization}
\label{app:extension:gft}

To prove \Cref{thm:extension:contextual-gft}, we first collect the setup and the geometric tools that both analyses use. The analysis for contextual bilateral trade and one-to-many markets under {\SinglePriceMech} is presented in \Cref{app:extension:gft:one-to-many}, and the analysis for contextual two-sided markets under {\SegPriceMech} is presented in \Cref{app:extension:gft:two-sided}.

\xhdr{Notation.}
For a compact convex set $\querySet\subseteq\reals^\dimension$ and a
vector $\context\in\reals^\dimension$, we write
$\width{\querySet}{\context}\deq
\max_{\genericVec\in\querySet}\lrangle{\genericVec,\context}
-\min_{\genericVec\in\querySet}\lrangle{\genericVec,\context}$
for the directional width of $\querySet$ along $\context$, and we
write $\vol{\cdot}$ for the $\dimension$-dimensional volume. A convex
body is a compact convex set with positive volume. In the tools below,
$u\in\reals^\dimension$ is a fixed unit vector: directional widths are
taken along $u$, and caps and slabs are bounded by hyperplanes
perpendicular to $u$. 

\xhdr{Uncertainty sets and directional bounds.}
Each seller $\sellerIndex\in\sellerSet$ maintains a convex uncertainty
set $\sellerVecsOfAt{\sellerIndex}{\ts}\subseteq\ball_\dimension$
containing its feature vector $\sellerVecAt{\sellerIndex}$.
Similarly, each buyer $\buyerIndex\in\buyerSet$ maintains
$\buyerVecsOfAt{\buyerIndex}{\ts}\subseteq\ball_\dimension$
containing $\buyerVecAt{\buyerIndex}$. The sets are initially
$\ball_\dimension$ and are updated by intersecting them with the
halfspaces implied by the feedback.

If $\contextAt{\ts}=0$, the algorithm posts price zero to every
trader and leaves the uncertainty sets unchanged; the regret of that
round is zero. Otherwise, write
$
    \lambda_\ts\triangleq\norm{\contextAt{\ts}}$ and
    $
    u_\ts\triangleq\contextAt{\ts}/\lambda_\ts$.
Since all costs and values are nonnegative, we intersect each uncertainty set
with $\set{\hvec:\lrangle{\hvec,u_\ts}\geq0}$, which contains
its trader's feature vector. 
We keep the notation of each set for the intersected set, and the
tools below are applied with $u=u_\ts$.

For each seller and buyer, respectively, define
\begin{align*}
    \sellerValueLbOfAt{\sellerIndex}{\ts}
    &\triangleq
    \min_{\hvec\in\sellerVecsOfAt{\sellerIndex}{\ts}}
        \lrangle{\hvec,\contextAt{\ts}},
    &\sellerValueUbOfAt{\sellerIndex}{\ts}
    &\triangleq
    \max_{\hvec\in\sellerVecsOfAt{\sellerIndex}{\ts}}
        \lrangle{\hvec,\contextAt{\ts}},\\
    \buyerValueLbOfAt{\buyerIndex}{\ts}
    &\triangleq
    \min_{\hvec\in\buyerVecsOfAt{\buyerIndex}{\ts}}
        \lrangle{\hvec,\contextAt{\ts}},
    &\buyerValueUbOfAt{\buyerIndex}{\ts}
    &\triangleq
    \max_{\hvec\in\buyerVecsOfAt{\buyerIndex}{\ts}}
        \lrangle{\hvec,\contextAt{\ts}}.
\end{align*}
For either side, let $\querySet$ denote a trader's current uncertainty set,
and let $w\triangleq\width{\querySet}{u_\ts}$ be its directional width.
Since $\querySet\subseteq\ball_\dimension$, the range of its projection
onto $u_\ts$ has width at most $2$. After intersecting with the
nonnegative halfspace, the directional bounds lie in $[0,1]$, so
$0\leq w\leq1$. The corresponding cost or value bounds lie in
$[0,\lambda_\ts]$.
The difference between a trader's upper and lower bounds is
$\lambda_\ts w$. For all algorithms in this section, we will keep that every uncertainty set always contains its trader's feature vector.

\xhdr{Cap geometry.}
We first recall a geometric bound for caps of convex bodies. For a
convex body $\querySet\subseteq\reals^\dimension$ with positive directional
width $L$, let $F(s)$ denote the fraction of its volume in a cap of
thickness $s$ measured from a supporting hyperplane. For
$0\leq s\leq L$, convexity gives
\begin{align}
    \left(\frac{s}{L}\right)^\dimension
    \leq F(s)
    \leq\frac{\dimension s}{L}.
    \label{eq:contextual-quantile-cap-bounds}
\end{align}
The lower estimate follows by contracting $\querySet$ toward a supporting point on the cap side. For the upper estimate, contract $\querySet$ toward a supporting point on the opposite side, leaving a subset of volume $(1-s/L)^\dimension\vol{\querySet}$ outside the cap. 
The same contraction argument also shows that
\begin{align}
    F(s_1)\geq
    \left(\frac{s_1}{s_2}\right)^\dimension F(s_2)
    \quad\text{whenever }0\leq s_1\leq s_2\leq L.
    \label{eq:contextual-quantile-cap-comparison}
\end{align}
Moreover, every slab of thickness $s\geq0$ contains at most a
fraction $\dimension s/L$ of $\vol{\querySet}$:
\begin{align}
    \vol{\querySet\cap\set{\genericVec:a\leq\lrangle{\genericVec,u}\leq a+s}}
    \leq\frac{\dimension s}{L}\vol{\querySet}
    \qquad\text{for every }a\in\reals,
    \label{eq:extension:slab-bound}
\end{align}
since joining any section of $\querySet$ perpendicular to $u$ with a point of $\querySet$ on each supporting hyperplane gives two cones in $\querySet$ of total volume $L/\dimension$ times the volume of the section.

\xhdr{Padded-volume potential.}
For a nonempty compact convex set $\querySet\subseteq\ball_\dimension$
and a padding radius $\padRadius>0$, the padded set
$\querySet+\padRadius\ball_\dimension$ is a convex body, even when
$\querySet$ is lower-dimensional. Define
\begin{align}
    \padPotential{\padRadius}{\querySet}\deq
    \log\frac{\vol{\querySet+\padRadius\ball_\dimension}}
             {\vol{\padRadius\ball_\dimension}}
    \in\bigl[0,\dimension\log(1+1/\padRadius)\bigr].
    \label{eq:extension:padded-potential}
\end{align}
The range holds because the padded set contains the ball
$\genericVec+\padRadius\ball_\dimension$ for every
$\genericVec\in\querySet$ and is contained in
$(1+\padRadius)\ball_\dimension$. The potential does not increase when $\querySet$ is replaced by its subset. 
In both algorithms, every uncertainty set only shrinks over the rounds, so for each trader and each fixed padding radius, the potential of the trader's uncertainty set never increases.

\begin{lemma}[\citealp{LLS-21}]\label{lem:extension:median-cut}
Let $\querySet\subseteq\ball_\dimension$ be a nonempty compact convex set
with directional width $\queryWidth\in(0,1]$. Let $\padIndex\geq0$ be
the unique integer with $2^{-\padIndex-1}<\queryWidth\leq2^{-\padIndex}$,
set $\medianPad{\padIndex}\deq2^{-\padIndex}/(8\dimension)$, and let
$a\in\reals$ satisfy
$
    \vol{\set{\genericVec\in\querySet+\medianPad{\padIndex}\ball_\dimension:
        \lrangle{\genericVec,u}\leq a}}
    =\frac12\vol{\querySet+\medianPad{\padIndex}\ball_\dimension}$.
If $\querySet'$ is
$\querySet\cap\set{\genericVec:\lrangle{\genericVec,u}\leq a}$ or
$\querySet\cap\set{\genericVec:\lrangle{\genericVec,u}\geq a}$, then
$\vol{\querySet'+\medianPad{\padIndex}\ball_\dimension}
\leq\frac34\vol{\querySet+\medianPad{\padIndex}\ball_\dimension}$.
\end{lemma}

Set
\begin{align}
    W_{-1}\triangleq1,
    \qquad
    W_h\triangleq\exp(-2^h)\quad(h\geq0).
    \label{eq:contextual-quantile-scales}
\end{align}
For a nonempty compact convex set $\querySet\subseteq\ball_\dimension$
with directional width $0<w\leq1$, choose the scale $h=-1$ if
$w>e^{-1}$; otherwise, choose the unique $h\geq0$ such that
$W_h^2<w\leq W_h$. Define
\begin{align}
    \rho_h\triangleq\frac{\sqrt{W_h}}8,
    \qquad
    z_h\triangleq\frac{W_h^{5/2}}{512\dimension^2},
    \qquad
    C\triangleq \querySet+z_h\ball_\dimension.
    \label{eq:contextual-quantile-parameters}
\end{align}

\begin{lemma}\label{lem:contextual-quantile-thickness}
Fix a nonempty convex set $\querySet \subseteq\ball_\dimension$ with
positive directional width $w$, and select $h$ by
\Cref{eq:contextual-quantile-scales}. Let $s$ be the thickness of
the cap of volume fraction $\rho_h$ in
$C=\querySet+z_h\ball_\dimension$, measured from the corresponding
supporting hyperplane. Then the quantile lies strictly inside the support interval of $\querySet$, and
$\frac{\rho_h(w+2z_h)}{\dimension}
\leq s<w+z_h\leq2w$,
$\frac{z_h}{s}\leq\frac1{16\dimension}$,
and $s\geq\frac{w^{3/2}}{8\dimension}$. Moreover, if $F$ is the cap-volume function of $C$, then
\begin{align}
    F(s+z_h)&\leq2\rho_h,
    &F(s-\eta)&\geq\left(1-\frac\eta s\right)^\dimension\rho_h
           \quad(0\leq\eta\leq s).
    \label{eq:contextual-quantile-cap-stability}
\end{align}
\end{lemma}
\begin{proof}
This lemma statement and its proof is adopted from the technical lemmas in \citet{LLS-21}. We provide its formal proof for the completeness.
Let $L=w+2z_h$ be the directional width of $C$.
By \cref{eq:contextual-quantile-cap-bounds},
$s\geq\rho_h L/\dimension\geq\rho_hw/\dimension$.
For $h\geq0$, the scale condition $w>W_h^2$ yields
$
    \frac{\dimension z_h}{L}
    \leq\frac{\dimension z_h}{w}
    <\frac{\sqrt{W_h}}{512\dimension}
    =\frac{\rho_h}{64\dimension}$.
For $h=-1$, we have $w>e^{-1}$ and
$\dimension z_h/L<e/(512\dimension)<\rho_h/8$.
Consequently, in every case the lower cap of $C$ ending at the
lower endpoint of $\querySet$ has volume fraction less than $\rho_h$.
The upper cap starting at the upper endpoint of $\querySet$ likewise has
fraction less than $\rho_h$. Since $\rho_h\leq1/8$, both upper-
and lower-tail quantiles lie strictly inside the support interval
of $\querySet$, proving $z_h<s<w+z_h$. In particular, $z_h<w$.

By $s\geq\rho_hw/\dimension$, for $h\geq0$,
$
    \frac{z_h}{s}
    \leq\frac{\dimension z_h}{\rho_hw}
    <\frac1{64\dimension}$,
whereas for $h=-1$ the ratio is less than
$e/(64\dimension)<1/(16\dimension)$.
Furthermore, $W_h\geq w$ in every scale and hence
$\rho_h=\sqrt{W_h}/8\geq\sqrt w/8$.
This proves $s\geq w^{3/2}/(8\dimension)$ and
$s<w+z_h\leq2w$.

Applying \cref{eq:contextual-quantile-cap-comparison} to
$(s,s+z_h)$ gives
$
    F(s+z_h)
    \leq\left(1+\frac{z_h}{s}\right)^\dimension F(s)
    \leq e^{1/16}\rho_h<2\rho_h$.
Applying the same estimate to $(s-\eta,s)$ gives the other
inequality in \cref{eq:contextual-quantile-cap-stability}.
\end{proof}

\subsubsection{Contextual Bilateral Trade and One-to-Many Markets}
\label{app:extension:gft:one-to-many}

Our algorithm, {\ContextualTradeThenRefine}, extends {\OneToManyGFT} to contextual markets. As in \Cref{sec:uniform-pricing}, we consider a single seller and either a single buyer in  bilateral trade or multiple buyers in one-to-many markets, the case of a single buyer and multiple sellers follows by symmetry.
We omit the seller's index, write $\sellerVec$ for the seller's feature vector, and denote the seller's cost by $\sellerValueAt{\ts}=\lrangle{\sellerVec,\contextAt{\ts}}$.

{\OneToManyGFT} first searches for a price that induces a trade, and then refines the price since under the trading price the seller may not trade with the highest-value buyer. In contextual markets, a price that induces a trade under one context need not do so under another.
Each round, {\ContextualTradeThenRefine} therefore uses the cost and value bounds implied by the current context to decide whether to search for a price that induces a trade or to refine the price.

\xhdr{Uncertainty sets.}
The algorithm uses the uncertainty sets, bounds and directional
widths defined at the start of \Cref{app:extension:gft}, with the
seller's index omitted: the seller's uncertainty set is
$\sellerVecsAt{\ts}$, with cost bounds $\sellerValueLbAt{\ts}$ and
$\sellerValueUbAt{\ts}$.

\xhdr{Median price for trade search.}
Trade search uses the padded-volume technique of \citet{LLS-21}. Let
$\querySet$ be the uncertainty set selected for trade search below,
suppose that its directional width
$\queryWidth\deq\width{\querySet}{\contextDirAt{\ts}}$ is positive, and
let $\padIndex$ and $\medianPad{\padIndex}$ be the scale and the padding
radius in \Cref{lem:extension:median-cut}. The
\emph{median price} $\postPriceAt{\ts}$ bisects the volume of the
padded set $\querySet+\medianPad{\padIndex}\ball_\dimension$ along
$\contextAt{\ts}$:
\begin{align}
    \vol{\set{\genericVec\in\querySet+\medianPad{\padIndex}\ball_\dimension:
        \lrangle{\genericVec,\contextAt{\ts}}\leq\postPriceAt{\ts}}}
    =\frac12\vol{\querySet+\medianPad{\padIndex}\ball_\dimension}.
    \label{eq:extension:ttr-median}
\end{align}
The padded set is a convex body, so the median price is unique. By
\Cref{lem:extension:median-cut} with $u=\contextDirAt{\ts}$ and
$a=\postPriceAt{\ts}/\contextNormAt{\ts}$, either accept/reject decision
of the corresponding trader leaves at most three quarters of the padded
volume.

\xhdr{Quantile price for price refinement.}
Price refinement uses a cautious lower-tail quantile of the selected
buyer's padded uncertainty set. Let $\querySet$ be this set and
$\queryWidth\deq\width{\querySet}{\contextDirAt{\ts}}$ its directional
width. If $\queryWidth\leq1/\timeHorizon$, the \emph{quantile price}
$\postPriceAt{\ts}$ is the buyer's lower value bound. Otherwise, let
$\quantileScale$, $\quantileFraction{\quantileScale}$ and
$\padPara{\quantileScale}$ be the scale and the parameters of
$\querySet$ in
\cref{eq:contextual-quantile-scales,eq:contextual-quantile-parameters},
with $u=\contextDirAt{\ts}$. The quantile price $\postPriceAt{\ts}$ cuts
off the lower tail of the padded set with volume fraction
$\quantileFraction{\quantileScale}$:
\begin{align}
    \vol{\set{\genericVec\in\querySet+\padPara{\quantileScale}\ball_\dimension:
        \lrangle{\genericVec,\contextAt{\ts}}\leq\postPriceAt{\ts}}}
    =\quantileFraction{\quantileScale}
     \vol{\querySet+\padPara{\quantileScale}\ball_\dimension}.
    \label{eq:extension:ttr-quantile}
\end{align}
Like conservative binary search in {\OneToManyGFT}, the quantile price
is cautious: only a small fraction of the selected buyer's padded
uncertainty set lies below the price. If the selected buyer accepts, the immediate regret is at most the width of the
buyer's value bounds. If the selected buyer rejects, the immediate
regret can be large, but the reduction in uncertainty is also
substantial: the padded volume shrinks by a factor of order
$\sqrt{\quantileWidthScale{\quantileScale}}$
(\Cref{lem:contextual-quantile-thickness}).

\xhdr{Selecting the price.}
In each round with a nonzero context, the algorithm selects a buyer $\selectedBuyerAt{\ts}$ with the largest upper value bound. 
If $\buyerValueUbOfAt{\selectedBuyerAt{\ts}}{\ts}\leq\sellerValueLbAt{\ts}$, no trade has positive GFT. Otherwise, if $\buyerValueLbOfAt{\selectedBuyerAt{\ts}}{\ts}<\sellerValueUbAt{\ts}$, this means it does not yet guarantee that the seller and the selected buyer accept a common price, so we let  the round be a \emph{trade-search round}. 
The algorithm then posts the median price for the wider of $\sellerVecsAt{\ts}$ and $\buyerVecsOfAt{\selectedBuyerAt{\ts}}{\ts}$, choosing $\sellerVecsAt{\ts}$ in a tie. If $\buyerValueLbOfAt{\selectedBuyerAt{\ts}}{\ts}\geq \sellerValueUbAt{\ts}$, the round is a \emph{refinement round}: the seller and the selected buyer both accept
every price between $\sellerValueUbAt{\ts}$ and
$\buyerValueLbOfAt{\selectedBuyerAt{\ts}}{\ts}$, and the algorithm posts
$\buyerValueLbOfAt{\selectedBuyerAt{\ts}}{\ts}$ in bilateral trade and
the quantile price for $\buyerVecsOfAt{\selectedBuyerAt{\ts}}{\ts}$
otherwise. Trade-search rounds correspond to exploration phase I of
{\OneToManyGFT}, and refinement rounds to exploration phase II and the
exploitation phase; unlike {\OneToManyGFT}, the algorithm returns to
trade search whenever the current context requires it. 
After posting the price, the algorithm intersects each uncertainty set with the closed halfspace implied by the trader's decision. 
We retain the boundary to keep the set compact, even when rejection implies a strict inequality. The algorithm is summarized in
\Cref{alg:contextual-trade-then-refine}.

\begin{algorithm}[t]
\caption{\ContextualTradeThenRefine}
\label[algorithm]{alg:contextual-trade-then-refine}
\SetAlgoLined
\DontPrintSemicolon
\KwIn{Time horizon $\timeHorizon$}
\KwOut{{\SinglePriceMech} in each round}
Initialize $\sellerVecsAt{1}\gets\ball_\dimension$ and
$\buyerVecsOfAt{\buyerIndex}{1}\gets\ball_\dimension$ for every buyer
$\buyerIndex\in\buyerSet$, and fix a tie-breaking order of buyers\;

\For{$\ts\in[\timeHorizon]$}{
    Observe $\contextAt{\ts}$\;

    \eIf{$\contextAt{\ts}=0$}{
        Post $\postPriceAt{\ts}\gets0$\;
        Set $\sellerVecsAt{\ts+1}\gets\sellerVecsAt{\ts}$ and
        $\buyerVecsOfAt{\buyerIndex}{\ts+1}\gets\buyerVecsOfAt{\buyerIndex}{\ts}$
        for every buyer $\buyerIndex\in\buyerSet$\;
    }{
        Set $\contextNormAt{\ts}\gets\norm{\contextAt{\ts}}$ and
        $\contextDirAt{\ts}\gets\contextAt{\ts}/\contextNormAt{\ts}$\;

        Set $\sellerVecsAt{\ts}\gets\sellerVecsAt{\ts}\cap
        \set{\genericVec:\lrangle{\genericVec,\contextDirAt{\ts}}\geq0}$\;
        Set $\buyerVecsOfAt{\buyerIndex}{\ts}\gets
        \buyerVecsOfAt{\buyerIndex}{\ts}\cap
        \set{\genericVec:\lrangle{\genericVec,\contextDirAt{\ts}}\geq0}$
        for every buyer $\buyerIndex\in\buyerSet$\;
        Compute the cost and value bounds from these updated uncertainty sets\;

        Select $\selectedBuyerAt{\ts}\in
        \argmax_{\buyerIndex\in\buyerSet}
        \buyerValueUbOfAt{\buyerIndex}{\ts}$ using the tie-breaking
        order\;

        \uIf{$\buyerValueUbOfAt{\selectedBuyerAt{\ts}}{\ts}
              \leq\sellerValueLbAt{\ts}$}{
            Post $\postPriceAt{\ts}\gets
            \buyerValueUbOfAt{\selectedBuyerAt{\ts}}{\ts}$\;
        }
        \uElseIf{$\buyerValueLbOfAt{\selectedBuyerAt{\ts}}{\ts}
                  <\sellerValueUbAt{\ts}$}{
            Post the median price \cref{eq:extension:ttr-median} for
            the wider of $\sellerVecsAt{\ts}$ and
            $\buyerVecsOfAt{\selectedBuyerAt{\ts}}{\ts}$
            \tcp*[r]{trade search}
        }
        \uElseIf{$\buyerNum=1$}{
            Post $\postPriceAt{\ts}\gets
            \buyerValueLbOfAt{\selectedBuyerAt{\ts}}{\ts}$
            \tcp*[r]{refinement}
        }
        \Else{
            Post the quantile price \cref{eq:extension:ttr-quantile}
            for $\buyerVecsOfAt{\selectedBuyerAt{\ts}}{\ts}$
            \tcp*[r]{refinement}
        }

        Observe all accept/reject decisions\;
        \lIf{the seller accepts}{
            $\sellerVecsAt{\ts+1}\gets\sellerVecsAt{\ts}\cap
            \set{\genericVec:\lrangle{\genericVec,\contextAt{\ts}}\leq\postPriceAt{\ts}}$
        }
        \lElse{
            $\sellerVecsAt{\ts+1}\gets\sellerVecsAt{\ts}\cap
            \set{\genericVec:\lrangle{\genericVec,\contextAt{\ts}}\geq\postPriceAt{\ts}}$
        }
        \ForEach{buyer $\buyerIndex\in\buyerSet$}{
            \lIf{buyer $\buyerIndex$ accepts}{
                $\buyerVecsOfAt{\buyerIndex}{\ts+1}\gets
                \buyerVecsOfAt{\buyerIndex}{\ts}\cap
                \set{\genericVec:\lrangle{\genericVec,\contextAt{\ts}}\geq\postPriceAt{\ts}}$
            }
            \lElse{
                $\buyerVecsOfAt{\buyerIndex}{\ts+1}\gets
                \buyerVecsOfAt{\buyerIndex}{\ts}\cap
                \set{\genericVec:\lrangle{\genericVec,\contextAt{\ts}}\leq\postPriceAt{\ts}}$
            }
        }
    }
}
\end{algorithm}

\xhdr{Per-round regret.}
We next bound the immediate regret in each case of the algorithm.

\begin{lemma}\label{lem:extension:ttr-round-regret}
Consider a round $\ts$ with a nonzero context, and let $\matchAt{\ts}$
be any maximum-cardinality matching among accepting traders.
\begin{enumerate}[label=(\roman*)]
    \item\label{item:ttr-no-gain} If
    $\buyerValueUbOfAt{\selectedBuyerAt{\ts}}{\ts}\leq\sellerValueLbAt{\ts}$,
    then $\OPT_\ts=\gft[\matchAt{\ts}]$.
    \item\label{item:ttr-search} In a trade-search round,
    $\OPT_\ts-\gft[\matchAt{\ts}]<2\queryWidth$, where $\queryWidth$ is
    the directional width that determines the median price.
    \item\label{item:ttr-refine} In a refinement round, the seller
    accepts. In bilateral trade, $\OPT_\ts=\gft[\matchAt{\ts}]$. In
    multiple-buyers setting, $\OPT_\ts-\gft[\matchAt{\ts}]\leq \width{\buyerVecsOfAt{\selectedBuyerAt{\ts}}{\ts}}{\contextDirAt{\ts}}$
    if the selected buyer accepts, and $\OPT_\ts-\gft[\matchAt{\ts}]\leq1$ otherwise. 
    If $\width{\buyerVecsOfAt{\selectedBuyerAt{\ts}}{\ts}}{\contextDirAt{\ts}} \leq1/\timeHorizon$, the selected buyer always accepts.
\end{enumerate}
\end{lemma}
\begin{proof}
With a single seller,
$\OPT_\ts=\max\{0,\buyerValueAt{(1),\ts}-\sellerValueAt{\ts}\}$, and
every realized trade has nonnegative GFT because the seller's cost is
at most $\postPriceAt{\ts}$ and the matched buyer's value is at least
$\postPriceAt{\ts}$. Every cost and value lies between its bounds by the definition of our update rules. In particular, $\buyerValueAt{(1),\ts}\leq\buyerValueUbOfAt{\selectedBuyerAt{\ts}}{\ts}$ by the choice of $\selectedBuyerAt{\ts}$.

For \ref{item:ttr-no-gain},
$\buyerValueAt{(1),\ts}\leq\buyerValueUbOfAt{\selectedBuyerAt{\ts}}{\ts}
\leq\sellerValueLbAt{\ts}\leq\sellerValueAt{\ts}$, so
$\OPT_\ts=0\leq\gft[\matchAt{\ts}]$, and equality holds because no
matching exceeds the benchmark.

For \ref{item:ttr-search}, a trade-search round has
$\buyerValueUbOfAt{\selectedBuyerAt{\ts}}{\ts}>\sellerValueLbAt{\ts}$ and $\buyerValueLbOfAt{\selectedBuyerAt{\ts}}{\ts}<\sellerValueUbAt{\ts}$, so $\OPT_\ts\leq\buyerValueUbOfAt{\selectedBuyerAt{\ts}}{\ts} -\sellerValueLbAt{\ts}$. 
We have $ \OPT_\ts-\gft[\matchAt{\ts}]
    \leq\buyerValueUbOfAt{\selectedBuyerAt{\ts}}{\ts}-\sellerValueLbAt{\ts}
    <\bigl(\buyerValueUbOfAt{\selectedBuyerAt{\ts}}{\ts}
          -\buyerValueLbOfAt{\selectedBuyerAt{\ts}}{\ts}\bigr)
     +\bigl(\sellerValueUbAt{\ts}-\sellerValueLbAt{\ts}\bigr)
    \leq2\queryWidth$,
where the last step holds because each difference is $\contextNormAt{\ts}\leq1$ times a directional width of at most $\queryWidth$.

For \ref{item:ttr-refine}, the posted price lies in
$[\buyerValueLbOfAt{\selectedBuyerAt{\ts}}{\ts},
\buyerValueUbOfAt{\selectedBuyerAt{\ts}}{\ts}]$, by
\Cref{lem:contextual-quantile-thickness} or because it equals the lower
value bound. Hence
$\postPriceAt{\ts}\geq\buyerValueLbOfAt{\selectedBuyerAt{\ts}}{\ts}
\geq\sellerValueUbAt{\ts}\geq\sellerValueAt{\ts}$, so the seller accepts,
and $\OPT_\ts=\buyerValueAt{(1),\ts}-\sellerValueAt{\ts}$ because
$\buyerValueAt{(1),\ts}\geq\buyerValueLbOfAt{\selectedBuyerAt{\ts}}{\ts}
\geq\sellerValueAt{\ts}$. In bilateral trade, the buyer accepts the
price $\buyerValueLbOfAt{\selectedBuyerAt{\ts}}{\ts}$, and the trade
attains $\OPT_\ts$. With multiple buyers, if the selected buyer accepts,
the matched buyer's value is at least $\postPriceAt{\ts}$, so
$\OPT_\ts-\gft[\matchAt{\ts}]\leq\buyerValueAt{(1),\ts}-\postPriceAt{\ts}
\leq\buyerValueUbOfAt{\selectedBuyerAt{\ts}}{\ts}
-\buyerValueLbOfAt{\selectedBuyerAt{\ts}}{\ts}
\leq\width{\buyerVecsOfAt{\selectedBuyerAt{\ts}}{\ts}}{\contextDirAt{\ts}}$;
otherwise, the bound of one holds because all costs and values lie in
$[0,1]$. Finally, if the directional width is at most $1/\timeHorizon$,
the posted price equals
$\buyerValueLbOfAt{\selectedBuyerAt{\ts}}{\ts}
\leq\buyerValueAt{\selectedBuyerAt{\ts},\ts}$, so the selected buyer
accepts.
\end{proof}

\xhdr{Regret analysis.}
We count the rounds of each type with the padded-volume potential
\cref{eq:extension:padded-potential}.

\begin{proof}[Proof of the contextual bilateral-trade and one-to-many claims of \Cref{thm:extension:contextual-gft}]
Since
\Cref{lem:extension:ttr-round-regret} holds for every
maximum-cardinality matching, its bounds also hold for the immediate
regret $\regAt{\ts}$ under the randomized matching rule, and the bounds
below hold for every sequence of maximum-cardinality matchings. Rounds
with a zero context have zero regret, and so do rounds with
$\buyerValueUbOfAt{\selectedBuyerAt{\ts}}{\ts}\leq\sellerValueLbAt{\ts}$,
by \Cref{lem:extension:ttr-round-regret}.

\emph{Trade-search rounds.} Fix a trader and a scale $\padIndex\geq0$.
In every trade-search round in which this trader's uncertainty set
determines the median price at scale $\padIndex$,
\Cref{lem:extension:median-cut} decreases the potential
$\padPotential{\medianPad{\padIndex}}{\cdot}$ of this set by at least
$\log(4/3)$. By \cref{eq:extension:padded-potential}, this potential is
initially at most
$\dimension\log(1+8\dimension\,2^{\padIndex})
=O(\dimension(\log\dimension+\padIndex))$, so there are
$O(\dimension(\log\dimension+\padIndex))$ such rounds. Each of them has
immediate regret less than $2\queryWidth\leq2^{1-\padIndex}$ by
\Cref{lem:extension:ttr-round-regret}. Since
$\sum_{\padIndex\geq0}2^{-\padIndex}(\log\dimension+\padIndex)
=O(\log\dimension)$, these rounds contribute
$O(\dimension\log\dimension)$ for each trader and
$O(\buyerNum\dimension\log\dimension)$ in total. With a single buyer,
refinement rounds have zero regret by
\Cref{lem:extension:ttr-round-regret}, so the regret in bilateral trade
is $O(\dimension\log\dimension)$.

\emph{Refinement rounds with multiple buyers.} By
\Cref{lem:extension:ttr-round-regret}, a refinement round in which the
selected buyer's directional width is at most $1/\timeHorizon$ has
immediate regret at most $1/\timeHorizon$, so these rounds contribute
at most one in total. Fix a buyer and a scale $\quantileScale\geq-1$,
and consider the remaining refinement rounds in which this buyer is
selected and the quantile price uses scale $\quantileScale$. Write
$L_\quantileScale\deq\log(1/\quantileWidthScale{\quantileScale})$, so
that $L_{-1}=0$ and $L_\quantileScale=2^\quantileScale$ for
$\quantileScale\geq0$, and let $\capThickness$ be the thickness of the
lower cap of the padded set below the hyperplane
$\set{\genericVec:\lrangle{\genericVec,\contextAt{\ts}}=\postPriceAt{\ts}}$.
If the buyer accepts, the new padded set misses the interior of the
lower cap of thickness $\capThickness-\padPara{\quantileScale}$, whose
volume fraction is at least
$(1-\padPara{\quantileScale}/\capThickness)^\dimension
\quantileFraction{\quantileScale}
\geq\frac{15}{16}\quantileFraction{\quantileScale}$ by
\Cref{lem:contextual-quantile-thickness} and Bernoulli's inequality. If
the buyer rejects, the new padded set lies in the lower cap of
thickness $\capThickness+\padPara{\quantileScale}$, whose volume
fraction is at most $2\quantileFraction{\quantileScale}$, by the same
lemma. Hence an acceptance decreases the potential
$\padPotential{\padPara{\quantileScale}}{\cdot}$ of the buyer's
uncertainty set by at least $\frac{15}{16}\quantileFraction{\quantileScale}$,
and a rejection decreases it by at least
$\log(1/(2\quantileFraction{\quantileScale}))
=\log4+L_\quantileScale/2$. By \cref{eq:extension:padded-potential},
this potential is initially at most
$\dimension\log(1+1/\padPara{\quantileScale})
=O(\dimension(\log\dimension+L_\quantileScale))$.

First consider the rounds in which the selected buyer accepts. There are
$O(\dimension(\log\dimension+L_\quantileScale)
/\quantileFraction{\quantileScale})$ of them, and each has immediate
regret at most $\queryWidth\leq\quantileWidthScale{\quantileScale}$ by
\Cref{lem:extension:ttr-round-regret}. Since
$\quantileWidthScale{\quantileScale}/\quantileFraction{\quantileScale}
=8\sqrt{\quantileWidthScale{\quantileScale}}$, their total regret is
$O(\dimension(\log\dimension+L_\quantileScale)
\sqrt{\quantileWidthScale{\quantileScale}})$. This is
$O(\dimension\log\dimension)$ for $\quantileScale=-1$. For
$\quantileScale\geq0$,
$\sqrt{\quantileWidthScale{\quantileScale}}=\exp(-2^{\quantileScale-1})$,
and both $\sum_{\quantileScale\geq0}\exp(-2^{\quantileScale-1})$ and
$\sum_{\quantileScale\geq0}2^{\quantileScale}\exp(-2^{\quantileScale-1})$
converge. Hence acceptances contribute $O(\dimension\log\dimension)$ for
each buyer.

Next consider the rounds in which the selected buyer rejects. Each has
immediate regret at most one by \Cref{lem:extension:ttr-round-regret},
and there are
$O(\dimension(\log\dimension+L_\quantileScale)/(1+L_\quantileScale))
=O(\dimension+\dimension\log\dimension/(1+L_\quantileScale))$ of them.
A directional width larger than $1/\timeHorizon$ and at most
$\quantileWidthScale{\quantileScale}=\exp(-2^{\quantileScale})$
requires $2^{\quantileScale}<\log\timeHorizon$, so at most
$1+\log_2\log\timeHorizon$ nonnegative scales occur. Summing over the
scales, each buyer rejects in
$O(\dimension\log\log\timeHorizon+\dimension\log\dimension)$ such
rounds.

Summing over the $\buyerNum$ buyers and adding the trade-search rounds
and the rounds with widths at most $1/\timeHorizon$, the regret with
multiple buyers is
$O(\buyerNum\dimension\log\log\timeHorizon
+\buyerNum\dimension\log\dimension)$.

\emph{Order-optimality.} The bound for bilateral trade does not depend
on the time horizon. For one-to-many markets, the one-seller,
two-buyer instances in
\Cref{thm:mismatch}\ref{item:mismatch:one-to-many} are contextual
markets in dimension one with the constant context $\contextAt{\ts}=1$,
and {\SinglePriceMech} is a special case of {\TwoPriceMech}. Hence every
online learning algorithm using {\SinglePriceMech} incurs
$\Omega(\log\log\timeHorizon)$ regret on some contextual one-to-many
market, and the dependence of our bound on the time horizon is
order-optimal.
\end{proof}

\subsubsection{Contextual Two-Sided Markets}
\label{app:extension:gft:two-sided}
Our algorithm, {\QuantileTradeAndExplore}, extends
{\ConservativeSegmentedPricing} to contextual markets.
It maintains a separate convex uncertainty set for
each trader and uses the current context to obtain
bounds on the trader's cost or value.

The main difference lies in the construction of
individual exploratory prices. The algorithm uses
directional volume quantiles of padded uncertainty
sets, while the selection of trading groups and
the construction of segment prices follow
\Cref{sec:segmented-pricing:algorithm}. Let $k_{\max}\triangleq\min\{\sellerNum,\buyerNum\}$.

\xhdr{Exploratory prices.}
Set $\epsilon\triangleq1/(16\dimension^2\timeHorizon^2)$. For a
trader with directional width $w>0$, let $h$, $\rho_h$ and $z_h$ be the
scale and the parameters of the trader's current uncertainty set $\querySet$ in
\cref{eq:contextual-quantile-scales,eq:contextual-quantile-parameters},
with $u=u_\ts$, and let $C=\querySet+z_h\ball_\dimension$.

For each seller with positive width, let $\xi^{\mathrm{s}}_{\sellerIndex,\ts}$
be the unique directional lower-tail quantile satisfying
$
    \vol{C\cap\set{y:\lrangle{y,\contextAt{\ts}}\leq
    \xi^{\mathrm{s}}_{\sellerIndex,\ts}}}
    =\rho_h\vol{C}$,
where $C=\sellerVecsOfAt{\sellerIndex}{\ts}+z_h\ball_\dimension$.
The quantile is a price in the same units as the realized cost;
its cap thickness below is measured along $u_\ts$.
For each buyer with positive width, let
$\xi^{\mathrm{b}}_{\buyerIndex,\ts}$ be the upper-tail quantile
satisfying
$
    \vol{C\cap\set{y:\lrangle{y,\contextAt{\ts}}\geq
    \xi^{\mathrm{b}}_{\buyerIndex,\ts}}}
    =\rho_h\vol{C}$,
where $C=\buyerVecsOfAt{\buyerIndex}{\ts}+z_h\ball_\dimension$.
Define the \emph{individual exploratory prices} by
\begin{align}
    \widetilde p_{\sellerIndex,\ts}
    &\triangleq\max\{0,\xi^{\mathrm{s}}_{\sellerIndex,\ts}
                    -\lambda_\ts\epsilon\},
    &\widetilde q_{\buyerIndex,\ts}
    &\triangleq\min\{\lambda_\ts,
              \xi^{\mathrm{b}}_{\buyerIndex,\ts}
                    +\lambda_\ts\epsilon\}.
    \label{eq:contextual-quantile-test-prices}
\end{align}
For a trader with zero width, the corresponding price is
$\widetilde p_{\sellerIndex,\ts}
=\max\{0,\sellerValueLbOfAt{\sellerIndex}{\ts}-\lambda_\ts\epsilon\}$
or $\widetilde q_{\buyerIndex,\ts}
=\min\{\lambda_\ts,\buyerValueUbOfAt{\buyerIndex}{\ts}+\lambda_\ts\epsilon\}$. As in the non-contextual algorithm, the seller's exploratory price lies near its lower cost bound, whereas the buyer's exploratory price lies near its upper value bound. 
The volume quantiles adapt these prices to the geometry of the uncertainty sets.
Acceptance at an exploratory price therefore yields a substantial reduction in the corresponding padded volume.

\xhdr{Selecting the trading groups.}
As in {\ConservativeSegmentedPricing}, sort the seller upper bounds
increasingly and the buyer lower bounds decreasingly, breaking ties by
fixed orders of trader identities. Select the largest equally sized
groups guaranteed to accept a common price, namely,
\begin{align}
    \widehat k_\ts\triangleq\max\bigl(
      \{0\}\cup\{\ell\in[k_{\max}]:
      \sellerValueOrdUbOfAt{\ell}{\ts}
      \leq\buyerValueOrdLbOfAt{\ell}{\ts}\}\bigr).
    \label{eq:contextual-quantile-selected-size}
\end{align}
Let $\sellerSubsetAt{1}{\ts}$ and $\buyerSubsetAt{1}{\ts}$ contain the
first $\widehat k_\ts$ traders in the respective orders, with
$\sellerSubsetAt{2}{\ts}$ and $\buyerSubsetAt{2}{\ts}$ denoting
their complements. As in the noncontextual model, we refer to traders
in the first segments as \emph{selected traders} and those in the
second segments as \emph{outside traders}. Set
\begin{align*}
 (\pi^{\mathrm{s}}_\ts,\pi^{\mathrm{b}}_\ts)
 &\triangleq
 \begin{cases}
 (\sellerValueOrdUbOfAt{\widehat k_\ts}{\ts},
  \buyerValueOrdLbOfAt{\widehat k_\ts}{\ts}),
     &\widehat k_\ts>0,\\
 (0,\lambda_\ts),&\widehat k_\ts=0,
 \end{cases}
 \end{align*}

\xhdr{Posting segment prices and updating uncertainty sets.}
Among outside traders, take the smallest individual
seller price and the largest individual buyer price:
$
 \widetilde p_\ts\triangleq
 \min_{\sellerIndex\in\sellerSubsetAt{2}{\ts}}
 \widetilde p_{\sellerIndex,\ts}$ and $
 \widetilde q_\ts\triangleq
 \max_{\buyerIndex\in\buyerSubsetAt{2}{\ts}}
 \widetilde q_{\buyerIndex,\ts}$.
If the outside seller or buyer group is empty, we set
$\widetilde p_\ts=\lambda_\ts$ or $\widetilde q_\ts=0$, respectively. 
Set
\begin{align}
    \pi_\ts
    &\triangleq\min\{\pi^{\mathrm{b}}_\ts,
        \max\{\pi^{\mathrm{s}}_\ts,
        (\widetilde p_\ts+\widetilde q_\ts)/2\}\},
    \label{eq:contextual-quantile-common-price}\\
    \sellerPriceAtOfseg{\ts}{1}
    &=\buyerPriceAtOfseg{\ts}{1}\triangleq\pi_\ts,\quad
    \sellerPriceAtOfseg{\ts}{2}
    \triangleq\min\{\widetilde p_\ts,\pi_\ts\},\quad
    \buyerPriceAtOfseg{\ts}{2}
    \triangleq\max\{\widetilde q_\ts,\pi_\ts\}.
    \label{eq:contextual-quantile-posted-prices}
\end{align}

\begin{algorithm}[tp]
\caption{\QuantileTradeAndExplore}
\label[algorithm]{alg:contextual-gft-quantile}
\SetAlgoLined
\SetNoFillComment
\DontPrintSemicolon
\KwIn{Time horizon $\timeHorizon$; seller and buyer sets $\sellerSet,\buyerSet$}
\KwOut{{\SegPriceMech} in each round $\ts\in[\timeHorizon]$}

Initialize $\sellerVecsOfAt{\sellerIndex}{1}\gets\ball_\dimension$ for all $\sellerIndex\in\sellerSet$ and
$\buyerVecsOfAt{\buyerIndex}{1}\gets\ball_\dimension$ for all $\buyerIndex\in\buyerSet$\;
Set $k_{\max}\gets\min\{\abs{\sellerSet},\abs{\buyerSet}\}$ and
$\epsilon\gets1/(16\dimension^2\timeHorizon^2)$, and fix an order on each side for tie breaking\;

\For{each round $\ts=1$ to $\timeHorizon$}{
    Observe $\contextAt{\ts}$\;
    \eIf{$\contextAt{\ts}=0$}{
        Post price zero to every trader and copy all uncertainty sets to round $\ts+1$\;
    }{
        Set $\contextNormAt{\ts}\gets\norm{\contextAt{\ts}}$ and
        $\contextDirAt{\ts}\gets\contextAt{\ts}/\contextNormAt{\ts}$\;
        \tcp{Compute the individual exploratory prices}
        \ForEach{trader $r\in\agentSet$}{
            Intersect its uncertainty set $\querySet_r$ with
            $\set{\genericVec:\lrangle{\genericVec,\contextDirAt{\ts}}\geq0}$\;
            Compute $w_r\gets\width{\querySet_r}{\contextDirAt{\ts}}$ and its cost or value bounds\;
            \eIf{$w_r=0$}{
                Set $\widetilde p_{r,\ts}\gets\max\{0,\sellerValueLbOfAt{r}{\ts}-\contextNormAt{\ts}\epsilon\}$ for a seller, or
                $\widetilde q_{r,\ts}\gets\min\{\contextNormAt{\ts},\buyerValueUbOfAt{r}{\ts}+\contextNormAt{\ts}\epsilon\}$ for a buyer\;
            }{
                Choose $h,\rho_h,z_h$ by
                \cref{eq:contextual-quantile-scales,eq:contextual-quantile-parameters}\;
                Compute the lower-tail seller quantile or upper-tail buyer quantile of
                $\querySet_r+z_h\ball_\dimension$ and set its exploratory price by
                \cref{eq:contextual-quantile-test-prices}\;
            }
        }
        \tcp{Select the trading groups and determine the segment prices}
        Sort seller upper bounds increasingly and buyer lower bounds decreasingly, using the fixed orders to break ties\;
        Compute $\widehat k_\ts$ by \cref{eq:contextual-quantile-selected-size} and select the first
        $\widehat k_\ts$ sellers and buyers as $\sellerSubsetAt{1}{\ts},\buyerSubsetAt{1}{\ts}$\;
        Set $\sellerSubsetAt{2}{\ts}\gets\sellerSet\setminus\sellerSubsetAt{1}{\ts}$ and
        $\buyerSubsetAt{2}{\ts}\gets\buyerSet\setminus\buyerSubsetAt{1}{\ts}$\;
        Set $\widetilde p_\ts\gets\min_{\sellerIndex\in\sellerSubsetAt{2}{\ts}}\widetilde p_{\sellerIndex,\ts}$ if
        $\sellerSubsetAt{2}{\ts}\neq\emptyset$, and $\widetilde p_\ts\gets\contextNormAt{\ts}$ otherwise\;
        Set $\widetilde q_\ts\gets\max_{\buyerIndex\in\buyerSubsetAt{2}{\ts}}\widetilde q_{\buyerIndex,\ts}$ if
        $\buyerSubsetAt{2}{\ts}\neq\emptyset$, and $\widetilde q_\ts\gets0$ otherwise\;
        Compute $\pi_\ts$ and the segment prices by
        \cref{eq:contextual-quantile-common-price,eq:contextual-quantile-posted-prices}\;
        Implement {\SegPriceMech} with these prices for the selected and outside groups\;
        \tcp{Update the uncertainty sets}
        Observe all traders' accept/reject decisions\;
        \ForEach{seller $\sellerIndex\in\sellerSet$ receiving price $\offeredPrice$}{
            \lIf{seller $\sellerIndex$ accepts}{set $\sellerVecsOfAt{\sellerIndex}{\ts+1}\gets\sellerVecsOfAt{\sellerIndex}{\ts}\cap
                \set{\genericVec:\lrangle{\genericVec,\contextAt{\ts}}\leq\offeredPrice}$}
            \lElse{set $\sellerVecsOfAt{\sellerIndex}{\ts+1}\gets\sellerVecsOfAt{\sellerIndex}{\ts}\cap
                \set{\genericVec:\lrangle{\genericVec,\contextAt{\ts}}\geq\offeredPrice}$}
        }
        \ForEach{buyer $\buyerIndex\in\buyerSet$ receiving price $\offeredPrice$}{
            \lIf{buyer $\buyerIndex$ accepts}{set $\buyerVecsOfAt{\buyerIndex}{\ts+1}\gets\buyerVecsOfAt{\buyerIndex}{\ts}\cap
                \set{\genericVec:\lrangle{\genericVec,\contextAt{\ts}}\geq\offeredPrice}$}
            \lElse{set $\buyerVecsOfAt{\buyerIndex}{\ts+1}\gets\buyerVecsOfAt{\buyerIndex}{\ts}\cap
                \set{\genericVec:\lrangle{\genericVec,\contextAt{\ts}}\leq\offeredPrice}$}
        }
    }
}
\end{algorithm}

\xhdr{Feasibility.}
By construction, all selected traders accept the common
price, and the segment prices satisfy weak budget balance.
As in the baseline model, the segment assignments do not
restrict the matching among accepting traders.

\xhdr{Regret analysis.}
As in the noncontextual setting, we separate the regret
into losses from outside acceptances and those arising
from the gap between the outside seller and buyer prices.
We bound the former by counting acceptances that reduce
the padded-volume potential, and the latter by charging
the price gap to at most two boundary traders.

\xhdr{Counted acceptances.}
For each outside trader $r$ of positive directional width, let
$s_{r,\ts}$ denote its quantile cap thickness measured along
$u_\ts$.
We call an outside acceptance \emph{counted} if
$s_{r,\ts}\geq\epsilon$, and let $A_\ts$ denote the number of
counted acceptances in round $\ts$.
An uncounted accepting outside seller must have cost zero and receive
price zero. Indeed, when $s_{r,\ts}<\epsilon$, its individual
exploratory price is
$\max\{0,\sellerValueLbOfAt{r}{\ts}
 +\lambda_\ts(s_{r,\ts}-z_h-\epsilon)\}$,
which is either zero or strictly below the seller's lower cost bound.
The same conclusion holds for zero width. Symmetrically, an
uncounted accepting outside buyer has value $\lambda_\ts$ and
receives price $\lambda_\ts$. Such uncounted acceptances cannot reduce the minimum GFT
attained by a maximum-cardinality matching of the selected traders.

\begin{lemma}\label{lem:contextual-quantile-counted}
The number of counted outside acceptances satisfies
$
\sum_{\ts\in[\timeHorizon]}A_\ts
    =O\!\left(
    \totalNum\dimension
    (\log\log\timeHorizon+\log\dimension)
    \right)$.
\end{lemma}

\begin{proof}
Fix a trader and a scale $h\geq-1$. Write $L_h\triangleq\log(1/W_h)$,
so that $L_{-1}=0$ and $L_h=2^h$ for $h\geq0$. As noted after
\cref{eq:extension:padded-potential}, the potential
$\padPotential{z_h}{\querySet}$ of the trader's uncertainty set $\querySet$ is
nonincreasing on every round, and
\begin{align}
  0\leq\padPotential{z_h}{\querySet}\leq\dimension\log(1+1/z_h)
  =O\bigl(\dimension(\log\dimension+L_h)\bigr).
  \label{eq:contextual-quantile-potential-bound}
\end{align}
If the trader is a seller and accepts an outside price, that
price does not exceed its individual exploratory price. The new padded
uncertainty set is contained in the old lower cap
of thickness at most $s+z_h$. By
\Cref{lem:contextual-quantile-thickness}, its volume is at most
$2\rho_h$ of the previous padded volume. Buyers are symmetric.
Each counted acceptance at scale $h$ decreases
$\padPotential{z_h}{\cdot}$ by at least
$\log(1/(2\rho_h))=\log4+L_h/2$, so the trader has
$O(\dimension+\dimension\log\dimension/(1+L_h))$ counted acceptances at
this scale.
An acceptance can be counted only if $s\geq\epsilon$.
By \Cref{lem:contextual-quantile-thickness}, $s\leq2w$;
therefore $w\geq\epsilon/2$. At a nonnegative scale,
$W_h\geq w$ and hence
$2^h\leq\log(2/\epsilon)=\log(32\dimension^2\timeHorizon^2)$.
Thus $O(1+\log\log\timeHorizon+\log\log\dimension)$ nonnegative
scales occur. Since $\sum_{h\geq0}1/(1+2^h)\leq2$, each trader has
$O(\dimension(\log\log\timeHorizon+\log\dimension))$ counted acceptances,
and summing over the $\totalNum$ traders proves the lemma.
\end{proof}

\xhdr{Boundary-price errors.}
When $\sellerSubsetAt{2}{\ts}\neq\emptyset$, let
$i_\ts$ be an outside seller attaining $\widetilde p_\ts$;
when $\buyerSubsetAt{2}{\ts}\neq\emptyset$, let $j_\ts$
be an outside buyer attaining $\widetilde q_\ts$. Define
$
   R^{\mathrm{s}}_\ts
   \triangleq\sellerValueUbOfAt{i_\ts}{\ts}
                  -\widetilde p_{i_\ts,\ts}$ and $
   R^{\mathrm{b}}_\ts
   \triangleq\widetilde q_{j_\ts,\ts}
                  -\buyerValueLbOfAt{j_\ts}{\ts}$.
Each radius is nonnegative and at most
$\lambda_\ts(w+\epsilon)$, where $w$ denotes the
corresponding trader's directional width. We charge a boundary radius only if
the trader receives exactly its individual exploratory price. Thus, at most two radii are charged in each round.

\begin{lemma}\label{lem:contextual-quantile-radius}
The sum of charged boundary radii satisfies
$
    \sum_{\ts\in[\timeHorizon]}
    \sum_{r\in\mathcal I_\ts}R_{r,\ts}
    =O(\totalNum\dimension\log\dimension)$,
where $\mathcal I_\ts$ is the set of boundary traders
charged in round $\ts$.
\end{lemma}

\begin{proof}
Fix a trader and a charged exploratory price with positive width $w$ and
cap thickness $s>2\dimension\epsilon$. By
\Cref{lem:contextual-quantile-thickness},
$z_h/s\leq1/(16\dimension)$ and
$\epsilon/s<1/(2\dimension)$.
For a seller, write $\ell$ for its lower cost bound. Its
lower-tail quantile satisfies
$\xi^{\mathrm{s}}=\ell+\lambda_\ts(s-z_h)$;
thus $\xi^{\mathrm{s}}-\lambda_\ts\epsilon>\ell$
and the lower clipping in the exploratory price is inactive.
The same holds symmetrically for a buyer's upper clipping.

If the seller rejects its individual exploratory price, the new
padded uncertainty set misses a lower cap of thickness
$s-\epsilon-z_h$, whose volume fraction is at least
$(1-(\epsilon+z_h)/s)^\dimension\rho_h
\geq(1-9/(16\dimension))^\dimension\rho_h\geq\rho_h/4$ by
\cref{eq:contextual-quantile-cap-stability}. If the seller accepts, the
new padded volume is at most $2\rho_h\leq1/4$ times its old value.
Either way, the potential $\padPotential{z_h}{\cdot}$ decreases by at
least $\rho_h/4$, and the buyer argument is symmetric. By
\cref{eq:contextual-quantile-potential-bound}, the trader therefore
receives $O(\dimension(\log\dimension+L_h)/\rho_h)$ such prices at
scale $h$, each with radius at most $\lambda_\ts(w+\epsilon)\leq2W_h$.
Since $W_h/\rho_h=8\sqrt{W_h}$, their radii sum to
$O(\dimension(\log\dimension+L_h)\sqrt{W_h})$. This is
$O(\dimension\log\dimension)$ for $h=-1$, and for $h\geq0$ it is
$O(\dimension(\log\dimension+2^h)e^{-2^{h-1}})$, which sums to
$O(\dimension\log\dimension)$ over all scales. Summing over the
traders, these radii total $O(\totalNum\dimension\log\dimension)$.

For any remaining charged exploratory price with positive width,
$s\leq2\dimension\epsilon$, so \Cref{lem:contextual-quantile-thickness}
gives $w^{3/2}\leq8\dimension s\leq16\dimension^2\epsilon
=\timeHorizon^{-2}$ and hence $w\leq\timeHorizon^{-1}$; the bound
holds when $w=0$. Each such radius is at most
$\lambda_\ts(1/\timeHorizon+\epsilon)\leq1/\timeHorizon+\epsilon$, and
at most two radii are charged per round, so these radii total at
most $2\timeHorizon(1/\timeHorizon+\epsilon)\leq3$.
\end{proof}

\xhdr{Per-round regret.}
We bound the per-round regret using the matching argument in
\Cref{lem:segmented-interval-round-regret,lem:segmented-interval-price-gap}.
Only the scale of costs and values and the definition of counted
acceptances change.

\begin{lemma}\label{lem:contextual-quantile-market}
For a round $\ts$ with nonzero context, let $A_\ts$ denote
the number of counted acceptances and let
$\matchAt{\ts}$ be any maximum-cardinality matching among all
accepting traders. Then
$
    \OPT_\ts-\gft[\matchAt{\ts}]
      \leq 2\lambda_\ts A_\ts
       +k_{\max}\sum_{r\in\mathcal I_\ts}R_{r,\ts}$,
where $\mathcal I_\ts$ consists of at most two boundary traders,
each receiving its own individual exploratory price.
\end{lemma}

\begin{proof}
The proofs of
\Cref{lem:segmented-interval-round-regret,lem:segmented-interval-price-gap}
in \Cref{app:segmented-pricing-upper-bound} use only the following facts
about a round: the selected groups are given by
Eqn.~\eqref{eq:segmented-selected-size} and accept the common price; the
common and segment prices are obtained from the bounds and the outside
proposals by \cref{eq:segmented-common-price} and the segment-price rule
that follows it, with the common price clipped to the whole price range
when no groups are selected; all costs, values and prices lie in
$[0,1]$; and every uncounted accepting outside seller has cost and price
zero, while every uncounted accepting outside buyer has value and price
one. A round of {\QuantileTradeAndExplore} has the same properties with
$[0,1]$ replaced by $[0,\lambda_\ts]$: its groups and prices are given by
\Cref{eq:contextual-quantile-selected-size,eq:contextual-quantile-common-price,eq:contextual-quantile-posted-prices},
with clipping interval $[0,\lambda_\ts]$ when no groups are selected; all
costs, values and prices lie in $[0,\lambda_\ts]$; and uncounted
acceptances behave as shown above, with $\lambda_\ts$ in place of one.
Dividing all costs, values, bounds and prices by $\lambda_\ts$ therefore
yields a round to which both proofs apply. Multiplying their
conclusions by $\lambda_\ts$ shows that $\OPT_\ts-\gft[\matchAt{\ts}]$
is at most
$2\lambda_\ts A_\ts+k_{\max}(\buyerPriceAtOfseg{\ts}{2}
-\sellerPriceAtOfseg{\ts}{2})$ and that the price gap is at most the sum
of the radii of at most two boundary traders, each receiving its own
individual exploratory price.
\end{proof}

\begin{proof}[Proof of the contextual two-sided-market claim of \Cref{thm:extension:contextual-gft}]
By \Cref{lem:contextual-quantile-market} and
$0<\lambda_\ts\leq1$, summing the per-round regret and applying
\Cref{lem:contextual-quantile-counted,lem:contextual-quantile-radius}
gives
$
    \reg
    \leq 2\sum_{\ts\in[\timeHorizon]}A_\ts
    +k_{\max}\sum_{\ts\in[\timeHorizon]}
    \sum_{r\in\mathcal I_\ts}R_{r,\ts}=O\!\left(
    \totalNum\dimension\log\log\timeHorizon
    +k_{\max}\totalNum\dimension\log\dimension
    \right)$.
This proves the stated regret bound. The argument holds for every
maximum-cardinality matching and therefore also under the
randomized matching rule of the baseline model.
\end{proof}

\subsection{Contextual Profit Maximization}
\label{app:extension:profit}

\xhdr{Algorithm overview.}
Our algorithm, {\MSSPM}, maintains a separate convex uncertainty set
for each seller and buyer. The market contains $\sellerNum$ sellers
and $\buyerNum$ buyers.
The sets
$\sellerVecsOfAt{\sellerIndex}{\ts}$ and
$\buyerVecsOfAt{\buyerIndex}{\ts}$ are initially equal to
$\ball_\dimension$ and are updated using individual accept/reject
decisions. Each set is nested across rounds and contains the
corresponding unknown feature vector. As in the contextual GFT
analysis, the current context determines the lower and upper bounds
$\sellerValueLbOfAt{\sellerIndex}{\ts}$,
$\sellerValueUbOfAt{\sellerIndex}{\ts}$ and
$\buyerValueLbOfAt{\buyerIndex}{\ts}$,
$\buyerValueUbOfAt{\buyerIndex}{\ts}$. We use
$\width{\querySet}{x}$ for the directional width of a convex set $\querySet$
along $x$. Here directional widths are evaluated along the actual
context $x$, rather than its normalized direction, so they are
measured in the same units as prices. For the profit objective, write
$\regProfitAt{\ts}\triangleq\profitOptAt{\ts}-\profitAt{\ts}$
for the immediate regret.

\xhdr{Optimistic profit bound and trading groups.}
At round $\ts$, consider an optimistic fictitious market in which
each seller's cost is its lower bound and each buyer's value is
its upper bound. For each
$k\in[\min\{\sellerNum,\buyerNum\}]$, let
$\sellerValueOrdLbOfAt{k}{\ts}$ be the $k$-th smallest seller
lower bound, and let $\buyerValueOrdUbOfAt{k}{\ts}$ be the
$k$-th largest buyer upper bound. Define
\begin{align}
    \optProfitUpperBoundAt{\ts}
    &\triangleq\max\left\{0,
        \max_{k\in[\min\{\sellerNum,\buyerNum\}]}
        k\left(\buyerValueOrdUbOfAt{k}{\ts}
               -\sellerValueOrdLbOfAt{k}{\ts}\right)\right\}.
    \label{eq:extension:profit-optimistic-bound}
\end{align}
If this maximum is zero, the algorithm posts a common price of zero.
Otherwise, let $\tradeNum^\star$ be a maximizing positive trade size.
In particular,
$1\leq\tradeNum^\star\leq\min\{\sellerNum,\buyerNum\}$.
The sets
$\ficTradingSellersAt{\ts}$ and $\ficTradingBuyersAt{\ts}$
consist of the $\tradeNum^\star$ sellers with smallest lower
cost bounds and the $\tradeNum^\star$ buyers with largest upper
value bounds, respectively, with ties broken in a fixed order.

\begin{lemma}\label{lem:extension:profit-optimistic}
For every round $\ts$,
$\profitOptAt{\ts}\leq\optProfitUpperBoundAt{\ts}$.
Moreover, if $\tradeNum^\star>0$, then for any
$\sellerIndex\in\ficTradingSellersAt{\ts}$ and
$\buyerIndex\in\ficTradingBuyersAt{\ts}$,
$
    \profitOptAt{\ts}
    \leq\tradeNum^\star\left(
        \buyerValueUbOfAt{\buyerIndex}{\ts}
        -\sellerValueLbOfAt{\sellerIndex}{\ts}\right)$.
\end{lemma}
\begin{proof}
By the definition of the optimum profit benchmark in
\Cref{sec:extension:profit},
we have
$\profitOptAt{\ts}
=\max\{0,\max_{k\in[\min\{\sellerNum,\buyerNum\}]}
k(\buyerValueOrdAt{k}{\ts}-\sellerValueOrdAt{k}{\ts})\}$.
Indeed, prices implementing $k$ trades cannot have a spread larger
than $\buyerValueOrdAt{k}{\ts}-\sellerValueOrdAt{k}{\ts}$;
conversely, whenever this spread is nonnegative, posting the
seller-side price $\sellerValueOrdAt{k}{\ts}$ and buyer-side price
$\buyerValueOrdAt{k}{\ts}$ implements at least $k$ trades and
attains at least this profit.
The expression is nondecreasing in every buyer value and
nonincreasing in every seller cost. Replacing the actual values
by their upper bounds and costs by their lower bounds therefore
gives the first claim. For the second, every seller in
$\ficTradingSellersAt{\ts}$ has a lower bound at most
$\sellerValueOrdLbOfAt{\tradeNum^\star}{\ts}$, and every buyer
in $\ficTradingBuyersAt{\ts}$ has an upper bound at least
$\buyerValueOrdUbOfAt{\tradeNum^\star}{\ts}$.
\end{proof}

\xhdr{Exploratory prices.}
For each trader with directional width
$w\triangleq\width{\querySet}{\contextAt{\ts}}>1/\timeHorizon$,
where $\querySet$ is its current uncertainty set, choose a scale $h$ as follows:
\begin{align}
    W_{-1}&\triangleq2,&
    \rho_{-1}&\triangleq\frac14,&
    z_{-1}&\triangleq\frac{1}{256\dimension^2},\nonumber\\
    W_h&\triangleq2^{-2^h},&
    \rho_h&\triangleq\sqrt{W_h},&
    z_h&\triangleq\frac{W_h^3}{128\dimension^2}
    \quad(h\geq0).
    \label{eq:extension:profit-quantile-scales}
\end{align}
Use $h=-1$ when $w>1/2$; otherwise choose the unique
$h\geq0$ such that $W_h^2<w\leq W_h$. These cases cover
all $w\in(0,2]$. Each trader chooses its scale separately.
For a seller $\sellerIndex$, let
$\querySet=\sellerVecsOfAt{\sellerIndex}{\ts}$ and
$x=\contextAt{\ts}$. Define the temporary price
$\tempPriceSellerOfAt{\sellerIndex}{\ts}$ by
$
 \vol{\set{\theta\in \querySet+z_h\ball_\dimension:
     \lrangle{\theta,x}\geq
     \tempPriceSellerOfAt{\sellerIndex}{\ts}}}
    =\rho_h\vol{\querySet+z_h\ball_\dimension}$
and set the individual exploratory price to
$\sellerPriceAtOf{\ts}{\sellerIndex}
\triangleq\tempPriceSellerOfAt{\sellerIndex}{\ts}
+z_h\norm{x}$. For a buyer $\buyerIndex$, let
$\querySet=\buyerVecsOfAt{\buyerIndex}{\ts}$ and use the
lower-tail quantile $\tempPriceBuyerOfAt{\buyerIndex}{\ts}$
with the same volume fraction, setting
$\buyerPriceAtOf{\ts}{\buyerIndex}
\triangleq\tempPriceBuyerOfAt{\buyerIndex}{\ts}
-z_h\norm{x}$.
When the width is at most $1/\timeHorizon$, set
$\sellerPriceAtOf{\ts}{\sellerIndex}
=\sellerValueUbOfAt{\sellerIndex}{\ts}$ or
$\buyerPriceAtOf{\ts}{\buyerIndex}
=\buyerValueLbOfAt{\buyerIndex}{\ts}$, respectively.

Among the selected traders, let $\sellerSearchIndex{\ts}$
attain the largest individual exploratory seller price and let
$\buyerSearchIndex{\ts}$ attain the smallest individual
exploratory buyer price, with ties broken in a fixed order.
We refer to these two traders as the price-setting seller and buyer.
If these two prices satisfy weak budget balance, the algorithm
posts them to all sellers and buyers, respectively. Otherwise,
it posts a common price, chosen as a padded-volume median for
whichever price-setting trader has larger directional width.
For this median query, choose the largest integer $g\geq-1$ such that
$w\leq2^{-g}$ and set $\bar z_g\triangleq2^{-g}/(8\dimension)$.
The median branch always has a price-setting trader of positive width:
if both widths were zero, positive optimistic profit would imply
that the two exploratory prices were already ordered.

\begin{algorithm}[tp]
\caption{\MSSPM}
\label[algorithm]{alg:-contextual-two-sided-market-profit}
\SetAlgoLined
\SetNoFillComment
\DontPrintSemicolon
\KwIn{Time horizon $\timeHorizon$; seller and buyer sets $\sellerSet,\buyerSet$}
\KwOut{{\TwoPriceMech} in each round $\ts\in[\timeHorizon]$}

Initialize $\sellerVecsOfAt{\sellerIndex}{1}\gets\ball_\dimension$ for all $\sellerIndex\in\sellerSet$ and
$\buyerVecsOfAt{\buyerIndex}{1}\gets\ball_\dimension$ for all $\buyerIndex\in\buyerSet$\;
Fix an order of trader identities on each side for tie breaking\;

\For{each round $\ts=1$ to $\timeHorizon$}{
    Observe $\contextAt{\ts}$ and compute all directional bounds\;
    Compute $\optProfitUpperBoundAt{\ts}$ from
    Eqn.~\eqref{eq:extension:profit-optimistic-bound}\;
    \eIf{$\optProfitUpperBoundAt{\ts}=0$}{
        Set $\sellerPriceAt{\ts}\gets0$ and $\buyerPriceAt{\ts}\gets0$\;
    }{
        Select a maximizing trade size $\tradeNum^\star$ in Eqn.~\eqref{eq:extension:profit-optimistic-bound}\;
        Sort seller lower bounds increasingly and buyer upper bounds decreasingly\;
        Let $\ficTradingSellersAt{\ts},\ficTradingBuyersAt{\ts}$ contain the first
        $\tradeNum^\star$ sellers and buyers in their respective sorted lists\;
        \ForEach{selected trader $r$ with uncertainty set $\querySet_r$}{
            Compute $w_r\gets\width{\querySet_r}{\contextAt{\ts}}$\;
            \eIf{$w_r\leq1/\timeHorizon$}{
                Set $\sellerPriceAtOf{\ts}{r}\gets\sellerValueUbOfAt{r}{\ts}$ for a seller, or
                $\buyerPriceAtOf{\ts}{r}\gets\buyerValueLbOfAt{r}{\ts}$ for a buyer\;
            }{
                Choose $h,\rho_h,z_h$ by
                Eqn.~\eqref{eq:extension:profit-quantile-scales}\;
                Compute $\tempPriceSellerOfAt{r}{\ts}$ or $\tempPriceBuyerOfAt{r}{\ts}$ as the upper- or lower-tail
                $\rho_h$-quantile of $\querySet_r+z_h\ball_\dimension$, respectively\;
                Set $\sellerPriceAtOf{\ts}{r}\gets\tempPriceSellerOfAt{r}{\ts}+z_h\norm{\contextAt{\ts}}$ for a seller, or
                $\buyerPriceAtOf{\ts}{r}\gets\tempPriceBuyerOfAt{r}{\ts}-z_h\norm{\contextAt{\ts}}$ for a buyer\;
            }
        }
        Select $\sellerSearchIndex{\ts}\in\argmax_{\sellerIndex\in\ficTradingSellersAt{\ts}}\sellerPriceAtOf{\ts}{\sellerIndex}$ and
        $\buyerSearchIndex{\ts}\in\argmin_{\buyerIndex\in\ficTradingBuyersAt{\ts}}\buyerPriceAtOf{\ts}{\buyerIndex}$, using the fixed orders to break ties\;
        \eIf{$\sellerPriceAtOf{\ts}{\sellerSearchIndex{\ts}}\leq\buyerPriceAtOf{\ts}{\buyerSearchIndex{\ts}}$}{
            Set $\sellerPriceAt{\ts}\gets\sellerPriceAtOf{\ts}{\sellerSearchIndex{\ts}}$ and
            $\buyerPriceAt{\ts}\gets\buyerPriceAtOf{\ts}{\buyerSearchIndex{\ts}}$\;
        }{
            Select the wider set $\querySet$ from
            $\sellerVecsOfAt{\sellerSearchIndex{\ts}}{\ts},\buyerVecsOfAt{\buyerSearchIndex{\ts}}{\ts}$ along $\contextAt{\ts}$\;
            Set $g\gets\max\{\ell\in\mathbb Z:\ell\geq-1,\ \width{\querySet}{\contextAt{\ts}}\leq2^{-\ell}\}$ and
            $\bar z_g\gets2^{-g}/(8\dimension)$\;
            Set $\sellerPriceAt{\ts}=\buyerPriceAt{\ts}$ to the padded-volume median of
            $\querySet+\bar z_g\ball_\dimension$ along $\contextAt{\ts}$\;
        }
    }
    Implement {\TwoPriceMech} with seller price $\sellerPriceAt{\ts}$ and buyer price $\buyerPriceAt{\ts}$\;
    \tcp{Update the uncertainty sets}
    Observe all traders' accept/reject decisions\;
    \ForEach{seller $\sellerIndex\in\sellerSet$}{
        \lIf{seller $\sellerIndex$ accepts}{set $\sellerVecsOfAt{\sellerIndex}{\ts+1}\gets\sellerVecsOfAt{\sellerIndex}{\ts}\cap
            \set{\genericVec:\lrangle{\genericVec,\contextAt{\ts}}\leq\sellerPriceAt{\ts}}$}
        \lElse{set $\sellerVecsOfAt{\sellerIndex}{\ts+1}\gets\sellerVecsOfAt{\sellerIndex}{\ts}\cap
            \set{\genericVec:\lrangle{\genericVec,\contextAt{\ts}}\geq\sellerPriceAt{\ts}}$}
    }
    \ForEach{buyer $\buyerIndex\in\buyerSet$}{
        \lIf{buyer $\buyerIndex$ accepts}{set $\buyerVecsOfAt{\buyerIndex}{\ts+1}\gets\buyerVecsOfAt{\buyerIndex}{\ts}\cap
            \set{\genericVec:\lrangle{\genericVec,\contextAt{\ts}}\geq\buyerPriceAt{\ts}}$}
        \lElse{set $\buyerVecsOfAt{\buyerIndex}{\ts+1}\gets\buyerVecsOfAt{\buyerIndex}{\ts}\cap
            \set{\genericVec:\lrangle{\genericVec,\contextAt{\ts}}\leq\buyerPriceAt{\ts}}$}
    }
}
\end{algorithm}

\xhdr{Geometric estimates.}
We use the standard cap-volume inequalities established
in~\Cref{app:extension:gft}, particularly the bounds on cap-volume
fractions and the comparison inequality for different
cap thicknesses.

\begin{lemma}\label{lem:extension:profit-quantile-geometry}
For a positive-width exploratory query at scale $h$, let $s$
be the corresponding cap thickness in $\querySet+z_h\ball_\dimension$, measured in the
$x$-projection coordinate.
Then $z_h\norm{x}/s\leq1/(32\dimension)$, and the individual
exploratory price lies strictly between the directional lower
and upper bounds of $\querySet$. If the trader rejects a posted price
that is at least its individual seller price or at most its
individual buyer price, respectively, its padded volume shrinks
by a factor at most $\rho_h$. If the trader receives exactly
its individual exploratory price, either feedback outcome
decreases the logarithmic padded-volume potential by at least
$\rho_h/4$.
\end{lemma}
\begin{proof}
Put $\lambda=\norm{x}>0$ and $w=\width{\querySet}{x}>0$, and write
$\ell=\min_{\theta\in \querySet}\lrangle{\theta,x}$ and
$u=\max_{\theta\in \querySet}\lrangle{\theta,x}$.
The directional width of the padded body along $x/\lambda$ is
$(w+2z_h\lambda)/\lambda$.
By \eqref{eq:contextual-quantile-cap-bounds},
$s\geq\rho_h(w+2z_h\lambda)/\dimension$, where $s$ is measured
in the $x$-projection coordinate. If $h\geq0$, the scale choice
gives $w>W_h^2$, so
$
 \frac{z_h\lambda}{s}
 \leq\frac{\dimension z_h\lambda}{\rho_hw}
 <\frac{\sqrt{W_h}}{128\dimension}
 \leq\frac{1}{128\dimension}$.
For $h=-1$, $w>1/2$ and the ratio is smaller than
$1/(32\dimension)$. Thus $s>2z_h\lambda$.
For a seller the upper-tail quantile is
$y=u+z_h\lambda-s$, so its individual exploratory
price is $u+2z_h\lambda-s$, which belongs to $(\ell,u)$ since
$2z_h\lambda<s<w+2z_h\lambda$.
The buyer argument is symmetric.

If a seller rejects a posted price at least its individual
price, the next uncertainty set is contained in
$\querySet\cap\set{\theta:\lrangle{\theta,x}\geq y+z_h\lambda}$.
After padding, it lies in the original upper cap above $y$,
whose volume fraction is $\rho_h$. A buyer rejection is
symmetric. If the seller instead accepts its exact exploratory
price, the padded next set lies below $y+2z_h\lambda$.
The upper cap removed therefore has thickness
$s-2z_h\lambda$ and, by \eqref{eq:contextual-quantile-cap-comparison}, volume fraction at least
$
  \left(1-\frac{2z_h\lambda}{s}\right)^\dimension\rho_h
   \geq\frac{15}{16}\rho_h$.
The same conclusion holds for buyer acceptance. Consequently,
acceptance reduces the logarithmic potential by at least
$-\log(1-15\rho_h/16)\geq\rho_h/4$; rejection reduces it
by at least $\log(1/\rho_h)\geq\rho_h/4$, using
$\rho_h\leq1/\sqrt2$.
\end{proof}

\begin{lemma}\label{lem:extension:profit-median-geometry}
Let $\querySet\subseteq\ball_\dimension$ be nonempty and convex,
$x\in\ball_\dimension$, and $\width{\querySet}{x}>0$. Choose the largest
integer $g\geq-1$ satisfying $\width{\querySet}{x}\leq2^{-g}$, set
$\bar z_g=2^{-g}/(8\dimension)$, and let $y$ bisect the
volume of $\querySet+\bar z_g\ball_\dimension$ in direction $x$.
For either halfspace $H$ bounded by $\lrangle{\theta,x}=y$,
$
 \vol{(\querySet\cap H)+\bar z_g\ball_\dimension}
 \leq\frac34\vol{\querySet+\bar z_g\ball_\dimension}$.
\end{lemma}
\begin{proof}
The median splits the padded volume equally. Padding either
retained half of $\querySet$ enlarges the corresponding half of the
padded set by at most a slab of thickness $\bar z_g$ in the
unit direction $x/\norm{x}$. To bound its volume, let $A_{\max}$
be the largest cross-sectional area of the padded body.
Convexity and the cones joining such a section to the two
supporting hyperplanes imply
$A_{\max}\leq\dimension\vol{\querySet+\bar z_g\ball_\dimension}/L$,
where $L=\width{\querySet}{x}/\norm{x}+2\bar z_g$. Thus the slab occupies at most
$\dimension\bar z_g\norm{x}/\width{\querySet}{x}$ times the original
padded volume. By maximality of $g$,
$\width{\querySet}{x}>2^{-g-1}$, so this fraction is smaller than $1/4$.
Adding the median half proves the claim.
\end{proof}

\xhdr{Regret analysis.}
Call a round \emph{Type-$1$} if the two exploratory prices
satisfy weak budget balance and all selected traders accept;
\emph{Type-$2$} if they satisfy weak budget balance but at least
one selected trader rejects; and \emph{Type-$3$} if the seller's
exploratory price exceeds the buyer's. Rounds with zero
optimistic profit have zero regret and are omitted from this
classification. Let $r_\ts$ be the number of rejecting
selected traders in a Type-$2$ round, and set
$
  w^\mathrm{s}_\ts
   \triangleq
   \width{\sellerVecsOfAt{\sellerSearchIndex{\ts}}{\ts}}
          {\contextAt{\ts}}$ and
  $w^\mathrm{b}_\ts
   \triangleq
   \width{\buyerVecsOfAt{\buyerSearchIndex{\ts}}{\ts}}
          {\contextAt{\ts}}$.

\begin{lemma}\label{lem:extension:profit-pointwise}
In every Type-$1$ or Type-$2$ round,
$
  \regProfitAt{\ts}
   \leq2\tradeNum^\star(w^\mathrm{s}_\ts+w^\mathrm{b}_\ts)
       +2r_\ts$,
where $r_\ts=0$ in Type-$1$ rounds. In Type-$3$ rounds,
$
  \regProfitAt{\ts}
   \leq4\tradeNum^\star\max\{w^\mathrm{s}_\ts,w^\mathrm{b}_\ts\}$.
\end{lemma}
\begin{proof}
Write $k=\tradeNum^\star$ and let $a$ and $b$ be the
individual exploratory prices of the price-setting seller and
buyer, respectively. By
\Cref{lem:extension:profit-quantile-geometry}, or by the
endpoint rule for widths at most $1/\timeHorizon$, we have
$\sellerValueLbOfAt{\sellerSearchIndex{\ts}}{\ts}
    \leq a\leq
      \sellerValueUbOfAt{\sellerSearchIndex{\ts}}{\ts}$ and $
  \buyerValueLbOfAt{\buyerSearchIndex{\ts}}{\ts}
    \leq b\leq
      \buyerValueUbOfAt{\buyerSearchIndex{\ts}}{\ts}$.
If $a\leq b$, the number of trades is at least $k-r_\ts$.
Moreover, $0\leq b-a\leq2$ because all realized types and
support endpoints lie in $[-1,1]$. Hence
$
 \profitAt{\ts}\geq(k-r_\ts)(b-a)
                     \geq k(b-a)-2r_\ts$.
Combining this with
\cref{lem:extension:profit-optimistic},
$
 \regProfitAt{\ts}
 \leq k\left(
   \buyerValueUbOfAt{\buyerSearchIndex{\ts}}{\ts}-b
   +a-\sellerValueLbOfAt{\sellerSearchIndex{\ts}}{\ts}
   \right)+2r_\ts \leq k(w^\mathrm{s}_\ts+w^\mathrm{b}_\ts)+2r_\ts$,
which proves the Type-$1$/Type-$2$ bound.
If $a>b$, the posted common price yields zero profit.
Since $a>b$ and the individual prices lie within their
corresponding bounds, $\buyerValueLbOfAt{\buyerSearchIndex{\ts}}{\ts}
 -\sellerValueUbOfAt{\sellerSearchIndex{\ts}}{\ts}
 \leq b-a<0$.
Therefore, by \cref{lem:extension:profit-optimistic},
$
 \regProfitAt{\ts}=\profitOptAt{\ts}
 \leq k\left(
   \buyerValueUbOfAt{\buyerSearchIndex{\ts}}{\ts}
   -\sellerValueLbOfAt{\sellerSearchIndex{\ts}}{\ts}
   \right)\leq k(w^\mathrm{s}_\ts+w^\mathrm{b}_\ts)
 \leq2k\max\{w^\mathrm{s}_\ts,w^\mathrm{b}_\ts\}$,
which proves the Type-$3$ bound.
\end{proof}

\begin{lemma}\label{lem:extension:profit-rejection-count}
The total number of selected-trader rejections in Type-$2$
rounds satisfies
$
 \sum_{\ts\text{ of Type-2}}r_\ts
 =O\!\left((\sellerNum+\buyerNum)\dimension
          (\log\log\timeHorizon+\log\dimension)\right)$.
\end{lemma}
\begin{proof}
Fix a trader and scale $h$, and use the nonincreasing
potential
$
 \padPotential{z_h}{\querySet}\triangleq
   \log\frac{\vol{\querySet+z_h\ball_\dimension}}
             {\vol{z_h\ball_\dimension}}$.
Since every uncertainty set is nested in $\ball_\dimension$,
$0\leq\padPotential{z_h}{\querySet}\leq\dimension\log(1+1/z_h)$.
A selected trader with width at most $1/\timeHorizon$
cannot reject: its own exploratory offer guarantees acceptance,
and the posted price is at least that offer for sellers and at most
that offer for buyers. Every rejection of a selected trader in a
Type-$2$ round at this scale therefore involves positive width
exceeding $1/\timeHorizon$ and reduces its padded volume
by a factor at most
$\rho_h$, by \Cref{lem:extension:profit-quantile-geometry}.
For $h\geq0$, the resulting decrease in $\padPotential{z_h}{\querySet}$ is at least
$\log(1/\rho_h)=\frac12\,2^h\log2$, whereas
$
 \padPotential{z_h}{\querySet}\leq
 \dimension\bigl(\log(1+128\dimension^2)
                     +3\cdot2^h\log2\bigr)$.
Thus the number of rejections by one trader at scale $h$
is at most
$O(\dimension+\dimension\log\dimension/2^h)$.
The scale $h=-1$ contributes $O(\dimension\log\dimension)$
rejections. A rejection can occur only when the corresponding
width is greater than $1/\timeHorizon$, which restricts
$h\geq0$ to $O(\log\log\timeHorizon)$ scales. Summing over
all scales gives
$O(\dimension(\log\log\timeHorizon+\log\dimension))$
rejections per trader. There are $\sellerNum+\buyerNum$ traders,
which proves the lemma.
\end{proof}
\begin{lemma}\label{lem:extension:profit-boundary-widths}
Over the course of all Type-$1$ and Type-$2$ rounds,
$
 \sum_{\ts\text{ of Type-1 or Type-2}}
       \tradeNum^\star(w^\mathrm{s}_\ts+w^\mathrm{b}_\ts)
 =O\!\left((\sellerNum+\buyerNum)
        \min\{\sellerNum,\buyerNum\}\dimension
        \log\dimension\right)$.
\end{lemma}
\begin{proof}
Each price-setting trader receives exactly its individual
exploratory price. Fix such a trader and a scale $h\geq-1$
with directional width greater than $1/\timeHorizon$.
By \Cref{lem:extension:profit-quantile-geometry}, either
feedback outcome decreases $\padPotential{z_h}{\querySet}$ by at least $\rho_h/4$.
Consequently, the number of price-setting queries involving this
trader and scale is at most
$4\dimension\log(1+1/z_h)/\rho_h$.
Each has width at most $W_h$. For $h\geq0$,
$W_h/\rho_h=\sqrt{W_h}=2^{-2^{h-1}}$, so the sum of charged
widths is bounded by
$
 O\!\left(\dimension\sum_{h\geq0}
   (\log\dimension+2^h)2^{-2^{h-1}}\right)
 =O(\dimension\log\dimension)$.
The coarse scale $h=-1$ contributes the same order.
For price-setting queries with width at most $1/\timeHorizon$,
at most two widths are charged in each round, so their
sum is at most $2$ over the entire horizon. Summing all charged
widths over the $\sellerNum+\buyerNum$ traders yields
$
 \sum_{\ts\text{ of Type-1 or Type-2}}
 (w^\mathrm{s}_\ts+w^\mathrm{b}_\ts)
 =O\!\left((\sellerNum+\buyerNum)\dimension
             \log\dimension\right)$.
Finally, $\tradeNum^\star\leq\min\{\sellerNum,\buyerNum\}$
in every round, proving the stated bound.
\end{proof}

\begin{lemma}\label{lem:extension:profit-median-cost}
The cumulative regret from Type-$3$ rounds is
$
 O\!\left((\sellerNum+\buyerNum)
        \min\{\sellerNum,\buyerNum\}\dimension
        \log\dimension\right)$.
\end{lemma}
\begin{proof}
In a Type-$3$ round, let $r$ be the price-setting trader
with larger directional width, and let $g\geq-1$ be the
median scale used by the algorithm. The chosen common price
is a median of $\querySet_r+\bar z_g\ball_\dimension$. For either
accept/reject decision of trader $r$, its updated set is
contained in one of the two median halfspaces. Hence
\Cref{lem:extension:profit-median-geometry} reduces the
corresponding padded volume by a factor at most $3/4$.
For a fixed trader and scale $g$, the potential
$
 \padPotential{\bar z_g}{\querySet}\triangleq
 \log\frac{\vol{\querySet+\bar z_g\ball_\dimension}}
           {\vol{\bar z_g\ball_\dimension}}$
is nonincreasing and initially at most
$\dimension\log(1+1/\bar z_g)$.
Thus there are at most
$O(\dimension\log(1+8\dimension\,2^g))$
such rounds at this scale. Since the chosen trader's width
is at most $2^{-g}$, the Type-$3$ bound in \Cref{lem:extension:profit-pointwise}
and $\tradeNum^\star\leq\min\{\sellerNum,\buyerNum\}$
bound the cumulative regret charged to one trader by
$
 O\!\left(\min\{\sellerNum,\buyerNum\}\dimension
  \sum_{g\geq-1}2^{-g}\log(1+8\dimension\,2^g)\right)
 =O\!\left(\min\{\sellerNum,\buyerNum\}\dimension
          \log\dimension\right)$.
Summing over the $\sellerNum+\buyerNum$ traders proves the lemma.
\end{proof}

\begin{proof}[Proof of \Cref{thm:extension:contextual-profit}]
If either side is empty, the optimal profit and regret are zero.
We henceforth assume $\sellerNum,\buyerNum\geq1$.
By \Cref{lem:extension:profit-pointwise}, the regret from
Type-$1$ and Type-$2$ rounds is at most $
 2\sum_{\ts\text{ of Type-1 or Type-2}}
       \tradeNum^\star(w^\mathrm{s}_\ts+w^\mathrm{b}_\ts)
       +2\sum_{\ts\text{ of Type-2}}r_\ts$.
Applying
\Cref{lem:extension:profit-boundary-widths,lem:extension:profit-rejection-count}
and adding the Type-$3$ bound from
\Cref{lem:extension:profit-median-cost} yields
$
 \sum_{\ts\in[\timeHorizon]}\regProfitAt{\ts}
 =O\!\left((\sellerNum+\buyerNum)\dimension
          \log\log\timeHorizon
          +(\sellerNum+\buyerNum)
          \min\{\sellerNum,\buyerNum\}\dimension
          \log\dimension\right)$.
Here the additional term
$O((\sellerNum+\buyerNum)\dimension\log\dimension)$
from the rejection count is absorbed because
$\min\{\sellerNum,\buyerNum\}\geq1$.
The asymptotic estimate is stated for sufficiently large
$\timeHorizon$; bounded horizons are covered by adjusting
absolute constants.

For the horizon lower bound, fix any $\sellerNum,\buyerNum\geq1$
and consider a constant one-dimensional context. Take one seller
with cost zero and one buyer with an unknown fixed value in
$[1/2,3/4]$. Give every additional seller the fixed cost $1$
and every additional buyer the fixed value $0$. Under weak
budget balance, a positive-profit trade can involve only the
distinguished seller and buyer. The other traders have fixed types
that may even be revealed to the learner, and their feedback provides
no information about the unknown value.
Hence the $\Omega(\log\log\timeHorizon)$ lower bound for
one-sided dynamic pricing with binary feedback
\citep{KL-03} applies. For any $\dimension\geq1$, this instance is realized by taking
$\contextAt{\ts}=e_1$ in every round and choosing all feature
vectors in the span of $e_1$. This establishes the claimed
order-optimal dependence on $\timeHorizon$.
\end{proof}

\section{Numerical Experiments}
\label{sec:numerical-experiments}
\begingroup
\long\def\mfedit#1{#1}\label{sec:experiment-design}
This appendix evaluates {\ConservativeSegmentedPricing} in randomly drawn many-to-many markets with 100 or 500 traders on each side. We compare {\ConservativeSegmentedPricing} with {\MSSPM}, the profit-maximizing algorithm of \Cref{sec:extension:profit} (see also \Cref{alg:-contextual-two-sided-market-profit}), and with two algorithms from the bilateral-trade literature. We describe the experimental setup in \Cref{sec:experiment:setup} and report the results in \Cref{sec:experiment-random}.

\subsection{Experiment Setup}
\label{sec:experiment:setup}
Our experiment setup is as follows.

\xhdr{Instances.}
To draw an instance with a given number of traders on each side, we first draw a seller center and a buyer center independently and uniformly from $[0.35,0.65]$. Seller costs and buyer values are then drawn independently from normal distributions with these centers as means and a common standard deviation (SD), truncated to $[0,1]$. Types are drawn once and stay fixed for all $\timeHorizon=1000$ rounds. We consider four instance configurations, which cross two market sizes, 100 or 500 traders on each side, with two SDs, 0.2 or 0.5. Each instance configuration contains 500 independently drawn instances.

\xhdr{Feedback and matching.}
In each round, every algorithm observes only the traders' individual accept/reject decisions. A seller accepts any price at least equal to the seller's cost, and a buyer accepts any price at most equal to the buyer's value. A matching is then drawn uniformly from the maximum-cardinality matchings between the accepting sellers and buyers, following the randomized matching rule of \Cref{sec:prelim}.

\xhdr{Algorithms.}
We compare four algorithms, all of which satisfy weak budget balance. The first two are ours. {\ConservativeSegmentedPricing} (\Cref{alg:conservative-segmented-pricing}) maximizes GFT using {\SegPriceMech}. {\MSSPM} (\Cref{alg:-contextual-two-sided-market-profit}) maximizes platform profit using {\TwoPriceMech}; we run {\MSSPM} with one-dimensional features and a constant context equal to one, so that the types are fixed, and keep the algorithm's uncertainty sets and feedback unchanged.

The other two are \emph{Scouting Bandits} and \emph{Scouting Blindits} of \citet[Algorithms~3 and~4 of the extended version]{CCCFL-24a}. Both were designed for bilateral trade in which the seller's cost and the buyer's value are drawn independently in each round from distributions with bounded densities. Scouting Bandits observes each trader's accept/reject decision, whereas Scouting Blindits observes only whether trade occurs. Both algorithms restrict attention to a fixed grid of candidate prices and run in two phases: a scouting phase that collects the statistics needed to estimate the expected GFT at each candidate price, followed by a phase that posts candidate prices as common prices.

Scouting Bandits uses the ten candidate prices $k/11$ for $k=1,\dots,10$. In each of the first 100 rounds, Scouting Bandits posts a common price drawn uniformly from $[0,1]$. In every later round, Scouting Bandits posts a candidate price selected by UCB1, with the estimated GFT at the posted price as the reward; UCB1 is our choice for the bandit subroutine that the original algorithm takes as an input.

Scouting Blindits uses the six candidate prices $k/7$ for $k=1,\dots,6$ and spends 32 scouting rounds on each candidate price $p$. In 16 of these rounds, Scouting Blindits posts the seller-side price $p$ and a buyer-side price drawn uniformly from $[p,1]$; in the other 16, Scouting Blindits posts the buyer-side price $p$ and a seller-side price drawn uniformly from $[0,p]$. After these 192 rounds, Scouting Blindits commits to the candidate price with the largest estimated GFT and posts that price as a common price in every remaining round. With the number of scouting rounds prescribed by \citet{CCCFL-24a}, which includes a logarithmic factor, scouting alone would exceed the horizon, so we drop that factor.

To apply both algorithms to markets with many traders, we pool feedback. All sellers face the same seller-side price, and all buyers face the same buyer-side price. The seller's decision in the original algorithms is replaced by the fraction of sellers who accept, the buyer's decision by the fraction of buyers who accept, and the trade indicator of Scouting Blindits by the product of these two fractions. With pooled feedback, both algorithms target the expected GFT at a common price of a seller and a buyer drawn independently and uniformly at random from the market, rather than the GFT of the market under the randomized matching rule.

\xhdr{Outcome measures.}
First-best GFT is the optimum GFT benchmark $\gftOpt$ of \Cref{sec:prelim}, the largest GFT over all feasible matchings. Because types are fixed, cumulative first-best GFT equals $\timeHorizon\cdot\gftOpt$. We compute the expected GFT of each algorithm under the randomized matching rule exactly rather than by simulating matchings. Given the sets of accepting traders, expected GFT equals the number of trades multiplied by the difference between the average value of the accepting buyers and the average cost of the accepting sellers.

To see who receives the gains, we split expected GFT into three parts. Seller utility is the matched sellers' payments minus their costs, buyer utility is the matched buyers' values minus their payments, and platform profit is the buyers' payments minus the sellers' payments. \emph{Forgone GFT} is cumulative first-best GFT minus the algorithm's cumulative expected GFT. Measured as shares of cumulative first-best GFT, the four quantities sum to 100\%.

\subsection{Experiment Results}
\label{sec:experiment-random}
\label{sec:experiment-surplus}

\begin{table}
    \centering
    \caption{Total GFT by market size and type dispersion.}
    \label{tab:experiment-market-size}
    {\footnotesize\begin{tabular}{lcccccc}
\toprule
\textbf{\makecell{Traders\\per side}} & \textbf{\makecell{Normal\\SD}} & \textbf{\makecell{Optimal\\GFT}} & \textbf{\ConservativeSegmentedPricing} & \parbox[c]{2.4cm}{\centering\bfseries\MSSPM} & \textbf{\makecell{Scouting\\Bandits}} & \textbf{\makecell{Scouting\\Blindits}} \\
\midrule
100 & 0.2 & 7.831 & 7.802 & 5.899 & 3.309 & 5.385 \\
100 & 0.5 & 11.364 & 11.343 & 8.661 & 6.141 & 8.383 \\
\midrule
500 & 0.2 & 40.113 & 39.855 & 30.229 & 16.901 & 27.844 \\
500 & 0.5 & 57.662 & 57.461 & 43.210 & 31.269 & 42.989 \\
\bottomrule
\end{tabular}
 }
    \begin{tablenotes}
        \item \emph{\footnotesize\underline{Note}:
        GFT is reported in millions and, for each instance configuration, summed over its 500 instances and 1000 rounds per instance. Optimal GFT is cumulative first-best GFT; algorithm columns report cumulative expected GFT. Normal SD is the standard deviation of the normal distribution before truncation to $[0,1]$. The scouting algorithms use pooled feedback.}
    \end{tablenotes}
\end{table}

Three findings emerge. First, {\ConservativeSegmentedPricing} captures nearly all of the first-best GFT in every instance configuration, more than any of the other three algorithms. Second, the two bilateral-trade algorithms, applied with pooled feedback, capture far less. Third, the objective of an algorithm determines who receives the gains: {\ConservativeSegmentedPricing} leaves almost all of the GFT to traders, whereas {\MSSPM} keeps nearly half of first-best GFT as platform profit and forgoes about a quarter.

\xhdr{Gains from trade.}
\Cref{tab:experiment-market-size} reports cumulative GFT in each instance configuration. {\ConservativeSegmentedPricing} captures 99.36--99.81\% of first-best GFT, compared with 74.94--76.21\% for {\MSSPM}, 68.76--74.55\% for Scouting Blindits, and 42.13--54.23\% for Scouting Bandits. No standard error of these shares exceeds 0.54 percentage points, so the differences are large relative to sampling error. The ranking of {\ConservativeSegmentedPricing} also holds instance by instance: in each of the 2000 instances, {\ConservativeSegmentedPricing} attains more cumulative expected GFT than each of the other three algorithms.

Market size and dispersion change little for {\ConservativeSegmentedPricing}, which forgoes at most 0.64\% of first-best GFT in any instance configuration. Higher dispersion raises the GFT available, from 40.113 to 57.662 million with 500 traders per side, and raises the shares of first-best GFT captured by the two scouting algorithms by 5 to 12 percentage points.

The two scouting algorithms end with a common price from a fixed grid. For either grid and in every instance configuration, at most 3.6\% of the instances have a grid price that attains first-best GFT, so after scouting, these algorithms lose GFT in almost every instance. Their losses also include the scouting rounds and the difference between the pooled random-pair objective and the GFT of the market; the experiment does not separate these sources.

{\MSSPM} falls short for a different reason: maximizing profit restricts trade. Even with known types, the profit-maximizing seller-side and buyer-side prices attain only 75.15--76.42\% of first-best GFT in the four instance configurations, and {\MSSPM} attains 74.94--76.21\%. Nearly all of this shortfall therefore stems from the profit objective rather than from learning.

\begin{figure}
    \centering
    \includegraphics[width=\linewidth]{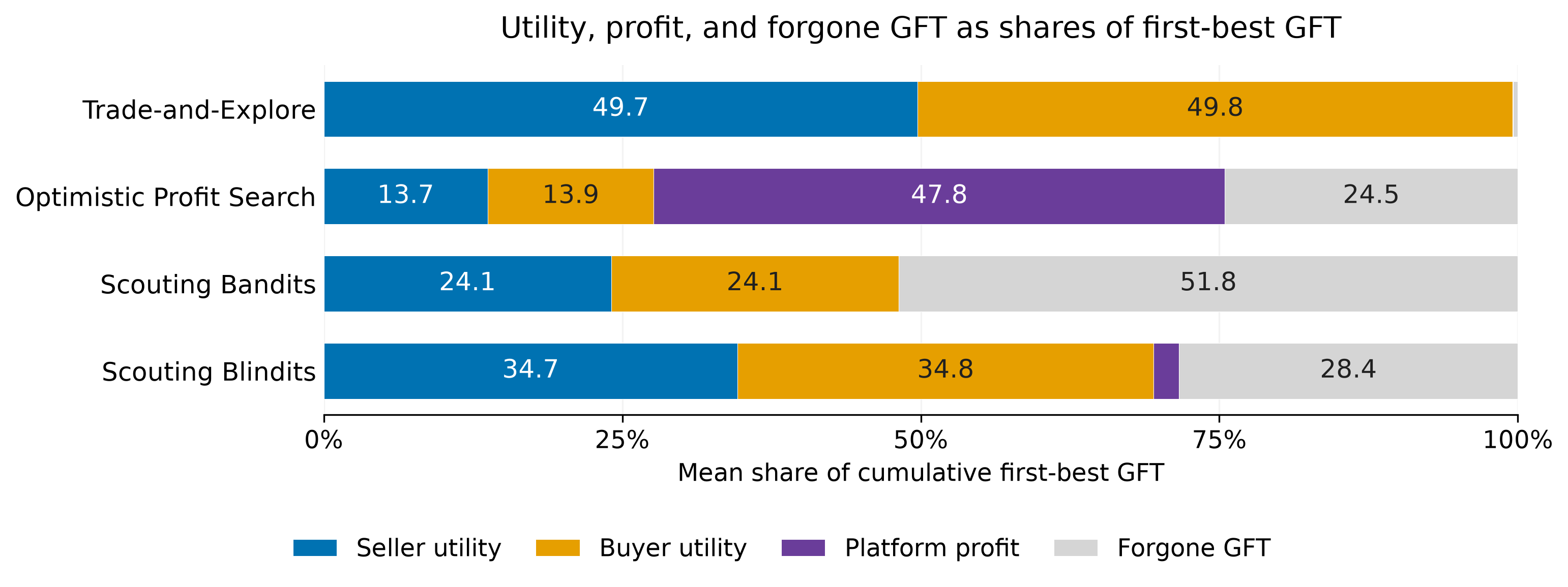}
    \caption{Division of first-best GFT into seller utility, buyer utility, platform profit, and forgone GFT under each algorithm. In each of the four instance configurations, each component is summed over the 500 instances and 1000 rounds and divided by cumulative first-best GFT; each bar segment shows the unweighted average of this share over the four instance configurations. The four shares in each bar sum to 100\%. The scouting algorithms use pooled feedback.}
    \label{fig:experiment-surplus-500}
\end{figure}

\xhdr{Who receives the gains.}
\Cref{fig:experiment-surplus-500} shows the average division of first-best GFT among sellers, buyers, the platform, and forgone GFT. {\ConservativeSegmentedPricing} leaves 49.74\% of first-best GFT to sellers and 49.82\% to buyers, keeps 0.05\% as platform profit, and forgoes 0.39\%. {\MSSPM} leaves 13.74\% to sellers and 13.90\% to buyers, keeps 47.82\% as platform profit, and forgoes 24.54\%. With known types, the profit-maximizing prices would keep 48.19\% of first-best GFT as profit, so {\MSSPM} collects nearly all of the available profit. Scouting Bandits posts only common prices, so the platform earns no profit, and 51.84\% of first-best GFT is forgone. Under Scouting Blindits, the platform earns 2.12\% during scouting, when seller-side and buyer-side prices differ, and 28.38\% is forgone.

\endgroup

\end{document}